\documentclass[%
 reprint,
 superscriptaddress,
 showpacs,
 amsmath,amssymb,
 aps,
 prl,
]{revtex4-2}

\usepackage{mathtools}
\usepackage{xcolor}
\usepackage{amsthm}
\usepackage{mathrsfs}
\usepackage{enumitem}
\usepackage{booktabs}
\theoremstyle{definition}
\newtheorem{proposition}{Proposition}
\newtheorem{remark}{Remark}
\newtheorem{observation}{Observation}
\newtheorem{lemma}{Lemma}

\usepackage{graphicx}
\usepackage{dcolumn}
\usepackage{bm}
\usepackage{hyperref}

\newcounter{mainresult}
\renewcommand{\themainresult}{\arabic{mainresult}}

\begin{document}

\preprint{APS/123-QED}

\title{Size-Independent Robustness in Multipartite Bell Self-Testing}

\author{Shen Cao}
\affiliation{
  QICI Quantum Information and Computation Initiative, School of Computing and Data Science,
  The University of Hong Kong, Pokfulam Road, Hong Kong SAR, China
}
\author{Xingjian Zhang}
\affiliation{
  QICI Quantum Information and Computation Initiative, School of Computing and Data Science,
  The University of Hong Kong, Pokfulam Road, Hong Kong SAR, China
}
\affiliation{
  Centre for Quantum Software and Information, University of Technology Sydney, Sydney, Australia
}
\author{Fei Shi}
\affiliation{
  QICI Quantum Information and Computation Initiative, School of Computing and Data Science,
  The University of Hong Kong, Pokfulam Road, Hong Kong SAR, China
}
\affiliation{Institute of Quantum Computing and Software, School of Computer Science and Engineering, Sun Yat-sen University, Guangzhou 510006, China}
\author{Qi Zhao}
 \email{{zhaoq@hku.hk}}
\affiliation{
  QICI Quantum Information and Computation Initiative, School of Computing and Data Science,
  The University of Hong Kong, Pokfulam Road, Hong Kong SAR, China
}

\date{\today}

\begin{abstract}
  Practical robust self-testing of multipartite entanglement has so far been restricted to small-scale systems due to error bounds that degrade severely with system size.
  In this work, we establish multipartite self-testing with robustness independent of the size of the quantum network.
  We derive a fully analytic, device-independent self-testing bound for $n$-qubit Greenberger-Horne-Zeilinger (GHZ) states.
  The bound scales linearly with the observed violation error and lies universally within a constant factor of two from a theoretical upper bound.
  Furthermore, the operator-inequality framework reduces the verification of the conjectured optimal bound to a highly efficient numerical check, which we perform up to $n=100$.
  Consequently, GHZ entanglement can be certified under a fixed noise level in arbitrarily large systems, enabling scalable device-independent verification.
\end{abstract}

\maketitle


\paragraph{Introduction.}
Multipartite entangled states, such as $n$-qubit Greenberger-Horne-Zeilinger (GHZ) states \cite{Greenberger_Horne_Zeilinger_1989,Greenberger_Horne_Shimony_Zeilinger_1990}, are fundamental workhorses of large-scale quantum networks \cite{McCutcheon_Pappa_Bell_et_al_2016,Bugalho_Coutinho_Monteiro_Omar_2023,Azuma_Economou_Elkouss_Hilaire_Jiang_Lo_Tzitrin_2023}.
Verifying the quality of these distributed resources is paramount \cite{Friis_Vitagliano_Malik_Huber_2019}, yet experimental devices are inherently noisy and often untrusted \cite{shor1996fault,Mayers_Yao_1998}.
Device-independent self-testing provides the most stringent form of certification, uniquely determining the entangled quantum systems based solely on classical input-output statistics \cite{Supic_Bowles_2020,Eisert_Hangleiter_Walk_Roth_Markham_Parekh_Chabaud_Kashefi_2020}.
In practice, experimental statistics rarely match the ideal perfectly; thus, self-testing is of practical interest only if it is robust \cite{zhang2019experimental,Wu_et_al_2021,Wu_et_al_2022}.
Specifically, a deficit $\varepsilon > 0$ of the violation of Bell inequality must imply a quantitative bound of the distance to the ideal state.
The central unresolved challenge for multipartite systems is how this robust bound scales with the number of parties $n$.

For bipartite systems, this method is mature.
The Clauser-Horne-Shimony-Holt (CHSH) inequality self-tests the singlet \cite{Clauser_Horne_Shimony_Holt_1969,Summers_Werner_1987,Mayers_Yao_1998,Mayers_Yao_2003}, with analytic bounds ranging from $O(\sqrt{\varepsilon})$ for SWAP-based methods \cite{McKague_Yang_Scarani_2012,Yang_Navascues_2013,Coladangelo_Goh_Scarani_2017} to Kaniewski's optimal linear bound \cite{Kaniewski_2016}.
For $n$-qubit GHZ states, the canonical multipartite generalization of CHSH is the MABK inequality \cite{Ardehali_1992,Belinskii_Klyshko_1993}.
Kaniewski's operator-inequality method gives the optimal linear bound for $n=3$ \cite{Kaniewski_2016}, recently extended to $n=4,5$ \cite{Jha_Singh_Pan_2026}.
The obstacle for quantum networks is therefore not exact self-testing but scaling.
Existing analytic bounds valid for arbitrary $n$ tolerate only violation errors that shrink polynomially \cite{McKague_2011,McKague_2016} or exponentially \cite{Singh_Sasmal_Pan_2025} with $n$.
For example, for $n \ge 10$, even if the observed violation misses the maximal quantum value by only a relative deficiency of $10^{-5}$, we cannot derive any nontrivial result.
At the same scale, numerical methods such as the Navascu\'{e}s-Pironio-Ac\'{i}n hierarchy or searches over measurement angles also become intractable for the exponential blow-up in computational complexity \cite{Navascues_Pironio_Acin_2008,Baccari_Cavalcanti_Wittek_Acin_2017,Baccari_Augusiak_Supic_Tura_Acin_2020,Zhao_Zhou_2022}.
Certification therefore becomes hypersensitive to microscopic noise precisely as networks grow.


In this Letter, we remove this scaling bottleneck by establishing size-independent robustness in self-testing of the multipartite GHZ state.
Adapting the operator inequality framework to the MABK inequality, we prove an analytic, device-independent robust self-testing bound of $n$-qubit GHZ states for all $n \ge 3$.
Crucially, this bound scales linearly with the violation error, with a slope twice that of an explicit linear upper bound known to be optimal for $n \le 5$ \cite{Kaniewski_2016,Jha_Singh_Pan_2026}.
Furthermore, we support the universal optimality of the latter bound up to $n=100$ by reducing its verification from an exponentially hard problem to a polynomial-size positivity check.
As a consequence, GHZ states can be certified under a fixed level of noise for any number of parties.
We further demonstrate its potential as a primitive for large-scale quantum networks by deriving the first robust, fully device-independent lower bound on the randomness certified from the $n$-qubit GHZ state for all $n \ge 3$, complete with finite-sampling guarantees.

\paragraph{Extractability.}
Consider an $n$-partite quantum system $\bigotimes_{k=1}^n \mathcal{H}_k$, where a multipartite entangled state $\rho$ is distributed among $n$ remote parties.
The parties have access only to local operations, and the quantum devices (state source and measurement apparatus) are untrusted.
The sole information about $\rho$ is the statistical distribution of measurement outcomes.
In such a scenario, exact identification of the state is impossible because local unitary rotations on the measurements and the state leave the outcome distribution unchanged.
Therefore, self-testing can at most certify the state up to local equivalence.
This leads to the notion of \emph{extractability} \cite{Bardyn_2009}.
For a target pure state $|\Psi\rangle \langle \Psi|$ and an $n$-partite input state $\rho$, the extractability of $\Psi$ from $\rho$ is defined as
\begin{equation*}
  \Xi \left( \rho \to \Psi \right) \coloneqq \max \left\{ F \left( \Lambda(\rho), \Psi \right) \middle| \Lambda = \bigotimes_{k=1}^n \Lambda_k \right\}.
\end{equation*}
where the maximum is taken over all local channels $\Lambda$, $F(\rho, \sigma) = \|\sqrt{\rho}\sqrt{\sigma}\|_1^2$ denotes the quantum fidelity, and $\|\cdot\|_1$ is the trace norm.

\paragraph{Result}\hspace{-0.2cm}\refstepcounter{mainresult}\themainresult\label{result:main}.
  Let $s_n = \left(\sqrt{2} + 1\right)/2^\frac{n}{2}$ and $\mu = - \left( 1 + \sqrt{2} \right)$.
  Then for any $n$-partite quantum state $\rho$ attaining violation value $\beta$ on the $n$-qubit MABK inequality \cite{Ardehali_1992,Belinskii_Klyshko_1993} with maximal quantum violation $\beta_{nQ}=2^{\frac{n+1}{2}}$, we have
  \begin{equation}
    \label{eq:linear-robustness}
    \Xi(\rho \to |\text{GHZ}_n\rangle \langle\text{GHZ}_n|) \geq s_n \beta + \mu,
  \end{equation}
  where 
  \begin{equation*}
    |\text{GHZ}_n\rangle \coloneqq \frac{1}{\sqrt{2}} \left( |0\rangle^{\otimes n} + |1\rangle^{\otimes n} \right).
  \end{equation*}
We note that this lower bound differs from a linear upper bound only by a constant factor of 2 in the slope, as shown below.

\paragraph{Optimality.}
Let $\beta_{nC} = 2^{\frac{n}{2}}$ be the maximal violation value of $n$-qubit MABK inequality achieved with a biseparable state.
Suppose we observe a violation value $\beta$.
To derive an upper bound, write $\beta = p \beta_{nQ} + (1-p) \beta_{nC}$ with $p \in [0, 1]$, and consider the state
\begin{equation*}
  \sigma \coloneqq p |\text{GHZ}_n\rangle \langle \text{GHZ}_n| + (1-p) \sigma_{\text{sep}},
\end{equation*}
where $\sigma_{\text{sep}}$ is a biseparable state attaining the biseparable bound $\beta_{nC}$, so that $\sigma$ attains the violation $\beta$.
Then, since extractability is convex in the input state, 
\begin{equation}
  \label{eq:1}
  \Xi \left( \sigma \to |\text{GHZ}_n\rangle \langle\text{GHZ}_n| \right) \leq p + (1 - p) \cdot 0.5.
\end{equation}
Here $\Xi \left( \sigma_{\text{sep}} \to |\text{GHZ}_n\rangle \langle\text{GHZ}_n| \right) \leq 0.5$ follows from the fact that the maximal Schmidt coefficient of $|\text{GHZ}_n\rangle$ is $1/\sqrt{2}$, the square of which upper bounds the fidelity between $|\text{GHZ}_n\rangle$ and any biseparable state \cite{Shimony_1995}.
Substituting $p = \frac{\beta - \beta_{nC}}{\beta_{nQ} - \beta_{nC}}$ into Eq.~\eqref{eq:1} yields
\begin{equation}
  \label{eq:upper-bound}
  \begin{aligned}
    &\Xi(\sigma \to |\text{GHZ}_n\rangle \langle\text{GHZ}_n|) \\
    \leq& \frac{1}{2} + \frac{1}{2}\frac{\beta - \beta_{nC}}{\beta_{nQ} - \beta_{nC}} = \frac{1 + \sqrt{2}}{2^{\frac{n}{2} + 1}} \beta - \frac{1}{\sqrt{2}}.
  \end{aligned}
\end{equation}
For $n \le 5$, this bound is tight: it recovers Kaniewski's optimal bound for the Mermin inequality ($n=3$) \cite{Kaniewski_2016} and the optimal bounds established analytically for $n = 4, 5$ in \cite{Jha_Singh_Pan_2026}.
Define 
\begin{equation}
  \label{eq:upper-bound-params}
  \bar{s}_n \coloneqq \frac{1 + \sqrt{2}}{2^{\frac{n}{2} + 1}},\quad \bar{\mu} \coloneqq - \frac{1}{\sqrt{2}},
\end{equation}
which are slope and intercept of the upper bound in Eq.~\eqref{eq:upper-bound}, respectively.
Since $s_n = 2 \bar{s}_n$, the gap to optimality remains constant, independent of $n$ (Fig.~\ref{fig:lower_bound_qubits}).

\begin{figure}[htbp]
  \begin{minipage}{0.49\textwidth}
    \includegraphics[width=\textwidth]{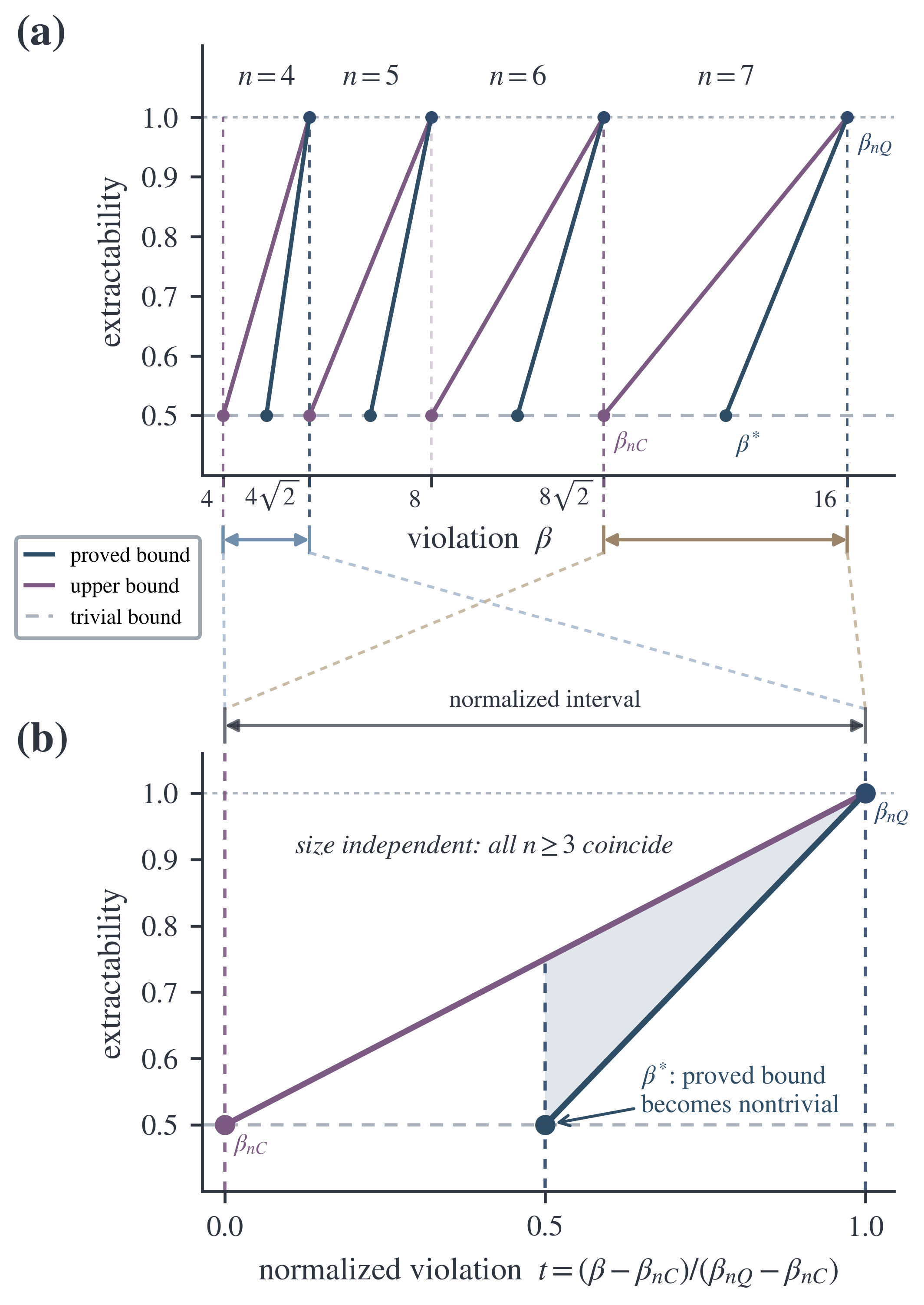}
  \end{minipage}
  \caption{\label{fig:lower_bound_qubits}Size-independent robust self-testing. (a) The proved bound and the upper bound of Eq.~\eqref{eq:upper-bound} versus the raw violation $\beta$ for the windows $[2^{n/2}, 2^{(n+1)/2}]$, $n=4,\ldots,7$: $\beta_{nC}$ and $\beta_{nQ}$ are the window boundaries, and $\beta^{*}=(\beta_{nC}+\beta_{nQ})/2$ is the onset of nontriviality. (b) Under the normalization $t=(\beta-\beta_{nC})/(\beta_{nQ}-\beta_{nC})$, for each $n$ the window collapses onto the two universal lines $\Xi=t$ and $\Xi=(1+t)/2$, i.e.\ the robustness is independent of the number of parties.}
\end{figure}

\paragraph{Proof Sketch.}
Following~\cite{Kaniewski_2016}, we convert the extractability bound of Eq.~\eqref{eq:linear-robustness} into an operator inequality
\begin{equation*}
  \Lambda^\dagger(|\text{GHZ}_n\rangle \langle\text{GHZ}_n|) \geq s_n W_n + \mu I,
\end{equation*}
where $\Lambda$ is the local extraction channel and $W_n$ is the Bell operator of $n$-qubit MABK inequality.
Proving this inequality for measurements of arbitrary dimension is the central technical challenge.
We address it in two stages.
First, Naimark dilation~\cite{Neumark1943} is applied, followed by block-diagonalization of qubit observables by Jordan's Lemma \cite{Aharon}.
The problem is thus effectively restricted to $n$-qubit systems.
Second, we project the simplified inequality onto $2^{n-1}$ orthogonal 2-dimensional subspaces using $n-1$ commuting Pauli strings.
On each subspace, the positivity condition reduces to a scalar inequality
involving a smooth function of the $n$ observable angles, which we verify
analytically for the parameters in Result~\ref{result:main}.

For the optimal parameters of Eq.~\eqref{eq:upper-bound-params},
the analytic proof succeeds on all but two families of subspaces,
reducing the remaining verification to a tractable numerical problem.

\paragraph{Result}\hspace{-0.2cm}\refstepcounter{mainresult}\themainresult\label{result:numeric}.
For each $n \ge 3$, a sufficient condition for Eq.~\eqref{eq:upper-bound} to be tight for every $\beta \ge \beta_{nC}$ (that is, for Eq.~\eqref{eq:linear-robustness} to hold with the optimal parameters of Eq.~\eqref{eq:upper-bound-params}) is the nonnegativity of two smooth functions defined respectively over the ordered simplexes
\begin{equation*}
  \begin{aligned}
    [0, 1] \times& \{0 \leq x_1 \leq \cdots \leq x_{n-1} \leq 1\}, \\
    \{0 \leq x_1 \leq x_2 \leq 1\} \times& \{0 \leq x_1 \leq \cdots \leq x_{n-2} \leq 1\},
  \end{aligned}
\end{equation*}
The ordering of the simplexes makes the verification tractable, reducing the number of lattice points required to check the condition by grid search from $\propto m^n$ to $\propto n^m$ (a polynomial in $n$ for fixed $m$), where $m$ is the number of grid points per dimension.
This conversion of an exponential verification problem into a polynomial one bypasses the exponential blow-up that confines numerical certification methods to few-party systems: SDP hierarchies grow exponentially with the number of parties \cite{Navascues_Singh_Acin_2020}, and numerical evaluations of operator-inequality bounds have reached only $n \le 7$ \cite{Baccari_Augusiak_Supic_Tura_Acin_2020}.
We conducted a grid search and performed gradient descent optimization for all grid points.
Under double precision, with fine grids ($m \ge 7$) up to $n = 40$ qubits and coarser grids (down to $m = 5$) up to $n = 100$, the numerical results support robust self-testing in Eq.~\eqref{eq:linear-robustness} with the optimal parameters in Eq.~\eqref{eq:upper-bound-params}, extending far beyond the party numbers covered by existing analytic proofs.
Details are provided in Supplemental Material~\cite{SM_key}.

\paragraph{Applications.}
Size-independent robustness propagates to downstream applications.
The following quantitative results are derived from the sub-optimal bound of Result~\ref{result:main}; the optimal parameters of Eq.~\eqref{eq:upper-bound-params} would yield correspondingly stronger guarantees.
Suppose we want to certify $|\text{GHZ}_n\rangle$ under white noise, that is, the input state $\rho$ can be written
\begin{equation*}
  \rho = p|\text{GHZ}_n\rangle\langle\text{GHZ}_n| + (1 - p) \frac{I}{2^n}.
\end{equation*}
We have $F(|\text{GHZ}_n\rangle\langle\text{GHZ}_n|, \rho) = p$ and the violation
\begin{equation*}
  \text{tr}[W\rho] = p \beta_{nQ} + (1-p) \cdot 0 = 2^{\frac{n+1}{2}}p.
\end{equation*}
The extractability certified by Result~\ref{result:main} and the upper bound derived from Eq.~\eqref{eq:upper-bound} are
\begin{equation*}
  \begin{aligned}
    &\Xi(\rho \to |\text{GHZ}_n\rangle\langle\text{GHZ}_n|) \\
    & \begin{cases}
      \geq s_n \beta + \mu = \left( 1 + \sqrt{2} \right) \left( \sqrt{2}p - 1 \right), & \text{Result~\ref{result:main}},\\
      \leq \overline{s}_n \beta + \overline{\mu} = \frac{2 + \sqrt{2}}{2}p - \frac{1}{\sqrt{2}}, & \text{Upper bound}.
    \end{cases}
  \end{aligned}
\end{equation*}
A visual comparison of these bounds is given in Fig.~\ref{fig:robust_noise}.

\begin{figure}
  \begin{minipage}{0.49\textwidth}
    \includegraphics[width=\textwidth]{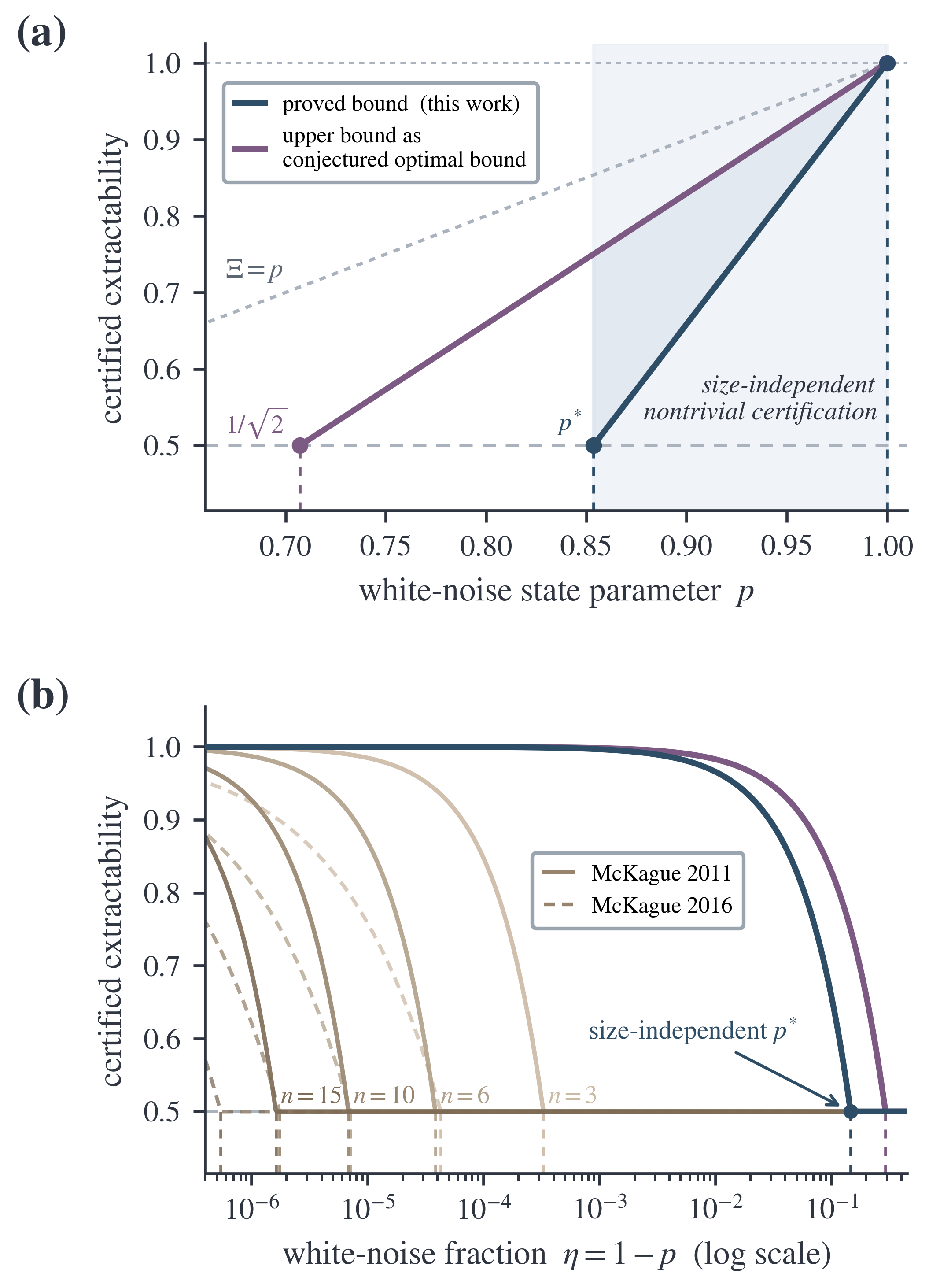}
  \end{minipage}
  \caption{\label{fig:robust_noise}White-noise robustness. (a) Certified extractability for $\rho_p = p\,|\text{GHZ}_n\rangle\langle\text{GHZ}_n|+(1-p)I/2^n$: the proved bound becomes nontrivial at $p^{*}=(2+\sqrt{2})/4$, while the upper bound of Eq.~\eqref{eq:upper-bound} crosses the trivial threshold $\Xi=1/2$ at $p=1/\sqrt{2}$; both curves are independent of $n$. (b) The same bounds versus the deficiency $\eta=1-p$ together with McKague's complete-graph tests ($n=3,6,10,15$), whose tolerance deteriorates with $n$; the vertical lines mark the crossings of $\Xi=1/2$.}
\end{figure}

For comparison, the previously known analytic robust bounds valid for arbitrary $n$ require a violation error that shrinks rapidly with $n$ \cite{McKague_2011,McKague_2016,Singh_Sasmal_Pan_2025}, as illustrated in Fig.~\ref{fig:tolerance_bound_qubits}. A formal derivation of the two McKague bounds displayed there, including the identification of the deviation parameters, is given in Section~\ref{sec:comparison} of the Supplemental Material~\cite{SM_key}.

\begin{figure}[htbp]
  \begin{minipage}{0.49\textwidth}
    \includegraphics[width=\textwidth]{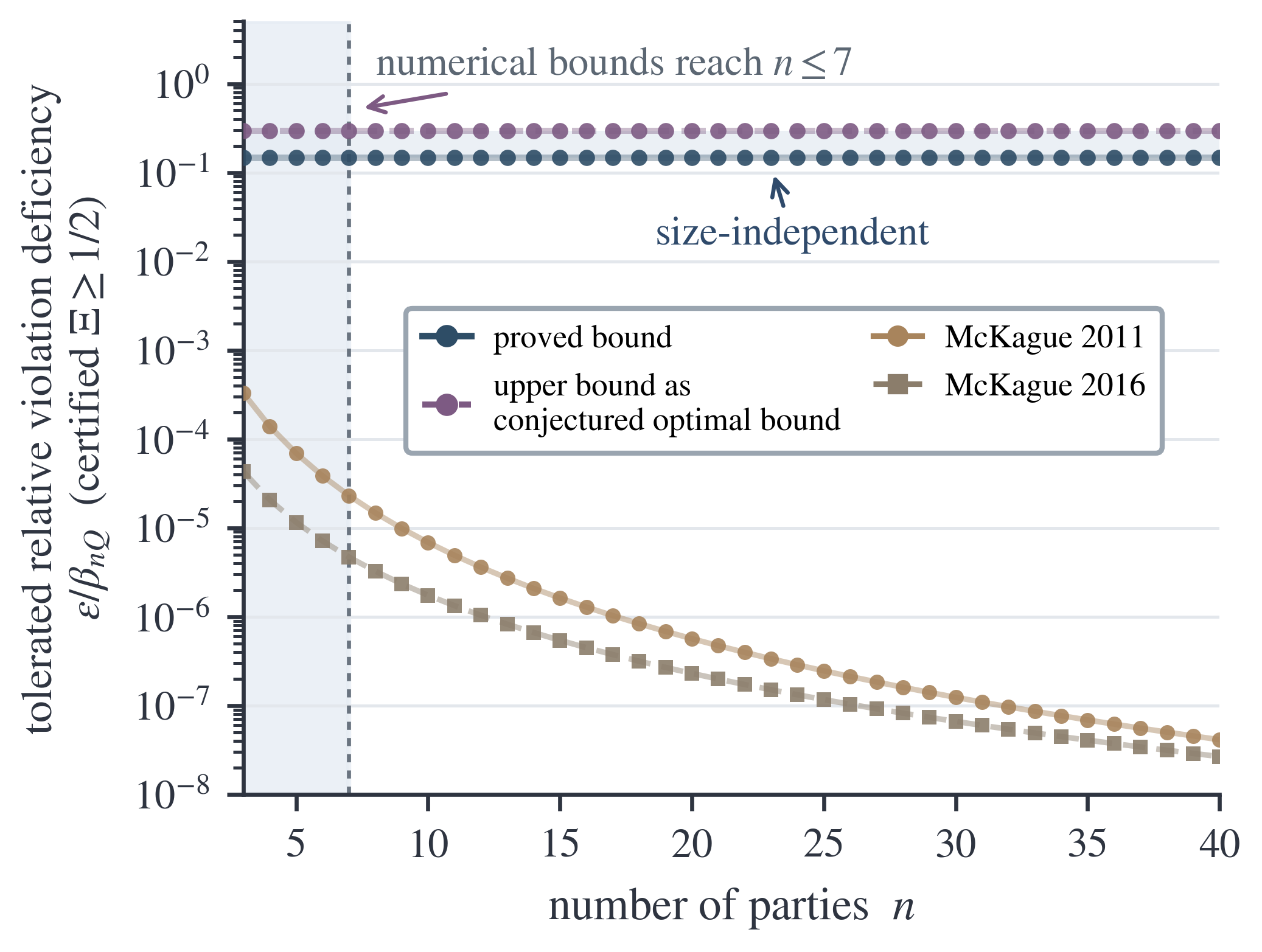}
  \end{minipage}
  \caption{\label{fig:tolerance_bound_qubits}Tolerance versus number of parties: largest relative violation deficiency $\eta=\varepsilon/\beta_{nQ}$ certifying $\Xi\geq 1/2$. The proved bound and the upper bound of Eq.~\eqref{eq:upper-bound} are flat in $n$ (a constant factor-of-two gap), whereas the previously known analytic tests for all $n$ (McKague) decay with the number of parties; the shaded band at $n\leq 7$ marks the reach of numerical operator-inequality bounds.}
\end{figure}

Beyond entanglement certification, size-independent robustness and the analytic nature of our bound also enable device-independent randomness analysis in large-scale quantum networks.
As a second application, we derive a device-independent lower bound on the conditional entropy certified by the $n$-qubit MABK inequality in a device-independent randomness generation (DIRG) protocol~\cite{colbeck3814quantum,Pironio_Acin_Massar_de_la_Giroday_Matsukevich_Maunz_Olmschenk_Hayes_Luo_Manning_et_al_2010,Colbeck_Kent_2011,Wooltorton_Brown_Colbeck_2025}.
At maximal violation $\beta_{nQ} = 2^{(n+1)/2}$, the protocol certifies $H_{\rm ideal}(n)$ bits of private randomness, where $H_{\rm ideal}(n) = n$ for odd $n$ (on input $\mathbf{0}$ for $n \equiv 3 \pmod 4$ and on input $\mathbf{1}$ for $n \equiv 1 \pmod 4$; input $\mathbf{0}$ alone certifies only $n-1$ bits when $n \equiv 1 \pmod 4$) and $H_{\rm ideal}(n) = n + \frac{1}{2} - \frac{\log_2(1+\sqrt{2})}{\sqrt{2}} \approx n - 0.4$ for even $n$~\cite{Wooltorton_Brown_Colbeck_2025}.
The following result bounds the randomness deficit when the observed violation falls short of the maximum.

\paragraph{Result}\hspace{-0.2cm}\refstepcounter{mainresult}\themainresult\label{result:dirg}.
  For $n \geq 3$, let $\beta = \beta_{nQ} - \varepsilon$ be the observed MABK violation
  with $\varepsilon \in [0,\, 2^{n/2}(\sqrt{2}-1)]$.
  Define $\eta \coloneqq \varepsilon\cdot 2^{-(n+1)/2}$,
  $s_n \coloneqq (1+\sqrt{2})/2^{n/2}$, and
  $w \coloneqq 4\eta - 2\eta^{2}$.
  The conditional entropy of the measurement outcomes given Eve's side information\footnote{The simplified bound is derived from the conditional-entropy continuity bound of Winter~\cite{Winter_2016} under the assumption $\delta_{\rm total} \coloneqq 2\sqrt{s_n\varepsilon} + n\sqrt{2w} + \frac{\sqrt{2}}{2}\,n\,w \leq 1$; for parameter regimes where this condition is violated, the bound $H(R|E) \geq H_{\rm ideal}(n) - n$ holds instead via the trivial classical-quantum entropy bound (valid because the post-measurement states are classical-quantum), see Eq.~\eqref{sm:dirg:complete_explicit} of the Supplemental Material.} satisfies
  \begin{equation*}
    \begin{aligned}
      H(R|E)(\beta) \geq H_{\rm ideal}(n) &- n\sqrt{s_n\varepsilon}
      - n^{2}\sqrt{w/2} - \frac{n^{2}w}{2\sqrt{2}} \\
      &- 2\,h\!\left(\sqrt{s_n\varepsilon} + n\sqrt{w/2} + \frac{n w}{2\sqrt{2}}\right),
    \end{aligned}
  \end{equation*}
  where $h(\cdot)$ is the binary entropy.
  $H_{\rm ideal}(n) = n$ for odd $n$ (input $\mathbf{0}$ for $n \equiv 3 \pmod 4$, input $\mathbf{1}$ for $n \equiv 1 \pmod 4$);
  $H_{\rm ideal}(n) = n + \tfrac{1}{2} - \log_2(1+\sqrt{2})/\sqrt{2} \approx n - 0.4$ for even $n$.
  Figure~\ref{fig:dirg} shows this certified rate as a function of the relative violation deficiency $\varepsilon/\beta_{nQ}$ for several numbers of parties.
  For $n = 11, 15$, we can still certify nontrivial randomness with relative violation deficiency $\eta \approx 10^{-4}$.
  While under the same noise level the analytic robustness bounds in Refs.~\cite{McKague_2011,McKague_2016} cannot certify any entanglement for $n \ge 10$ (shaded region in Fig.~\ref{fig:dirg}).

\begin{figure}[htbp]
  \centering
  \includegraphics[width=0.45\textwidth]{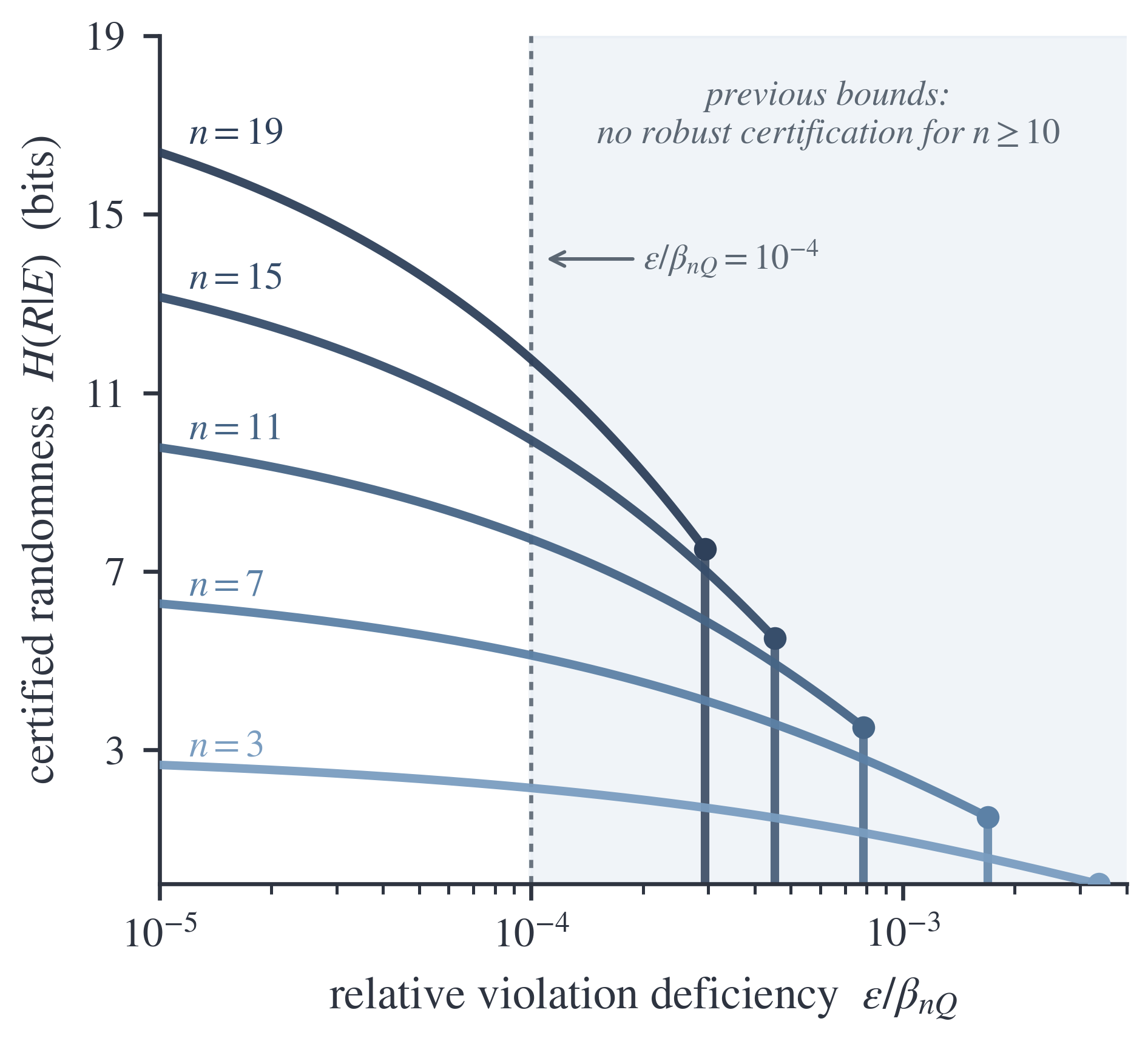}
  \caption{\label{fig:dirg}Certified randomness versus the relative violation deficiency $\varepsilon/\beta_{nQ}$ for $n=3,7,11,15,19$ (all $n\equiv3 \pmod 4$, for which $H_{\mathrm{ideal}}(n)=n$ on input $\mathbf{0}$): each curve approaches its asymptotic limit, the $H_{\mathrm{ideal}}(n)=n$ ticks on the vertical axis, and terminates at its own nontriviality threshold, marked by the solid vertical segments. At the reference noise $\varepsilon/\beta_{nQ}=10^{-4}$ (dashed line), the certified randomness ranges from $\approx 2$ bits ($n=3$) to $\approx 12$ bits ($n=19$).}
\end{figure}

\paragraph{Proof Sketch.}
The state-imperfection contribution follows from the extractability bound
of Result~\ref{result:main} in combination with the Fuchs--van~de~Graaf
inequality~\cite{Fuchs_van_de_Graaf_1999} and Winter's tight conditional-entropy
continuity bound~\cite{Winter_2016,Alicki_Fannes_2004}.
The measurement-imperfection term is obtained by combining
commutator-based self-testing~\cite{Kaniewski_2017} with a diamond-norm
analysis of the $n$-party measurement channel.
A measurement-translation residual, arising from the non-commutation of the
physical measurement with the dephasing part of the extraction channel, is
bounded by an elementary one-variable inequality
(Lemma~\ref{lem:sm:dirg:translation} of the Supplemental Material).
See Supplemental Material~\cite{SM_key} for the full derivation.

A finite-sampling analysis determines the measurement budget needed to certify these bounds with prescribed statistical confidence.
Combining Hoeffding's concentration inequality with the linear bound of Result~\ref{result:main}, the required number of experimental rounds scales as $N \propto 2^{n}$ (up to logarithmic factors), reflecting the $2^{n}$ measurement settings of the MABK inequality.

\paragraph{Discussion.}

In this Letter, we have established a size-independent, linearly robust self-testing framework for $n$-qubit GHZ states.
By advancing beyond per-party analytical constructions \cite{Kaniewski_2016,Jha_Singh_Pan_2026} to an $n$-uniform operator inequality reduction, our work provides the first fully analytical bound that does not degrade as the quantum network scales.

Beyond the immediate analytical results, our approach resolves a computational bottleneck in multipartite nonlocality.
The explicit reduction of the optimality verification from an exponentially hard grid search ($m^n$) to a tractable polynomial space ($n^m$) demonstrates how analytical symmetries can bypass the prohibitive costs that currently paralyze standard numerical hierarchies for multipartite systems \cite{Navascues_Singh_Acin_2020,Kempe_Kobayashi_Matsumoto_Toner_Vidick_2011}.
This paradigm shift is essential for analyzing macroscopic quantum correlations where conventional numerical methods fail.

Physically, the analytical and size-independent nature of our bounds provides the missing theoretical primitive required for scalable device-independent protocols.
As demonstrated, the bound provides device-independent certification of GHZ states under white noise for an arbitrary number of parties, whereas previously available analytical and numerical methods would require data of polynomially or exponentially increasing accuracy as the network grows \cite{McKague_2011,McKague_2016,Singh_Sasmal_Pan_2025}.
Further, combined with observable self-testing \cite{Kaniewski_2017}, the analytical bound yields a fully device-independent lower bound on the conditional entropy in DIRG with MABK inequality.
We also anticipate possible use in a broader spectrum of multiparty tasks, ranging from secure cryptography (e.g. device-independent quantum secret sharing (DIQSS) \cite{Hillery_Buzek_Berthiaume_1999,Cleve_Gottesman_Lo_1999,Roy_Mukhopadhyay_2019,Moreno_Brito_Nery_Chaves_2020} and anonymous communication~\cite{Broadbent_Tapp_2007,Unnikrishnan_MacFarlane_Yi_Diamanti_Markham_Kerenidis_2019,Thalacker_Hahn_de_Jong_Pappa_Barz_2021}) to distributed quantum sensing~\cite{Eldredge_Foss-Feig_Gross_Rolston_Gorshkov_2018,Huang_Macchiavello_Maccone_2019,Ho_Webb_Brooks_Grasselli_Gauger_Fedrizzi_2026}, where certified GHZ states constitute a critical resource.

Looking forward, the strong numerical evidence up to $n=100$ invites a complete analytical proof for the optimal parameters for all $n \ge 6$.
To further enhance practical deployment, pruning the MABK inequality to reduce sample complexity while preserving robustness remains a crucial next step.
Finally, tightening the device-independent randomness bound derived from our robustness estimate in the DIRG application is likewise a meaningful direction for future investigation.

\begin{acknowledgments}
S.C and Q.Z. acknowledges funding from Quantum Science and Technology-National Science and Technology Major Project 2024ZD0301900, National Natural Science Foundation of China (NSFC) via Project No. 12347104 and No. 12305030, Hong Kong Research Grant Council (RGC) via No. 27300823, 17310926, N\_HKU718/23, and R6010-23.
X.Z. acknowledges the Chancellor's Research Fellowship program of the University of Technology Sydney.
F. S. acknowledges funding from the National Natural Science Foundation of China (Grant No. 62601907 and No. 12571493), and the 2026 Basic Start-up Fund of Sun Yat-sen University (Grant No. 67000-12266020).
\end{acknowledgments}

\bibliography{bibliography_aps}

\newpage
\clearpage
\onecolumngrid

\makeatletter
\setcounter{secnumdepth}{3}
\def\@seccntformat#1{\csname the#1\endcsname.\quad}
\@addtoreset{section}{part}
\makeatother
\renewcommand{\thesection}{\Alph{section}}
\renewcommand{\thesubsection}{\arabic{subsection}}
\makeatletter
\@addtoreset{subsection}{section}
\makeatother

\begin{center}
  \textbf{\large Supplemental Material}\\[.2cm]
\end{center}
\setcounter{equation}{0}
\renewcommand{\theHequation}{S\arabic{equation}}

\renewcommand{\theequation}{S\arabic{equation}}

\section{Proof of Result~\protect\ref{result:main}}
\label{sec:proof-main}

\subsection{Robust Self-Testing from Operator Inequality}

Suppose a Bell inequality $\mathcal{B}$ self-tests the pure state $\Psi$ when its quantum maximum $\beta_Q$ is attained.
Our aim is to bound the extractability from below by a linear function of the observed violation value $\beta$:
\begin{equation}
  \label{eq:self-test}
  \inf_{\rho \in \mathcal{R}_{\mathcal{B}(\beta)}} \Xi \left( \rho \to \Psi \right) \geq s\beta+ \mu,
\end{equation}
where $s, \mu \in \mathbb{R}$ are constants and $\mathcal{R}_{\mathcal{B}(\beta)}$ denotes the set of states compatible with violation $\beta$.

To translate this into an operator inequality, we write the Bell operator of $\mathcal{B}$ as
\begin{equation*}
  W_n \coloneqq \sum_{k=1}^n\sum_{\substack{x_k \in \Sigma_k \\ o_k \in \Omega_k}} 
  c\bigl(\{o_k, x_k\}_{k=1}^n\bigr) \,
  O\bigl(\{o_k, x_k\}_{k=1}^n\bigr),
\end{equation*}
where $c(\{o_k, x_k\}) \in \mathbb{R}$ are real coefficients, $x_k$ and $o_k$ are input and output of the $k$-th party, and
\begin{equation*}
  O\bigl( \{o_k, x_k\}_{k=1}^{n} \bigr) = \bigotimes_{k=1}^n O_{o_k}^{x_k}
\end{equation*}
with $\{O_{o_k}^{x_k}\}_{o_k}$ being the measurement conducted for input $x_k$.
For any state $\rho \in \mathcal{R}_{\mathcal{B}(\beta)}$, we have $\beta \leq \operatorname{tr}[\rho W]$.

Because $\Psi$ is pure, the fidelity simplifies to
\begin{align*}
  F\bigl(\Lambda(\rho), \Psi\bigr) 
    &= \operatorname{tr}\bigl[\Lambda(\rho) \Psi\bigr] \\
    &= \langle \Lambda(\rho), \Psi \rangle 
     = \langle \rho, \Lambda^\dagger(\Psi) \rangle \\
    &= \operatorname{tr}\bigl[\rho \Lambda^\dagger(\Psi)\bigr],
\end{align*}
where $\langle \cdot, \cdot\rangle$ is the Hilbert-Schmidt inner product and $\Lambda^\dagger$ is the dual map of $\Lambda$.

Define $K \coloneqq \Lambda^\dagger(\Psi)$.
If, for a suitable local channel $\Lambda$, we can establish the operator inequality (in the L\"owner order)
\begin{equation}
  \label{eq:operator-ineq}
  K \geq sW + \mu I,
\end{equation}
then for any $\rho \in \mathcal{R}_{\mathcal{B}}(\beta)$,
\begin{align*}
  \Xi (\rho \to \Psi) 
    &= \max_{\Lambda} F\bigl(\Lambda(\rho), \Psi\bigr) \geq \operatorname{tr}[\rho K] \\
    &\geq \operatorname{tr}\bigl[\rho (sW + \mu I)\bigr] \geq s\beta + \mu,
\end{align*}
which is precisely the desired bound Eq.~\eqref{eq:self-test}.
For the bound to be device-independent, Eq.~\eqref{eq:operator-ineq} must hold for arbitrary measurement operators.

Proving such an inequality is generally difficult because the measurement operators can act on arbitrarily high-dimensional spaces and the optimization over local channels is non-convex.
A tractable route was introduced in \cite{Kaniewski_2016} for the family of Werner-Wolf-\.Zukowski-Brukner inequality \cite{Werner_Wolf_2001,Zukowski_Brukner_2002}, which includes the $n$-qubit Mermin-Ardehali-Belinsk\u{i}-Klyshko (MABK) inequality \cite{Ardehali_1992,Belinskii_Klyshko_1993} that we focus on in this Letter.

\subsection{Multi-qubit MABK Inequality.}

$n$-qubit MABK inequality has two inputs and two outputs per party.
Denote the two observables of party $k$ by $\{O_{r_k}, O_{r_k \oplus 1}\}$, where $\oplus$ denotes addition modulo~2 and
\begin{equation*}
  r_k \coloneqq \frac{1}{2} \bigl( 1 - (-1)^{k-1} \bigr) \in \{0, 1\}
\end{equation*}
identifies the observable associated with input $x=0$.
The Bell operator $W_n$ for the $n$-party MABK inequality admits the recursive construction \cite{Fabritiis_Roditi_Sorella_2023}:
\begin{equation*}
  \begin{aligned}
    W_n = &\frac{1}{2} W_{n-1} \otimes \bigl( O_{r_n} + O_{r_n \oplus 1} \bigr) \\
          &- \frac{1}{2} \overline{W}_{n-1} \otimes \bigl( O_{r_n} - O_{r_n \oplus 1} \bigr),
  \end{aligned}
\end{equation*}
with $W_1 = 2O_{r_1}$ and $\overline{W}_{n}$ obtained from $W_n$ by replacing every $O_{r_k}$ with $O_{r_k \oplus 1}$.
The familiar three-qubit Mermin inequality corresponds to the case $n=3$.

By Naimark's dilation theorem \cite{Neumark1943}, we may assume without loss of generality that $O_0$ and $O_1$ are Hermitian projectors for all $k$.
Applying Jordan's lemma followed by a unitary rotation on each block (the first part of the extraction map \cite{Kaniewski_2016}), the observables are reduced to
\begin{equation*}
  O_{r_k} = \cos\theta_k \,\sigma_x + (-1)^{r_k} \sin\theta_k \,\sigma_z,
\end{equation*}
with angles $\theta_k \in [0, \pi/2]$.
This parametrization encompasses \emph{all possible choices} of observables \cite{Kaniewski_2016}.

The second part of the extraction map is a dephasing channel
\begin{equation*}
  [\Lambda(\theta)](\rho) = 
  \frac{1 + g(\theta)}{2}\,\rho + 
  \frac{1 - g(\theta)}{2}\,\Gamma(\theta)\rho\Gamma(\theta),
\end{equation*}
where $\theta$ is the angle,
\begin{equation*}
  \begin{aligned}
  g(\theta) \coloneqq& \bigl(1+\sqrt{2}\bigr)(\sin\theta + \cos\theta - 1),
  \\
  \Gamma(\theta) \coloneqq& \begin{cases}
    \sigma_x, & \theta \in [0,\pi/4], \\[2pt]
    \sigma_z, & \theta \in (\pi/4,\pi/2].
  \end{cases}
  \end{aligned}
\end{equation*}

The maximal quantum violation of the $n$-qubit MABK inequality is $\beta_{nQ} = 2^{(n+1)/2}$, attainable only when all angles satisfy $\theta_1 = \theta_2 = \cdots = \theta_n = \pi/4$.
The corresponding optimal state, denoted $|\Phi_n\rangle$, is locally equivalent to the $n$-qubit GHZ state
\begin{equation*}
  |\text{GHZ}_n\rangle \coloneqq \frac{1}{\sqrt{2}} \bigl( |0\rangle^{\otimes n} + |1\rangle^{\otimes n} \bigr).
\end{equation*}

\subsection{Multi-qubit Maximally Entangled States}

The optimal state $|\Phi_n\rangle$ of the $n$-qubit MABK inequality can be recursively constructed by
\begin{equation*}
  |\Phi_n\rangle = \cos\left( \frac{\pi}{8} \right) |\Psi^+_n\rangle + \sin \left( \frac{\pi}{8} \right) |\Phi^-_n\rangle,
\end{equation*}
where
\begin{equation*}
    \begin{aligned}
        &|\Psi_{n}^+\rangle = \frac{1}{\sqrt{2}} \left( |\Psi_{n-1}^+\rangle \otimes |1\rangle + |\Phi_{n-1}^-\rangle \otimes |0\rangle \right), \\
        &|\Phi_n^-\rangle = \frac{1}{\sqrt{2}} \left( |\Psi_{n-1}^+\rangle \otimes |0\rangle - |\Phi_{n-1}^-\rangle \otimes |1\rangle \right), \\
        &|\Psi_1^+\rangle := |0\rangle, \qquad |\Phi_1^-\rangle := |1\rangle.
    \end{aligned}
\end{equation*}
It is a pure state locally equivalent to $|\text{GHZ}_n\rangle$.

\begin{proposition}[Result~\protect\ref{result:main}, restated]
  \label{prop:main}
  Let $s_n = \frac{\sqrt{2} + 1}{2^\frac{n}{2}}$ and $\, \mu = - \left( 1 + \sqrt{2} \right)$, then for all integer $n \geq 3$, the following operator inequality
  \begin{equation}
    \label{eq:operator-ineq-main-sm}
    K \left( \{\theta_k\}_{k=1}^n \right) \geq s_n W_n \left( \{\theta_k\}_{k=1}^n \right) + \mu I
  \end{equation}
  holds for all $\{\theta_k\}_{k=1}^n \in [0, \pi/2]^n$.
  Here $K \left( \{\theta_k\}_{k=1}^n \right) \coloneqq \left[\Lambda \left( \{\theta_k\}_{k=1}^n \right)\right](|\Phi_n\rangle \langle\Phi_n|)$ is the dephased target state.
\end{proposition}

\begin{remark}
    \label{remark:s:1} Consider the following local unitary operator:
    \begin{equation*}
        \begin{aligned}
            &U \coloneqq U_1 \otimes U_2 \otimes \cdots \otimes U_n \\
            &U_1 \coloneqq \sigma_z R_z(-\pi/2) R_x(\pi/2) R_y(-\pi/4), \quad U_k \coloneqq R_x \left( \omega_k \right), \qquad \omega_k = (-1)^k \frac{\pi}{2}.
        \end{aligned}
    \end{equation*}
    One can verify that
    \begin{equation*}
        U|\Phi_n\rangle = \frac{1}{\sqrt{2}} \left[ |0\rangle^{\otimes n} + |1\rangle^{\otimes n} \right],
    \end{equation*}
    which is exactly the $n$-qubit GHZ state.
\end{remark}

  \subsection{Generalized Bloch Representation of Operator Inequality}

First, we reduce the domain of angles of observables from $[0, \pi/2]$ to $[0, \pi/4]$.

\begin{observation}[Symmetry of angle of qubit observables]
  \label{obv:angle-symmetry}
  On $n$-qubit quantum system, there exists $n$ unitary operators $\mathcal{S}_n = \{S_n\}_{k=1}^n$ satisfies
  \begin{equation*}
    S_k T_n \left( \{\theta_k\}_{k=1}^n \right) S_k = T_n \left( \left\{ \theta_l \right\}_{l=1}^{k-1} \cup \left\{ \frac{\pi}{2} - \theta_k \right\} \cup \{ \theta_l \}_{l=k+1}^{n} \right), \quad\quad \forall k \in [n].
  \end{equation*}
  Specifically, $\mathcal{S}_n$ can be recursively constructed by
  \begin{equation}
    \label{eq:stabilizers}
    \begin{aligned}
      &\mathcal{S}_n = \left\{ S \otimes \sigma_x \middle| S \in \mathcal{S}_{n-1} \right\} \cup \left\{ \sigma_x^{\otimes n-1} \otimes \frac{1}{\sqrt{2}}  \left( \sigma_x + (-1)^n \sigma_z \right) \right\}, \\
      &\mathcal{S}_1 = \{H\}.
    \end{aligned}
  \end{equation}
\end{observation}
To prove this Observation, we first identify $\mathcal{S}_n$ as a set of stabilizers of $|\Phi_n\rangle$.

\begin{lemma}
  \label{lemma:angle-symmetry}
  Let $\mathcal{S}_n$ be the set of operators defined in Eq.~\eqref{eq:stabilizers}. Then we have
  \begin{equation*}
    \forall S_k \in \mathcal{S}_n, \qquad S_k |\Phi_n\rangle = |\Phi_n\rangle.
  \end{equation*}
\end{lemma}
\begin{proof}
  Prove by induction on $n$.
  Start from the 2-qubit case, we have 
\begin{enumerate}
    \item Operator $H\otimes \sigma_x$ and $\sigma_x \otimes H$ stabilize $|\Phi_2\rangle$, that is,
    \begin{equation*}
        \begin{aligned}        
            &(H \otimes \sigma_x) |\Phi_2 \rangle 
            = |\Phi_2\rangle, \quad
            &(\sigma_x \otimes H) |\Phi_2 \rangle 
            = |\Phi_2\rangle. 
        \end{aligned}
    \end{equation*}
    Note that two states
    \begin{equation*}
        \sin \left( \frac{\pi}{8} \right) |0\rangle + \cos \left( \frac{\pi}{8} \right) |1\rangle, \quad\quad \cos \left( \frac{\pi}{8} \right) |0\rangle - \sin \left( \frac{\pi}{8} \right) |1\rangle
    \end{equation*}
    are eigenvectors of unitary operator $V = (\sigma_x - \sigma_z) / \sqrt{2}$.
    \item Define two pairs of vectors by:
    \begin{equation*}
      \begin{aligned}
        &|\Upsilon_2^+\rangle \coloneqq \sin \left( \frac{\pi}{8} \right) |\Psi_2^+\rangle + \cos \left( \frac{\pi}{8} \right) |\Phi_2^-\rangle, &\quad&
        |\Xi_2^+\rangle \coloneqq \cos \left( \frac{\pi}{8} \right) |\Psi_2^+\rangle + \sin \left( \frac{\pi}{8} \right) |\Phi_2^-\rangle \\
        &|\Upsilon_2^-\rangle \coloneqq \cos \left( \frac{\pi}{8} \right) |\Psi_2^+\rangle - \sin \left( \frac{\pi}{8} \right) |\Phi_2^-\rangle, &\quad&
        |\Xi_2^-\rangle \coloneqq - \sin \left( \frac{\pi}{8} \right) |\Psi_2^+\rangle + \cos \left( \frac{\pi}{8} \right) |\Phi_2^-\rangle
      \end{aligned}
    \end{equation*}
    We have the following equalities:
    \begin{equation*}
        \begin{aligned}
            &(H \otimes \sigma_x) |\Upsilon_2^+\rangle 
            = |\Upsilon_2^-\rangle, 
            \quad
            (\sigma_x \otimes H) |\Upsilon_2^+\rangle 
            = |\Upsilon_2^-\rangle, 
            \\
            &(H \otimes \sigma_x) |\Xi_2^+\rangle 
            = |\Xi_2^+\rangle, 
            \quad
            (\sigma_x \otimes H) |\Xi_2^+\rangle 
            = |\Xi_2^+\rangle, 
            \\
            &(H \otimes \sigma_x) |\Xi_2^-\rangle 
            = - |\Xi_2^-\rangle, 
            \quad
            (\sigma_x \otimes H) |\Xi_2^-\rangle 
            = - |\Xi_2^-\rangle. 
        \end{aligned}
    \end{equation*}
    \item The action of $\sigma_x \otimes \sigma_x$ on $|\Psi_2^+\rangle$ and $|\Phi^-_2\rangle$: 
    \begin{equation*}
        \begin{aligned}
            (\sigma_x \otimes \sigma_x) |\Psi_2^+\rangle 
            = |\Psi_2^+\rangle, 
            \quad
            (\sigma_x \otimes \sigma_x) |\Phi_2^-\rangle 
            = - |\Phi_2^-\rangle. 
        \end{aligned}
    \end{equation*}
\end{enumerate}

For simplicity, define the following vectors recursively:
\begin{equation*}
  \begin{aligned}
        &|\Upsilon_n^+\rangle \coloneqq \sin \left( \frac{\pi}{8} \right) |\Psi_n^+\rangle + \cos \left( \frac{\pi}{8} \right) |\Phi_n^-\rangle, &\quad&
        |\Xi_n^+\rangle \coloneqq \cos \left( \frac{\pi}{8} \right) |\Psi_n^+\rangle + \sin \left( \frac{\pi}{8} \right) |\Phi_n^-\rangle \\
        &|\Upsilon_n^-\rangle \coloneqq \cos \left( \frac{\pi}{8} \right) |\Psi_n^+\rangle - \sin \left( \frac{\pi}{8} \right) |\Phi_n^-\rangle, &\quad&
        |\Xi_n^-\rangle \coloneqq - \sin \left( \frac{\pi}{8} \right) |\Psi_n^+\rangle + \cos \left( \frac{\pi}{8} \right) |\Phi_n^-\rangle
  \end{aligned}
\end{equation*}
Then we prove the following statements by induction on $n$: 

    \textbf{Statement 1}. Operator $S_k$ stabilizes $|\Phi_{n}\rangle$ for all $S_k \in \mathcal{S}_n$. 

    \textbf{Statement 2}. For arbitrary $S_k \in \mathcal{S}_{n}$, we have $S_k |\Upsilon_n^+\rangle = |\Upsilon_n^-\rangle$,
    $S_k |\Xi_n^+\rangle = |\Xi_n^+\rangle$, and $S_k|\Xi_n^-\rangle = - |\Xi_n^-\rangle$.

    \textbf{Statement 3}. The following equalities holds:
    \begin{equation*}
        \begin{aligned}
            \sigma_x^{\otimes n} |\Psi_n^+\rangle =& \frac{1}{2} \left( |\Psi_n^+\rangle + |\Phi_n^-\rangle \right) + \frac{(-1)^{n}}{2} \left( |\Psi_n^+\rangle - |\Phi_n^-\rangle \right), \\
            \sigma_x^{\otimes n} |\Phi_n^-\rangle =& \frac{1}{2} \left( |\Psi_n^+\rangle - |\Phi_n^-\rangle \right) - \frac{(-1)^{n}}{2} \left( |\Psi_n^+\rangle + |\Phi_n^-\rangle \right). 
        \end{aligned}
    \end{equation*}

First, we prove \textbf{Statement 3}. Suppose the statement holds for $(n-1)$-qubit system, then by definition of $|\Psi_n^+\rangle$ and $|\Phi_n^-\rangle$, we have
\begin{equation*}
    \begin{aligned}
        \sigma_x^{\otimes n} |\Psi_n^+\rangle =& \frac{1}{\sqrt{2}} \left( \sigma_x^{\otimes (n-1)} |\Psi^+_{n-1}\rangle \otimes \sigma_x|1\rangle + \sigma_x^{\otimes (n-1)} |\Phi_{n-1}^-\rangle \otimes \sigma_x|0\rangle \right) \\
        =& \left\{ \begin{matrix*}[l]
            \frac{1}{\sqrt{2}} \left( |\Phi_{n-1}^-\rangle \otimes |0\rangle + |\Psi_{n-1}^+\rangle \otimes |1\rangle \right) = |\Psi_n^+\rangle &\quad (n \mathrm{\ even})\\
            \frac{1}{\sqrt{2}} \left( |\Psi_{n-1}^+\rangle \otimes |0\rangle - |\Phi_{n-1}^-\rangle \otimes |1\rangle \right) = |\Phi_n^-\rangle &\quad (n \mathrm{\ odd})
        \end{matrix*} \right. \\
        =& \frac{1}{2} \left( |\Psi_n^+\rangle + |\Phi_n^-\rangle \right) + \frac{(-1)^{n}}{2} \left( |\Psi_n^+\rangle - |\Phi_n^-\rangle \right), \\
        \sigma_x^{\otimes n} |\Phi_n^-\rangle =& \frac{1}{\sqrt{2}} \left( \sigma_x^{\otimes (n-1)} |\Psi^+_{n-1}\rangle \otimes \sigma_x|0\rangle - \sigma_x^{\otimes (n-1)} |\Phi_{n-1}^-\rangle \otimes \sigma_x|1\rangle \right) \\
        =& \left\{ \begin{matrix*}[c]
            \frac{1}{\sqrt{2}} \left( |\Phi_{n-1}^-\rangle \otimes |1\rangle - |\Psi_{n-1}^+\rangle \otimes |0\rangle \right) = - |\Phi_{n}^-\rangle &\quad (n \mathrm{\ even}) \\
            \frac{1}{\sqrt{2}} \left( |\Psi_{n-1}^+\rangle \otimes |1\rangle + |\Phi_{n-1}^-\rangle \otimes |0\rangle \right) = |\Psi_n^+\rangle &\quad (n \mathrm{\ odd})
        \end{matrix*} \right. \\
        =& \frac{1}{2} \left( |\Psi_n^+\rangle - |\Phi_n^-\rangle \right) - \frac{(-1)^{n}}{2} \left( |\Psi_n^+\rangle + |\Phi_n^-\rangle \right). 
    \end{aligned}
\end{equation*} 

Next, we prove \textbf{Statement 2}. Suppose the equalities hold for $(n-1)$-qubit system, then for some $S_k \in \mathcal{S}_{n-1}$, we have
\begin{equation*}
    \begin{aligned}
        &(S_k \otimes \sigma_x) |\Upsilon_{n}^+\rangle = (S_k \otimes \sigma_x) \left( \sin \left( \frac{\pi}{8} \right) |\Psi_n^+\rangle + \cos \left( \frac{\pi}{8} \right) |\Phi_n^-\rangle \right) \\
        =& \frac{1}{\sqrt{2}} \left( S_k |\Xi_{n-1}^+\rangle \otimes \sigma_x|0\rangle - S_k |\Xi_{n-1}^-\rangle \otimes \sigma_x|1\rangle \right) \\
        =& \frac{1}{\sqrt{2}} \left( \cos \left( \frac{\pi}{8} \right) \left( |\Psi_{n-1}^+\rangle \otimes |1\rangle + |\Phi_{n-1}^-\rangle \otimes |0\rangle \right) + \sin \left( \frac{\pi}{8} \right) \left( |\Phi_{n-1}^-\rangle \otimes |1\rangle - |\Psi_{n-1}^+\rangle \otimes |0\rangle \right) \right) 
        = |\Upsilon_{n}^-\rangle. 
    \end{aligned}
\end{equation*}
Regarding $|\Xi_n^+\rangle$ and $|\Xi_n^-\rangle$, we have 
\begin{equation*}
    \begin{aligned}
        &(S_k \otimes \sigma_x) |\Xi_n^+\rangle = (S_k \otimes \sigma_x) \left( \cos \left( \frac{\pi}{8} \right) |\Psi_n^+\rangle + \sin \left( \frac{\pi}{8} \right) |\Phi_n^-\rangle \right) \\
        =& \frac{1}{\sqrt{2}} \left( S_k |\Upsilon_{n-1}^+\rangle \otimes \sigma_x|0\rangle + S_k |\Upsilon_{n-1}^-\rangle \otimes \sigma_x|1\rangle \right) \\
        =& \frac{1}{\sqrt{2}} \left( \cos \left( \frac{\pi}{8} \right) \left( |\Psi_{n-1}^+\rangle \otimes |1\rangle + |\Phi_{n-1}^-\rangle \otimes |0\rangle \right) + \sin \left( \frac{\pi}{8} \right) \left( |\Psi_{n-1}^+\rangle \otimes |0\rangle - |\Phi_{n-1}^-\rangle \otimes |1\rangle \right) \right) 
        = |\Xi_n^+\rangle, 
    \end{aligned}
\end{equation*}
\begin{equation*}
    \begin{aligned}
        &(S_k \otimes \sigma_x) |\Xi_n^-\rangle = (S_k \otimes \sigma_x) \left( - \sin \left( \frac{\pi}{8} \right) |\Psi_n^+\rangle + \cos \left( \frac{\pi}{8} \right) |\Phi_n^-\rangle \right) \\
        =& \frac{1}{\sqrt{2}} \left( S_k |\Upsilon_{n-1}^-\rangle \otimes \sigma_x|0\rangle - S_k |\Upsilon_{n-1}^+\rangle \otimes \sigma_x|1\rangle \right) \\
        =& \frac{1}{\sqrt{2}} \left( \sin \left( \frac{\pi}{8} \right) \left( |\Psi_{n-1}^+\rangle \otimes |1\rangle + |\Phi_{n-1}^-\rangle \otimes |0\rangle \right) - \cos \left( \frac{\pi}{8} \right) \left( |\Psi_{n-1}^+\rangle \otimes |0\rangle - |\Phi_{n-1}^-\rangle \otimes |1\rangle \right) \right) 
        = - |\Xi_n^-\rangle.
    \end{aligned}
\end{equation*}
Now consider operator $S_n = (1/\sqrt{2}) \sigma_x^{\otimes (n-1)} \otimes (\sigma_x + (-1)^n \sigma_z) \in \mathcal{S}_n$ which is not constructed from operators in $\mathcal{S}_{n-1}$, by \textbf{Statement 3} we have
\begin{equation*}
    \begin{aligned}
        &S_n |\Upsilon_n^+\rangle = \frac{1}{\sqrt{2}} \left( \sigma_x^{\otimes (n-1)} \otimes \left( \sigma_x + (-1)^{n} \sigma_z \right) \right) |\Upsilon_n^+\rangle \\
        =& \frac{1}{\sqrt{2}} \left[ \begin{matrix*}[l]
            + & \left( \frac{1}{2} \left( |\Psi_{n-1}^+\rangle + |\Phi_{n-1}^-\rangle \right) + \frac{(-1)^{n-1}}{2} \left( |\Psi_{n-1}^+\rangle - |\Phi_{n-1}^-\rangle \right) \right) \\
            & \otimes \frac{1}{2} \left( \left( |\Xi_1^+\rangle + |\Xi_1^-\rangle \right) + (-1)^n \left( |\Xi_1^+\rangle - |\Xi_1^-\rangle \right) \right) \\
            + & \frac{1}{2} \left( |\Psi_{n-1}^+\rangle - |\Phi_{n-1}^-\rangle \right) - \frac{(-1)^{n-1}}{2} \left( |\Psi_{n-1}^+\rangle + |\Phi_{n-1}^-\rangle \right) \\
            & \otimes \frac{1}{2} \left( (-1)^n \left( |\Xi_1^+\rangle + |\Xi_1^-\rangle \right) - \left( |\Xi_1^+\rangle - |\Xi_1^-\rangle \right) \right) 
        \end{matrix*}\right] \\
        =& \frac{1}{\sqrt{2}} \left( |\Psi_{n-1}^+\rangle \otimes \left( -\sin \left( \frac{\pi}{8} \right) |0\rangle + \cos \left( \frac{\pi}{8} \right) |1\rangle \right) + |\Phi_{n-1}^-\rangle \otimes \left( \cos \left( \frac{\pi}{8} \right) |0\rangle + \sin \left( \frac{\pi}{8} \right) |1\rangle \right) \right) \\
        =& \cos \left( \frac{\pi}{8} \right) |\Psi_{n}^+\rangle - \sin \left( \frac{\pi}{8} \right) |\Phi_{n}^-\rangle = |\Upsilon_n^-\rangle.  
    \end{aligned}
\end{equation*}
For $|\Xi_n^+\rangle$ and $|\Xi_n^-\rangle$, we have 
\begin{equation*}
    \begin{aligned}
        &S_n |\Xi_n^+\rangle = \frac{1}{\sqrt{2}} \left( \sigma_x^{\otimes (n-1)} \otimes \left( \sigma_x + (-1)^{n} \sigma_z \right) \right) |\Xi_n^+\rangle \\
        =& \frac{1}{\sqrt{2}} \left[ \begin{matrix*}[l]
            + & \left( \frac{1}{2} \left( |\Psi_{n-1}^+\rangle + |\Phi_{n-1}^-\rangle \right) + \frac{(-1)^{n-1}}{2} \left( |\Psi_{n-1}^+\rangle - |\Phi_{n-1}^-\rangle \right) \right) \\
            & \otimes \frac{1}{2} \left( \left( |\Upsilon_1^+\rangle + |\Upsilon_1^-\rangle \right) - (-1)^n \left( |\Upsilon_1^+\rangle - |\Upsilon_1^-\rangle \right) \right) \\
            + & \frac{1}{2} \left( |\Psi_{n-1}^+\rangle - |\Phi_{n-1}^-\rangle \right) - \frac{(-1)^{n-1}}{2} \left( |\Psi_{n-1}^+\rangle + |\Phi_{n-1}^-\rangle \right) \\
            & \otimes \frac{1}{2} \left( \left( |\Upsilon_1^+\rangle - |\Upsilon_1^-\rangle \right) + (-1)^n \left( |\Upsilon_1^+\rangle + |\Upsilon_1^-\rangle \right) \right)
        \end{matrix*} \right] \\
        =& \frac{1}{\sqrt{2}} \left( |\Psi_{n-1}^+\rangle \otimes \left( \sin \left( \frac{\pi}{8} \right) |0\rangle + \cos \left( \frac{\pi}{8} \right) |1\rangle \right) + |\Phi_{n-1}^-\rangle \otimes \left( \cos \left( \frac{\pi}{8} \right) |0\rangle - \sin \left( \frac{\pi}{8} \right) |1\rangle \right) \right) \\
        =& \cos \left( \frac{\pi}{8} \right) |\Psi_{n}^+\rangle + \sin \left( \frac{\pi}{8} \right) |\Phi_n^-\rangle = |\Xi_n^+\rangle. 
    \end{aligned}
\end{equation*}
\begin{equation*}
    \begin{aligned}
        &\frac{1}{\sqrt{2}} \left( \sigma_x^{\otimes (n-1)} \otimes \left( \sigma_x + (-1)^{n} \sigma_z \right) \right) |\Xi_n^-\rangle \\
        =& \frac{1}{\sqrt{2}} \left[ \begin{matrix*}[l]
            + & \left( \frac{1}{2} \left( |\Psi_{n-1}^+\rangle + |\Phi_{n-1}^-\rangle \right) + \frac{(-1)^{n-1}}{2} \left( |\Psi_{n-1}^+\rangle - |\Phi_{n-1}^-\rangle \right) \right) \\
            & \otimes \frac{1}{2} \left( \left( |\Upsilon_1^+\rangle - |\Upsilon_1^-\rangle \right) + (-1)^n \left( |\Upsilon_1^+\rangle + |\Upsilon_1^-\rangle \right) \right) \\
            - & \frac{1}{2} \left( |\Psi_{n-1}^+\rangle - |\Phi_{n-1}^-\rangle \right) - \frac{(-1)^{n-1}}{2} \left( |\Psi_{n-1}^+\rangle + |\Phi_{n-1}^-\rangle \right) \\
            & \otimes \frac{1}{2} \left( \left( |\Upsilon_1^+\rangle + |\Upsilon_1^-\rangle \right) - (-1)^n \left( |\Upsilon_1^+\rangle - |\Upsilon_1^-\rangle \right) \right) 
        \end{matrix*} \right] \\
        =& \frac{1}{\sqrt{2}} \left( - |\Psi_{n-1}^+\rangle \otimes \left( \cos \left( \frac{\pi}{8} \right) |0\rangle - \sin \left( \frac{\pi}{8} \right) |1\rangle \right) + |\Phi_{n-1}^-\rangle \otimes \left( \sin \left( \frac{\pi}{8} \right) |0\rangle + \cos \left( \frac{\pi}{8} \right) |1\rangle \right) \right) \\
        =& \sin \left( \frac{\pi}{8} \right) |\Psi_n^+\rangle - \cos \left( \frac{\pi}{8} \right) |\Phi_n^-\rangle = - |\Xi_n^-\rangle. 
    \end{aligned}
\end{equation*}

We are ready to inductively prove \textbf{Statement 1} now, by definition of $|\Phi_n\rangle$, we have:
\begin{equation*}
    \begin{aligned}
        |\Phi_n\rangle =& \cos \left( \frac{\pi}{8} \right) |\Psi_n^+\rangle + \sin \left( \frac{\pi}{8} \right) |\Phi_n^-\rangle \\
        =& \frac{1}{\sqrt{2}} \cos \left( \frac{\pi}{8} \right) \left( |\Psi_{n-1}^+\rangle \otimes |1\rangle + |\Phi_{n-1}^-\rangle \otimes |0\rangle \right) + \frac{1}{\sqrt{2}} \sin \left( \frac{\pi}{8} \right) \left( |\Psi_{n-1}^+\rangle \otimes |0\rangle - |\Phi_{n-1}^-\rangle \otimes |1\rangle \right) \\
        =& \frac{1}{\sqrt{2}} |\Psi_{n-1}^+\rangle \otimes \left( \sin \left( \frac{\pi}{8} \right) |0\rangle + \cos \left( \frac{\pi}{8} \right) |1\rangle \right) + \frac{1}{\sqrt{2}} |\Phi_{n-1}^-\rangle \otimes \left( \cos \left( \frac{\pi}{8} \right) |0\rangle - \sin \left( \frac{\pi}{8} \right) |1\rangle \right).
    \end{aligned}
\end{equation*}
That is,
    \begin{align}
        |\Phi_n\rangle =& \frac{1}{\sqrt{2}} \left( |\Upsilon_{n-1}^+\rangle \otimes |0\rangle + |\Upsilon_{n-1}^-\rangle \otimes |1\rangle \right) \label{eq:87}\\
        =& \frac{1}{\sqrt{2}} \left( |\Psi_{n-1}^+\rangle \otimes |\Upsilon_1^+\rangle + |\Phi_{n-1}^-\rangle \otimes |\Upsilon_1^-\rangle \right). \label{eq:88}
    \end{align}
Suppose \textbf{Statement 1} holds for $(n-1)$-qubit system, then, for arbitrary $S_{k} \in \mathcal{S}_{n-1}$, using Eq.~\eqref{eq:87} and \textbf{Statement 2}, we have 
\begin{equation*}
    \begin{aligned}
        (S_{k} \otimes \sigma_x) |\Phi_{n}\rangle 
        =& \frac{1}{\sqrt{2}} \left( S_k |\Upsilon_{n-1}^+\rangle \otimes \sigma_x|0\rangle + S_k |\Upsilon_{n-1}^-\rangle \otimes \sigma_x|1\rangle \right) \\
        =& \frac{1}{\sqrt{2}} \left( |\Upsilon_{n-1}^-\rangle \otimes |1\rangle + |\Upsilon_{n-1}^+\rangle \otimes |0\rangle \right) = |\Phi_{n}\rangle. 
    \end{aligned}
\end{equation*}
For $S_n = (1/\sqrt{2}) \sigma_x^{\otimes (n-1)} \otimes (\sigma_x + (-1)^n \sigma_z)$, using Eq.~\eqref{eq:88} and \textbf{Statement 3}, we have:
\begin{equation*}
    \begin{aligned}
        &\frac{1}{\sqrt{2}} \left( \sigma_x^{\otimes (n-1)} \otimes (\sigma_x + (-1)^n \sigma_z) \right) |\Phi_n\rangle \\
        =& \frac{1}{2} \left( \begin{matrix*}[l]
            & \sigma_x^{\otimes (n-1)} |\Psi_{n-1}^+\rangle \otimes (\sigma_x + (-1)^n \sigma_z) \left( \sin \left( \frac{\pi}{8} \right) |0\rangle + \cos \left( \frac{\pi}{8} \right) \right) \\
            + & \sigma_x^{\otimes (n-1)} |\Phi_{n-1}^-\rangle \otimes (\sigma_x + (-1)^n \sigma_z) \left( \cos \left( \frac{\pi}{8} \right) |0\rangle - \sin \left( \frac{\pi}{8} \right) |1\rangle \right)
        \end{matrix*} \right) \\
        =& \frac{1}{\sqrt{2}} \left( \begin{matrix*}[l]
            & \left(\frac{1}{2} \left( |\Psi_{n-1}^+\rangle + |\Phi_{n-1}^-\rangle \right) + \frac{(-1)^{n-1}}{2} \left( |\Psi_{n-1}^+\rangle - |\Phi_{n-1}^-\rangle \right) \right) \\
            & \otimes \frac{1}{2} \left( \left( |\Upsilon_1^+\rangle + |\Upsilon_1^-\rangle \right) - (-1)^n \left( |\Upsilon_1^+\rangle - |\Upsilon_1^-\rangle \right) \right) \\
            + & \left( \frac{1}{2} \left( |\Psi_{n-1}^+\rangle - |\Phi_{n-1}^-\rangle \right) - \frac{(-1)^{n-1}}{2} \left( |\Psi_{n-1}^+\rangle + |\Phi_{n-1}^-\rangle \right) \right) \\
            & \otimes \frac{1}{2} \left( \left( |\Upsilon_1^+\rangle - |\Upsilon_1^-\rangle \right) + (-1)^n \left( |\Upsilon_1^+\rangle + |\Upsilon_1^-\rangle \right) \right)
        \end{matrix*} \right) \\
        =& \frac{1}{\sqrt{2}} \left( |\Psi_{n-1}^+\rangle \otimes |\Upsilon_1^+\rangle + |\Phi_{n-1}^-\rangle \otimes |\Upsilon_1^-\rangle \right) = |\Phi_n\rangle. 
    \end{aligned}
\end{equation*}
This completes the proof of \textbf{Statement 1}, which is exactly the expected Lemma.
\end{proof}

Now, we can prove Observation~\ref{obv:angle-symmetry}.

\begin{proof}[Proof of Observation~\ref{obv:angle-symmetry}]
This proof is divided into two stages, in the first one we prove the symmetry of angles of Bell operator $W_n \left( \{\theta_k\}_{k=1}^n \right)$, and in the second stage we consider angles of the dephased state $K \left( \{\theta_k\}_{k=1}^n \right)$. 

\noindent \textbf{Stage 1}. 
First, we inductively prove that $\sigma_x^{\otimes n} W_n \sigma_x^{\otimes n} = \overline{W}_n$ for all $n \in \mathbb{N}$.
Here $\overline{W}_n$ is obtained from $W_n$ by replacing every $O_{r_k}$ with $O_{r_k \oplus 1}$.
We have 
\begin{equation}
    \begin{aligned}
        \sigma_x W_1(\theta_1) \sigma_x = &2 \left( \cos \theta_1 \sigma_x\sigma_x\sigma_x + \sin \theta_1 \sigma_x\sigma_z\sigma_x \right) = 2 \left( \cos \theta_1 \sigma_x - \sin \theta_1 \sigma_z \right) = \overline{W}_1 (\theta_1), \\
        \sigma_x^{\otimes n} W_n \left( \{\theta_l\}_{l=1}^n \right) \sigma_x^{\otimes n} = &\sigma_x^{\otimes (n-1)} W_{n-1} \left( \{\theta_l\}_{l=1}^{n-1} \right) \sigma_x^{\otimes (n-1)} \otimes \cos \theta_n \sigma_x\sigma_x\sigma_x \\
        &+ (-1)^n \sigma_x^{\otimes (n-1)} \overline{W}_{n-1} \left( \{\theta_l\}_{l=1}^{n-1} \right) \sigma_x^{\otimes (n-1)} \otimes \sin \theta_n \sigma_x\sigma_z\sigma_x \\
        = &\overline{W}_{n-1} \left( \{\theta_l\}_{l=1}^{n-1} \right) \otimes \cos \theta_n \sigma_x + (-1)^{n-1} {W}_{n-1} \left( \{\theta_l\}_{l=1}^{n-1} \right) \otimes \sin \theta_n \sigma_z = \overline{W}_{n}. 
    \end{aligned}
    \label{eq:85}
\end{equation}

Next, we prove that operator $S_k \in \mathcal{S}_n$ satisfies
\begin{equation}
    \label{eq:sym-induction}
    \begin{aligned}
        S_k W \left( \left\{ \theta_l \right\}_{l=1}^n \right) S_k =& W \left( \{ \theta_k \}_{l=1}^{k-1} \cup \left\{ \frac{\pi}{2} - \theta_k \right\} \cup \{ \theta_l \}_{l=k+1}^n \right), \\
        S_k \overline{W} \left( \left\{ \theta_l \right\}_{l=1}^n \right) S_k =& - \overline{W} \left( \{ \theta_k \}_{l=1}^{k-1} \cup \left\{ \frac{\pi}{2} - \theta_k \right\} \cup \{ \theta_l \}_{l=k+1}^n \right).
    \end{aligned}
\end{equation}

Start from single qubit system, we have 
\begin{equation*}
    \begin{aligned}
        H W_1 (\theta_1) H = &2 \left( \cos \theta_1 H \sigma_x H + \sin \theta_1 H \sigma_z H \right) = 2 \left( \sin \left( \frac{\pi}{2} - \theta_1 \right) \sigma_z + \cos \left( \frac{\pi}{2} - \theta_1 \right) \sigma_x \right) = W_1 \left( \frac{\pi}{2} - \theta_1 \right), \\
        H \overline{W}_1 (\theta_1) H = &2 \left( \cos \theta_1 H \sigma_x H - \sin \theta_1 H \sigma_z H \right) = 2 \left( \sin \left( \frac{\pi}{2} - \theta_1 \right) \sigma_z - \cos \left( \frac{\pi}{2} - \theta_1 \right) \sigma_x \right) = - \overline{W}_1 \left( \frac{\pi}{2} - \theta_1 \right). 
    \end{aligned}
\end{equation*}
On 2-qubit system, it can also be easily verified that 
\begin{equation*}
    \begin{aligned}
        &(H \otimes \sigma_x) W_2(\theta_1, \theta_2) (H \otimes \sigma_x) = W_2 \left( \frac{\pi}{2} - \theta_1, \theta_2 \right), &\quad& (\sigma_x \otimes H) W_2(\theta_1, \theta_2) (\sigma_x \otimes H) = W_2 \left( \theta_1, \frac{\pi}{2} - \theta_2 \right), \\ 
        &(H \otimes \sigma_x) \overline{W}_2(\theta_1, \theta_2) (H \otimes \sigma_x) = - \overline{W}_2 \left( \frac{\pi}{2} - \theta_1, \theta_2 \right), &\quad& (\sigma_x \otimes H) \overline{W}_2(\theta_1, \theta_2) (\sigma_x \otimes H) = - \overline{W}_2 \left( \theta_1, \frac{\pi}{2} - \theta_2 \right). 
    \end{aligned}
\end{equation*}

Suppose Eq.~\eqref{eq:sym-induction} holds for $n-1$.
Then, by straightforward calculation (the angles of $W_n$ are omitted in expansions for simplicity),
\begin{equation*}
    \begin{aligned}
        &\left( S_k \otimes \sigma_x \right) W_n \left( \{ \theta_l \}_{l=1}^{n-1} \right) \left( S_k \otimes \sigma_x \right) \\
        =& \left[ \begin{matrix*}[l]
            +1 & \left( S_k \otimes \sigma_x \right) \left( W_{n-1} \otimes \cos \theta_k \sigma_x \right) \left( S_k \otimes \sigma_x \right) \\
            +(-1)^n &\left( S_k \otimes \sigma_x \right) \left( \overline{W}_{n-1} \otimes \sin \theta_k \sigma_z \right) \left( S_k \otimes \sigma_x \right)
        \end{matrix*} \right] \\
        =& \left[ \begin{matrix*}[l]
            +1 & W \left( \{ \theta_l \}_{l=1}^{k-1} \cup \left\{ \frac{\pi}{2} - \theta_k \right\} \cup \{\theta_l\}_{l=k+1}^{n-1} \right) \otimes \cos \theta_k \sigma_x \\
            +(-1)^n & (-1) \overline{W} \left( \{ \theta_l \}_{l=1}^{k-1} \cup \left\{ \frac{\pi}{2} - \theta_k \right\} \cup \{\theta_l\}_{l=k+1}^{n-1} \right) \otimes (-1) \sin \theta_k \sigma_z
        \end{matrix*} \right] \\
        = &W_n \left( \{ \theta_l \}_{l=1}^{k-1} \cup \left\{ \frac{\pi}{2} - \theta_k \right\} \cup \{\theta_l\}_{l=k+1}^{n} \right). 
    \end{aligned}
\end{equation*}
Similarly, 
\begin{equation*}
    \begin{aligned}
        &\left( S_k \otimes \sigma_x \right) \overline{W}_n \left( \{ \theta_l \}_{l=1}^{n-1} \right) \left( S_k \otimes \sigma_x \right) \\
        =& \left[ \begin{matrix*}[l]
            +1 &\left( S_k \otimes \sigma_x \right) \left( \overline{W}_{n-1} \otimes \cos \theta_k \sigma_x \right) \left( S_k \otimes \sigma_x \right) \\
            +(-1)^{n-1} &\left( S_k \otimes \sigma_x \right) \left( {W}_{n-1} \otimes \sin \theta_k \sigma_z \right) \left( S_k \otimes \sigma_x \right)
        \end{matrix*} \right] \\
        =& \left[ \begin{matrix*}[l]
            +1 & (-1) \overline{W} \left( \{ \theta_l \}_{l=1}^{k-1} \cup \left\{ \frac{\pi}{2} - \theta_k \right\} \cup \{\theta_l\}_{l=k+1}^{n-1} \right) \otimes \cos \theta_k \sigma_x \\
            +(-1)^{n-1} & {W} \left( \{ \theta_l \}_{l=1}^{k-1} \cup \left\{ \frac{\pi}{2} - \theta_k \right\} \cup \{\theta_l\}_{l=k+1}^{n-1} \right) \otimes (-1) \sin \theta_k \sigma_z
        \end{matrix*} \right] \\
        = &- \overline{W}_n \left( \{ \theta_l \}_{l=1}^{k-1} \cup \left\{ \frac{\pi}{2} - \theta_k \right\} \cup \{\theta_l\}_{l=k+1}^{n} \right). 
    \end{aligned}
\end{equation*}

With Eq.~\eqref{eq:85} established at the begining of this stage, the following equality can be easily verified: 
\begin{equation*}
    \begin{aligned}
        &S_n W_n \left( \{\theta_k\}_{k=1}^n \right) S_n = \left( \sigma_x^{\otimes (n-1)} \otimes \frac{1}{\sqrt{2}} \left( \sigma_x + (-1)^n \sigma_z \right) \right) W_n \left( \{ \theta_k \}_{k=1}^{n} \right) \left( \sigma_x^{\otimes (n-1)} \otimes \frac{1}{\sqrt{2}} \left( \sigma_x + (-1)^n \sigma_z \right) \right) \\
        =&\left[ \begin{matrix*}[l]
            +1 & \frac{1}{2} \left( \sigma_x^{\otimes (n-1)} \otimes \left( \sigma_x + (-1)^n \sigma_z \right) \right) \left( W_{n-1} \otimes \cos \theta_n \sigma_x \right) \left( \sigma_x^{\otimes (n-1)} \otimes \left( \sigma_x + (-1)^n \sigma_z \right) \right) \\
            +(-1)^{n} & \frac{1}{2} \left( \sigma_x^{\otimes (n-1)} \otimes \left( \sigma_x + (-1)^n \sigma_z \right) \right) \left( \overline{W}_{n-1} \otimes \sin \theta_n \sigma_z \right) \left( \sigma_x^{\otimes (n-1)} \otimes \left( \sigma_x + (-1)^n \sigma_z \right) \right)
        \end{matrix*}\right] \\
        =&\left[ \begin{matrix*}[l]
            +1 & \overline{W}_{n-1} \otimes \cos\theta_n (-1)^{n} \sigma_z \\
            +(-1)^n & W_{n-1} \otimes \sin \theta_n (-1)^n \sigma_x 
        \end{matrix*}\right] = \left[ \begin{matrix*}[l]
            +1 & \overline{W}_{n-1} \otimes \sin \left( \frac{\pi}{2} - \theta_n \right) (-1)^{n} \sigma_z \\
            +(-1)^n & W_{n-1} \otimes \cos \left( \frac{\pi}{2} - \theta_n \right) (-1)^n \sigma_x 
        \end{matrix*}\right] \\
        =& W_{n-1} \otimes \cos \left( \frac{\pi}{2} - \theta_n \right) \sigma_x + (-1)^n \overline{W}_{n-1} \otimes \sin \left( \frac{\pi}{2} - \theta_n \right) \sigma_z = W_n \left( \{\theta_k\}_{k=1}^{n-1} \cup \left\{ \frac{\pi}{2} - \theta_n \right\} \right). 
    \end{aligned}
\end{equation*}
Similarly, 
\begin{equation*}
    \begin{aligned}
        &S_n \overline{W}_n \left( \{ \theta_k \}_{k=1}^{n} \right) S_n = \left( \sigma_x^{\otimes (n-1)} \otimes \frac{1}{\sqrt{2}} \left( \sigma_x + (-1)^n \sigma_z \right) \right) \overline{W}_n \left( \{ \theta_k \}_{k=1}^{n} \right) \left( \sigma_x^{\otimes (n-1)} \otimes \frac{1}{\sqrt{2}} \left( \sigma_x + (-1)^n \sigma_z \right) \right) \\
        =&\left[ \begin{matrix*}[l]
            +1 & \frac{1}{2} \left( \sigma_x^{\otimes (n-1)} \otimes \left( \sigma_x + (-1)^n \sigma_z \right) \right) \left( \overline{W}_{n-1} \otimes \cos \theta_n \sigma_x \right) \left( \sigma_x^{\otimes (n-1)} \otimes \left( \sigma_x + (-1)^n \sigma_z \right) \right) \\
            +(-1)^{n-1} & \frac{1}{2} \left( \sigma_x^{\otimes (n-1)} \otimes \left( \sigma_x + (-1)^n \sigma_z \right) \right) \left( {W}_{n-1} \otimes \sin \theta_n \sigma_z \right) \left( \sigma_x^{\otimes (n-1)} \otimes \left( \sigma_x + (-1)^n \sigma_z \right) \right)
        \end{matrix*}\right] \\
        =&\left[ \begin{matrix*}[l]
            +1 & {W}_{n-1} \otimes \cos\theta_n (-1)^{n} \sigma_z \\
            +(-1)^{n-1} & \overline{W}_{n-1} \otimes \sin \theta_n (-1)^n \sigma_x 
        \end{matrix*}\right] = \left[ \begin{matrix*}[l]
            -1 & {W}_{n-1} \otimes \sin \left( \frac{\pi}{2} - \theta_n \right) (-1)^{n-1} \sigma_z \\
            -(-1)^{n} & \overline{W}_{n-1} \otimes \cos \left( \frac{\pi}{2} - \theta_n \right) (-1)^n \sigma_x 
        \end{matrix*}\right] \\
        =& - \overline{W}_{n-1} \otimes \cos \left( \frac{\pi}{2} - \theta_n \right) \sigma_x + (-1)^n {W}_{n-1} \otimes \sin \left( \frac{\pi}{2} - \theta_n \right) \sigma_z = - \overline{W}_n \left( \{\theta_k\}_{k=1}^{n-1} \cup \left\{ \frac{\pi}{2} - \theta_n \right\} \right)= - \overline{W}_n. 
    \end{aligned}
\end{equation*}

\noindent \textbf{Stage 2}. Since $\mathcal{S}_n$ has been constructed in \textbf{Stage 1}, in this stage we only need to prove that each operator $S_k \in \mathcal{S}_n$ works as an reflection on the $k$-th angle centering at $\pi / 2$. By construction, for $k \geq 2$ $S_k$ can be analytically expressed as 
\begin{equation}
    S_k = \sigma_x^{\otimes (k-1)} \otimes \frac{1}{\sqrt{2}} \left( \sigma_x + (-1)^k \sigma_z \right) \otimes \sigma_x^{\otimes (n-k)},
    \label{eq:86}
\end{equation}
and $S_1 = H \otimes \sigma_x^{\otimes (n-1)}$.
Then we have (unless otherwise specified, we define channel $\Gamma^0$ to be the identity channel, and $\Gamma^1 = \Gamma$): 
\begin{equation*}
    \begin{aligned}
        &S_k \left[ \Lambda_n \left( \{\theta_l\}_{l=1}^n \right) \right] (\rho) S_k = \sum_{e_1\cdots e_n \in \{0, 1\}^n} \left( \left( \prod_{l=1}^{n} \frac{1 + (-1)^{e_l} g(\theta_l)}{2} \right) S_k \left( \bigotimes_{l=1}^n \Gamma^{e_l} (\theta_l) \right) \rho \left( \bigotimes_{l=1}^n \Gamma^{e_l} (\theta_l) \right) S_k \right). 
    \end{aligned}
\end{equation*}

By definition, real function $g$ is symmetric about its input centering at $\pi/2$, i.e. $g(x) = g(\pi/2 - x)$ for all $x \in \mathbb{R}$. Therefore, it suffices to prove the following equality:
\begin{equation*}
    \begin{aligned}
        &S_k \left( \bigotimes_{l=1}^n \Gamma^{e_l} (\theta_l) \right) |\Phi\rangle_n \langle\Phi| \left( \bigotimes_{l=1}^n \Gamma^{e_l} (\theta_l) \right) S_k = D_k \left( \{e_l\}_{l=1}^n, \{\theta_l\}_{l=1}^n \right) |\Phi\rangle_n \langle\Phi| D_k \left( \{e_l\}_{l=1}^n, \{\theta_l\}_{l=1}^n \right)
    \end{aligned}    
\end{equation*}
holds for all $(e_1 \cdots e_n) \in \{0, 1\}^n$, where 
\begin{equation*}
    D_k \left( \{e_l\}_{l=1}^n, \{\theta_l\}_{l=1}^n \right) := \left[ \left( \bigotimes_{l=1}^{k-1} \Gamma^{e_l} (\theta_l) \right) \otimes \Gamma^{e_k} \left( \frac{\pi}{2} - \theta_k \right) \otimes \left( \bigotimes_{l=k+1}^{n} \Gamma^{e_l} (\theta_l) \right) \right].
\end{equation*}

For this, we can ``switch'' $S_k$ and the Kraus operator by adding additional $S_k$ by its unitary and Hermitian properties: 
\begin{equation*}
    S_k \left( \bigotimes_{l=1}^n \Gamma^{e_l} (\theta_l) \right) |\Phi\rangle_n \langle\Phi| \left( \bigotimes_{l=1}^n \Gamma^{e_l} (\theta_l) \right) S_k = S_k \left( \bigotimes_{l=1}^n \Gamma^{e_l} (\theta_l) \right) S_k S_k |\Phi\rangle_n \langle\Phi| S_k S_k \left( \bigotimes_{l=1}^n \Gamma^{e_l} (\theta_l) \right) S_k. 
\end{equation*}
Note that, alternation of sign of $\Gamma^{e_k}(\theta_k)$ for possible $V := (1/\sqrt{2})(\sigma_x - \sigma_z)$ contained in $S_k$ for even $k$ can be counterpated since there are two Kraus operator on two sides of $S_k |\Phi\rangle_n \langle \Phi| S_k$, i.e. 
\begin{equation*}
    S_k \left( \bigotimes_{l=1}^n \Gamma^{e_l} (\theta_l) \right) S_k \rho S_k \left( \bigotimes_{l=1}^n \Gamma^{e_l} (\theta_l) \right) S_k = D_k \left( \{e_l\}_{l=1}^n, \{\theta_l\}_{l=1}^n \right) \rho D_k \left( \{e_l\}_{l=1}^n, \{\theta_l\}_{l=1}^n \right)
\end{equation*}
for arbitrary density operator $\rho$. Therefore the problem has been reduced to, proving $S_k$ stabilizes $|\Phi\rangle_n \langle\Phi|$ for all $k \in [n]$ and $n$, which has been proven by Lemma~\ref{lemma:angle-symmetry}.
\end{proof}

Then, we find a set of $2^{n-1}$ mutually independent Pauli strings.
All elements commute with operator $T_n \left( \{\theta_k\}_{k=1}^n \right)$, so that we can construct projectors from it.

\begin{observation}
    \label{obv:projectors}
  On $n$-qubit quantum system ($n \geq 2$), consider the following set of $n-1$ Pauli operators:
  \begin{equation*}
    \mathcal{Y}_n = \left\{ Q_{k} \coloneqq \sigma_y \otimes I^{\otimes (k-1)} \otimes \sigma_y \otimes I^{\otimes (n-1-k)} \middle| k=1, \cdots, (n-1)\right\}.
  \end{equation*}
  This set generates a subgroup of $n$-qubit Pauli group, which contains $2^{n-1}$ elements.
  Denote this subgroup by $\mathcal{G}_Y$, we have
  \begin{equation*}
    \left[ K \left( \{\theta_k\}_{k=1}^n \right), Q \right] = \left[ W_n \left( \{\theta_k\}_{k=1}^n \right), Q \right] = 0, \qquad \forall Q \in \mathcal{G}_Y.
  \end{equation*}
\end{observation}
To prove this observation, we first establish the following lemma:
\begin{lemma}
    \label{lemma:2}
    {Let $\mathcal{Y}_n$ be the set of Pauli operators defined in Observation~\ref{obv:projectors}, then for all $Q_k \in \mathcal{Y}_n$, we have} 
    \begin{equation*}
        Q_k |\Phi_n^-\rangle = (-1)^{k-1} |\Phi_n^-\rangle, \quad Q_k |\Psi^+_n\rangle = (-1)^{k-1} |\Psi_n^+\rangle,  
    \end{equation*}
    and therefore we have 
    \begin{equation}
        \label{eq:lemma:2}
        Q_k |\Phi_n\rangle = (-1)^{k-1} |\Phi_n\rangle. 
    \end{equation}
\end{lemma}
\begin{proof}
Prove by induction.
For 2-qubit system, we have 
\begin{equation*}
    \begin{aligned}
        &Q_1 |\Phi_n^-\rangle = (\sigma_y \otimes \sigma_y) \frac{1}{\sqrt{2}} (|00\rangle - |11\rangle) = \frac{1}{\sqrt{2}} ( |00\rangle - |11\rangle), \\
        &Q_1 |\Psi_n^+\rangle = (\sigma_y \otimes \sigma_y) \frac{1}{\sqrt{2}} (|01\rangle + |10\rangle) = \frac{1}{\sqrt{2}} (|01\rangle + |10\rangle).
    \end{aligned}
\end{equation*}
Suppose for $(n-1)$-qubit system we have
\begin{equation*}
    Q_k |\Phi_{n-1}^-\rangle = (-1)^{k-1} |\Phi_{n-1}^-\rangle, \quad Q_k |\Psi^+_{n-1}\rangle = (-1)^{k-1} |\Psi_{n-1}^+\rangle
\end{equation*}
for all $Q_k \in \mathcal{Y}_{n-1}$. Then according to the definition of $|\Phi_n^-\rangle$ and $|\Psi_n^+\rangle$, for $Q_k \in \mathcal{Y}_{n-1}$ we have 
\begin{equation*}
    \begin{aligned}
        (Q_k \otimes I) |\Phi_n^-\rangle = &(Q_k \otimes I) \frac{1}{\sqrt{2}} \left( |\Psi_{n-1}^+\rangle \otimes |0\rangle - |\Phi_{n-1}^-\rangle \otimes |1\rangle \right) = \frac{1}{\sqrt{2}} \left( Q_k|\Psi_{n-1}^+\rangle \otimes I|0\rangle - Q_k|\Phi_{n-1}^-\rangle \otimes I|1\rangle \right) \\
        = &\frac{1}{\sqrt{2}} \left( (-1)^{k-1}|\Psi_{n-1}^+\rangle \otimes |0\rangle - (-1)^{k-1}|\Phi_{n-1}^-\rangle \otimes |1\rangle \right) \\
        = &(-1)^{k-1}|\Phi_n^-\rangle, \\
        (Q_k \otimes I) |\Psi_n^+\rangle = &(Q_k \otimes I) \frac{1}{\sqrt{2}} \left( |\Psi_{n-1}^+\rangle \otimes |1\rangle + |\Phi_{n-1}^-\rangle \otimes |0\rangle \right) = \frac{1}{\sqrt{2}} \left( Q_k|\Psi_{n-1}^+\rangle \otimes I|1\rangle + Q_k|\Phi_{n-1}^-\rangle \otimes I|0\rangle \right) \\
        = &\frac{1}{\sqrt{2}} \left( (-1)^{k-1}|\Psi_{n-1}^+\rangle \otimes |1\rangle + (-1)^{k-1}|\Phi_{n-1}^-\rangle \otimes |0\rangle \right) \\
        = &(-1)^{k-1}|\Psi_n^+\rangle.
    \end{aligned}
\end{equation*}
Hence Eq.~\eqref{eq:lemma:2} holds for all $k =1, 2, \cdots, n-1$.  
For $Q_{n-1} = \sigma_y \otimes I^{\otimes (n-2)} \otimes \sigma_y$, first note that 
\begin{equation*}
    \begin{aligned}
        (\sigma_y \otimes I) |\Phi_2^-\rangle =& \frac{1}{\sqrt{2}} (\sigma_y|0\rangle \otimes |0\rangle - \sigma_y|1\rangle \otimes |1\rangle) = \frac{\text{i}}{\sqrt{2}} (|10\rangle + |01\rangle) = \text{i} |\Psi_2^+\rangle, \\
        (\sigma_y \otimes I) |\Psi_2^+\rangle =& \frac{1}{\sqrt{2}} (\sigma_y|0\rangle \otimes |1\rangle + \sigma_y|1\rangle \otimes |0\rangle) = -\frac{\text{i}}{\sqrt{2}} (|00\rangle - |11\rangle) = -\text{i} |\Phi_2^-\rangle. 
    \end{aligned}
\end{equation*}
This implies that 
\begin{equation*}
    \begin{aligned}
        (\sigma_y \otimes I \otimes \sigma_y) |\Phi_3^-\rangle =& \frac{1}{\sqrt{2}} ((\sigma_y \otimes I)|\Psi_2^+\rangle \otimes \sigma_y|0\rangle - (\sigma_y \otimes I)|\Phi_2^-\rangle \otimes \sigma_y|1\rangle) \\
        =& \frac{1}{\sqrt{2}} (-\text{i}|\Phi_2^-\rangle \otimes \text{i}|1\rangle - \text{i}|\Psi_2^+\rangle \otimes (-\text{i})|0\rangle) = -\frac{1}{\sqrt{2}} (|\Psi_2^+\rangle \otimes |0\rangle - |\Phi_2^-\rangle \otimes |1\rangle) = (-1)^{2-1} |\Phi_3^-\rangle \\
        =& -|\Phi_3^-\rangle, \\
        (\sigma_y \otimes I \otimes \sigma_y) |\Psi_3^+\rangle =& \frac{1}{\sqrt{2}} ((\sigma_y \otimes I)|\Psi_2^+\rangle \otimes \sigma_y|1\rangle + (\sigma_y \otimes I)|\Phi_2^-\rangle \otimes \sigma_y|0\rangle) \\
        =& \frac{1}{\sqrt{2}} (-\text{i}|\Phi_2^-\rangle \otimes (-\text{i})|0\rangle + \text{i}|\Psi_2^+\rangle \otimes \text{i}|1\rangle) = -\frac{1}{\sqrt{2}} (|\Psi_2^+\rangle \otimes |1\rangle + |\Phi_2^-\rangle \otimes |0\rangle) = (-1)^{2-1} |\Psi_3^+\rangle \\
        =& -|\Psi_3^+\rangle.
    \end{aligned}
\end{equation*}
Suppose for $(n-1)$-qubit system, we have 
\begin{equation*}
    (\sigma_y \otimes I^{\otimes (n-2)}) |\Phi_{n-1}^-\rangle =  (-1)^{n-1} \text{i}|\Psi^+_{n-1}\rangle, \quad
    (\sigma_y \otimes I^{\otimes (n-2)}) |\Psi_{n-1}^+\rangle =  (-1)^{n} \text{i}|\Phi^-_{n-1}\rangle.
\end{equation*}
Then we have 
\begin{equation*}
    \begin{aligned}
        (\sigma_y \otimes I^{n-1}) |\Phi_{n}^-\rangle =& \frac{1}{\sqrt{2}} ((\sigma_y \otimes I^{n-2})|\Psi_{n-1}^+\rangle \otimes |0\rangle - (\sigma_y \otimes I^{n-2})|\Phi_{n-1}^-\rangle \otimes |1\rangle) \\
        =& \frac{1}{\sqrt{2}} ((-1)^n \text{i}|\Phi_{n-1}^-\rangle \otimes |0\rangle + (-1)^n \text{i} |\Psi_{n-1}^+\rangle \otimes |1\rangle) \\
        =& (-1)^{n} \text{i}|\Psi_n^+\rangle, \\
        (\sigma_y \otimes I^{n-1}) |\Psi_{n}^+\rangle =& \frac{1}{\sqrt{2}} ((\sigma_y \otimes I^{n-2})|\Psi_{n-1}^+\rangle \otimes |1\rangle + (\sigma_y \otimes I^{n-2})|\Phi_{n-1}^-\rangle \otimes |0\rangle) \\
        =& -\frac{1}{\sqrt{2}} (-(-1)^{n+1} \text{i}|\Phi_{n-1}^-\rangle \otimes |1\rangle + (-1)^{n+1} \text{i} |\Psi_{n-1}^+\rangle \otimes |0\rangle) \\
        =& (-1)^{n+1} \text{i}|\Phi_n^-\rangle, 
    \end{aligned}
\end{equation*}
Hence for arbitrary $n$, the following equality holds:
\begin{equation*}
    \begin{aligned}
        (\sigma_y \otimes I^{\otimes (n-2)} \otimes \sigma_y) |\Phi_n^-\rangle =& \frac{1}{\sqrt{2}} ((\sigma_y \otimes I^{\otimes (n-2)})|\Psi_{n-1}^+\rangle \otimes \sigma_y|0\rangle - (\sigma_y \otimes I^{\otimes (n-2)})|\Phi_{n-1}^-\rangle \otimes \sigma_y|1\rangle) \\
        =& \frac{1}{\sqrt{2}} ((-1)^n \text{i} |\Phi_{n-1}^-\rangle \otimes \text{i}|1\rangle - (-1)^{n-1} \text{i} |\Psi_{n-1}^+\rangle \otimes (-\text{i})|0\rangle) \\
        =& \frac{(-1)^n}{\sqrt{2}} ( |\Psi_{n-1}^+\rangle \otimes |0\rangle - |\Phi_{n-1}^-\rangle \otimes |1\rangle ) = (-1)^{n} |\Phi_n^-\rangle \\
        =& (-1)^{k-1} |\Phi_n^-\rangle, \\
        (\sigma_y \otimes I^{\otimes (n-2)} \otimes \sigma_y) |\Psi_n^+\rangle =& \frac{1}{\sqrt{2}} ((\sigma_y \otimes I^{\otimes (n-2)})|\Psi_{n-1}^+\rangle \otimes \sigma_y|1\rangle + (\sigma_y \otimes I^{\otimes (n-2)})|\Phi_{n-1}^-\rangle \otimes |0\rangle) \\
        =& \frac{1}{\sqrt{2}} ((-1)^n \text{i}|\Phi_{n-1}^-\rangle \otimes (-\text{i})|0\rangle + (-1)^{n-1} \text{i}|\Psi_{n-1}^+\rangle \otimes \text{i}|0\rangle) \\
        =& \frac{(-1)^n}{\sqrt{2}} (|\Psi_{n-1}^+\rangle \otimes |0\rangle + |\Phi_{n-1}^-\rangle \otimes |0\rangle) = (-1)^n |\Psi_{n}^+\rangle \\
        =& (-1)^{k-1} |\Psi_{n}^+\rangle. 
    \end{aligned}
\end{equation*}

Therefore, according to the definition of $|\Phi_n\rangle$, Eq.~\eqref{eq:lemma:2} also holds for $k=n-1$.
\end{proof}

\begin{proof}[Proof of Observation~\ref{obv:projectors}]
In \cite{Kaniewski_2016}, it has been proven on 2-qubit and 3-qubit systems that $\left[ T_n (\{\theta_k\}_{k=1}^n), Q \right] = 0$ for all $Q \in \mathcal{G}_{Y_n}$ and $\{\theta_{k}\}_{k=1}^n \in [0, \pi/4]$.
Here we prove the general case by induction. We first prove that $\left[ K(\{\theta_k\}_{k=1}^n), Q \right] = 0$ for all $Q \in \mathcal{G}_{Y_n}$. Since $Q$ is unitary and Hermitian, it suffices to prove $QK(\{\theta_k\}_{k=1}^n)Q = K(\{\theta_k\}_{k=1}^n)$. Now consider the Kraus representation of the $n$-qubit dephasing channel $\Lambda_n \left( \{ \theta_{k} \}_{k=1}^n \right)$: 
\begin{equation*}
    \left[ \Lambda_n \left( \{ \theta_{k} \}_{k=1}^n \right) \right] (\rho) = \sum_{e_1 \cdots e_n \in \{0, 1\}^{n}} \left( \left( \prod_{k=1}^n \frac{1 + (-1)^{e_k} g(\theta_k)}{2} \right) \left( \bigotimes_{k=1}^n \Gamma^{e_k} (\theta_k) \right) \rho \left( \bigotimes_{k=1}^n \Gamma^{e_k} (\theta_k) \right) \right),  
\end{equation*}
where each Kraus operator $\otimes_{k=1}^n \left( \Gamma^{e_k} (\theta_k) \right)$ is a Pauli operator, and hence either commutes or anti-commutes with $Q$. This implies that
\begin{equation*}
    \begin{aligned}
        &Q \left[ \Lambda_n \left( \{ \theta_{k} \}_{k=1}^n \right) \right] (\rho) Q \\
        =&  \sum_{e_1 \cdots e_n \in \{0, 1\}^{n}} Q\left( \left( \prod_{k=1}^n \frac{1 + (-1)^{e_k} g(\theta_k)}{2} \right) \left( \bigotimes_{k=1}^n \Gamma^{e_k} (\theta_k) \right) \rho \left( \bigotimes_{k=1}^n \Gamma^{e_k} (\theta_k) \right) \right) Q \\
        =& \sum_{e_1 \cdots e_n \in \{0, 1\}^{n}} \left( \left( \prod_{k=1}^n \frac{1 + (-1)^{e_k} g(\theta_k)}{2} \right) Q \left( \bigotimes_{k=1}^n \Gamma^{e_k} (\theta_k) \right) \rho \left( \bigotimes_{k=1}^n \Gamma^{e_k} (\theta_k) \right) Q \right) \\
        =& \sum_{e_1 \cdots e_n \in \{0, 1\}^{n}} \left( \left( \prod_{k=1}^n \frac{1 + (-1)^{e_k} g(\theta_k)}{2} \right) \left( \bigotimes_{k=1}^n \Gamma^{e_k} (\theta_k) \right) Q \rho Q \left( \bigotimes_{k=1}^n \Gamma^{e_k} (\theta_k) \right) \right) \\
        =& \left[ \Lambda_n \left( \{\theta_k\}_{k=1}^n \right) \right] (Q \rho Q).
    \end{aligned}
\end{equation*}
Therefore, we only need to prove $[|\Phi\rangle_n \langle\Phi|, Q] = 0$, which is supported by Lemma~\ref{lemma:2}.

Lastly we prove that $\left[ W_n \left( \{ \theta_k \}_{k=1}^n \right), Q \right] = 0$. This can be similarly proven by the definition of $W_n$. $W_n$ is linear combination of multiple operators with each position being either Pauli-$X$ or Pauli-$Z$, which implies that it commutes with a Pauli operator with even number of Pauli-$Y$s (and the remaining positions being all identity $I$s.). Therefore, $W_n$ also commutes with $Q$ for all $Q \in \mathcal{G}_{Y_n}$, which completes the proof.
\end{proof}

We'll use the following projectors to project the inequality onto $2^{n-1}$ disjoint 2-dimensional subspaces:

\begin{equation}
    \label{eq:s:2d-projector}
  P_{\vec{a}} \coloneqq \prod_{k=1}^{n-1} \frac{I^{\otimes n} + (-1)^{a_k} Q_k}{2},
\end{equation}
where $\vec{a} = \{a_k\}_{k=1}^{n-1} \in \{0, 1\}^{n-1}$ is a binary string as index of a projector.
These projectors are $2$-dimensional, because $Q_k$ are independent and orthogonal.
Also, since all operators in $\mathcal{Y}_n$ commute with $T_n\left( \{\theta_k\}_{k=1}^n \right)$, we can prove $T_n\left( \{\theta_k\}_{k=1}^n \right) \geq 0$ by establishing the following operator inequality:
\begin{equation*}
  M_{\vec{a}} \left( \{\theta_k\}_{k=1}^n \right) \coloneqq P_{\vec{a}} T_n \left( \{\theta_k\}_{k=1}^n \right) P_{\vec{a}} \geq 0
\end{equation*}
for all $\vec{a} \in \{0, 1\}^{n-1}$ and $\{\theta_k\}_{k=1}^n \in \left[ 0, \pi/4 \right]^n$.

To compute analytic expression of $M_{\vec{a}} \left( \{\theta_k\}_{k=1}^n \right)$, we use generalized Bloch representation of the dephased state $K\left( \{\theta_k\}_{k=1}^n \right)$, which is derived from the generalized Bloch representation of $|\Phi\rangle_n$.

\begin{observation}[Projector-based representation of $n$-qubit maximally entangled state]
  \label{obv:s:proj-based}
  For integer $n \geq 3$, density operator $|\Phi\rangle_n\langle\Phi|$ can be written
  \begin{equation}
    \label{eq:proj-state}
    |\Phi\rangle_n \langle\Phi| = \prod_{O_k \in \mathcal{O}_n} \frac{I^{\otimes n} + O_k}{2}.
  \end{equation}
  Here $\mathcal{O}_n$ is a set of $n$ Hermitian unitary operators:
  \begin{equation*}
    \mathcal{O}_n \coloneqq \mathcal{Y}_{n}^\mathrm{AL} \cup \left\{ O_1 \coloneqq H \otimes \sigma_x^{\otimes (n-1)} \right\},
  \end{equation*}
  where
  \begin{equation*}
    \mathcal{Y}_{n}^\mathrm{AL} = \left\{ O_k \coloneqq (-1)^k \otimes I^{\otimes (k-2)} \otimes \sigma_y \otimes I^{\otimes (n-k)} \middle| k=2, \cdots, n \right\}.
  \end{equation*}
  Also, the set $\mathcal{O}_n$ satisfies:
  \begin{enumerate}
    \item $\forall O_k \in \mathcal{O}_n, \text{tr}[O_k] = 0$;
    \item $\forall O_k \in \mathcal{O}_n, O_k |\Phi\rangle_n = |\Phi\rangle_n$;
    \item $\forall O_k, O_l \in \mathcal{O}_n$ and $O_k \neq O_l$, we have $[O_k, O_l] = 0$.
  \end{enumerate}
\end{observation}
\begin{proof}
    It is easy to verify that $\mathcal{Y}_{n}^{\mathrm{AL}}$ is a set of independent Pauli operators, since each $O_k$ completely decides the exponent of Pauli-$Y$ on $k$-th qubit.
    Also, note that $\mathrm{tr}[O_1^\dagger O_k] = 0$ for any $k=2, \cdots, n$, which can be simply derived from $\mathrm{tr}[YH]=0$. Also, $[O_1^\dagger, O_k] = 0$ for any $k=2, \cdots, n$ since Pauli-$Y$ anti-commutes with $H$ and Pauli-$X$.
    Lastly, according to Lemma~\ref{lemma:angle-symmetry} and \ref{lemma:2}, $O_k$ stabilizes $|\Phi_n\rangle$ for all $k=1, \cdots, n$.
    Therefore, we can use Eq.~\eqref{eq:proj-state} to determine the density operator $|\Phi\rangle_n \langle \Phi|$.
\end{proof}

Now we can calculate the generalized Bloch representation of $|\Phi\rangle_n$.

\begin{observation}
  \label{obv:s:bloch-coef}
  Denote the generalized Bloch representation of $|\Phi\rangle_n$ by
  \begin{equation*}
    |\Phi\rangle_n\langle\Phi| = \frac{1}{2^n} \sum_{k=0}^{4^n - 1} r_k P_k.
  \end{equation*}
  Denote the set of indices of those $r_k \neq 0$ by
  \begin{equation*}
    \mathfrak{K} \coloneqq \left\{ k \in \{0, 1, 2, \cdots, 4^n - 1\} \middle| r_k \neq 0 \right\}.
  \end{equation*}
  Then we have
  \begin{equation*}
    P_{\mathfrak{K}} \coloneqq \{P_k|k \in \mathfrak{K}\} = \mathcal{G}_Y \cup \mathcal{P}_{XZ}.
  \end{equation*}
  Here, $\mathcal{P}_XZ$ is the set of Pauli operators that are tensor products of $\sigma_x$ and $\sigma_z$ only.
  That is,
  \begin{equation*}
    \mathcal{P}_{XZ} = \left\{ P = \bigotimes_{k=1}^n \left( \frac{\sigma_x + \sigma_z}{2} + (-1)^{s_k} \frac{\sigma_x - \sigma_z}{2} \right) \middle| \vec{s} \coloneqq \{s_k\}_{k=1}^n \in \{0, 1\}^n \right\}.
  \end{equation*}
  The analytic expression of Bloch coefficients are presented in the full proof.
\end{observation}

\begin{figure}[htbp]
    \centering
    \includegraphics[width=0.8\textwidth]{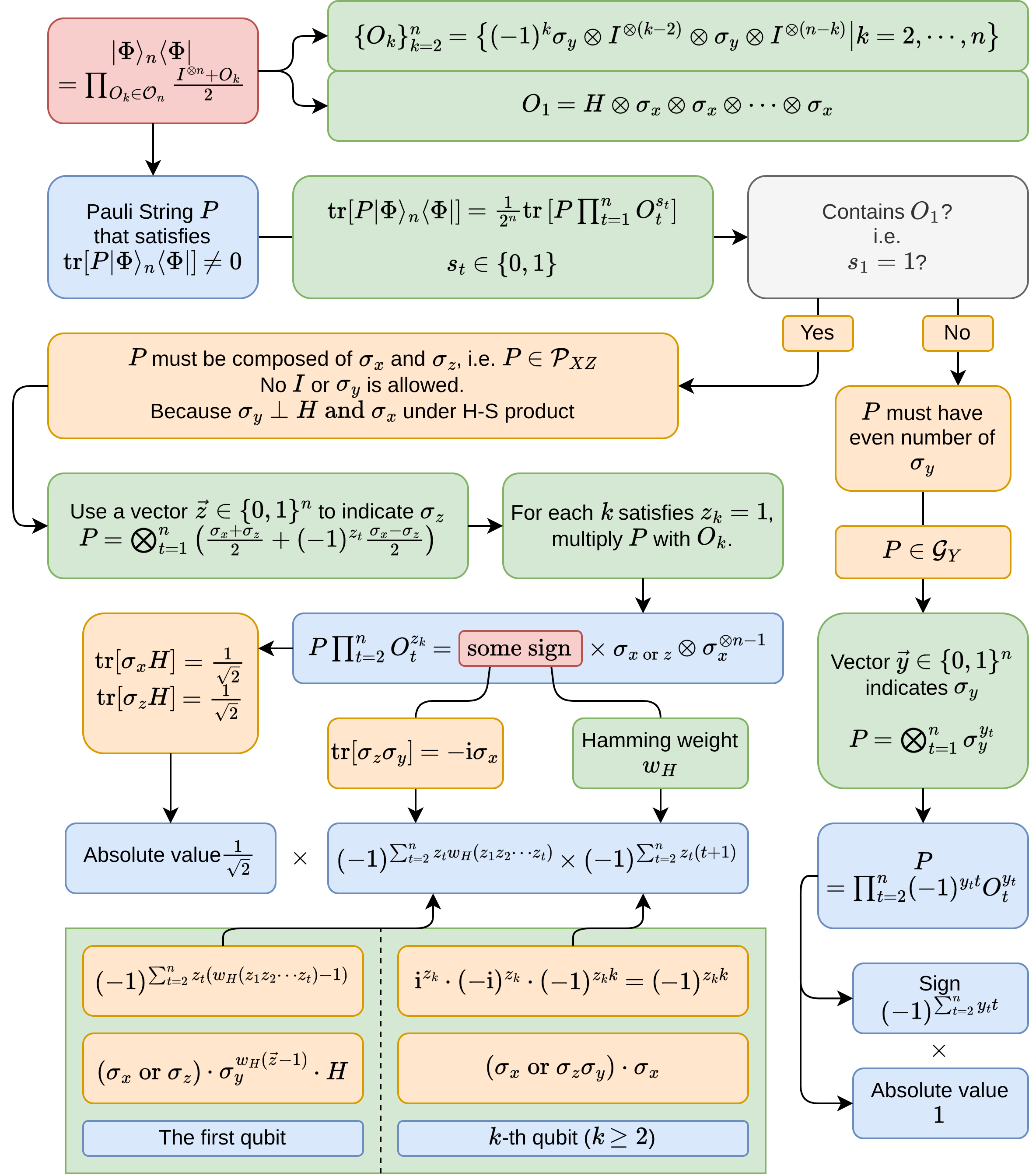}
    \caption{Sketch of proof of Observation~\protect\ref{obv:s:bloch-coef}.}
    \label{fig:sketch-bloch-coef}
\end{figure}

\begin{proof}
Let $P_k \in \mathcal{P}_n$ be an arbitrary Pauli operator.
Then the proposition is equivalent to: Case 1, if $P_k$ takes Pauli-$X$ or Pauli-$Z$ being at some qubit, then the rest of all position should also be Pauli-$X$ or Pauli-$Z$, otherwise $r_k = 0$;
Case 2, if $P_k$ takes Pauli-$Y$ on some qubit, then it must has even number Pauli-$Y$s and the rest of all positions are all identities, otherwise $r_k = 0$.
We start from the former case.

Suppose $P_k$ takes Pauli-$X$ on $l$-th qubit.
Consider the expansion of the HS-inner product: 
\begin{equation*}
    \mathrm{tr}[P_k |\Phi\rangle_n \langle\Phi|] = \sum_{\vec{s} \in \{0, 1\}^n} \mathrm{tr} \left[ P_k \prod_{t=1}^n O_t^{s_t} \right]. 
\end{equation*}
Here $\vec{s} = (s_1 s_2 \cdots s_n)$, and for some operator $O$ we define $O^0 = I$ and $O^1 = O$.
Use the following set to filter the product with non-zero trace: 
\begin{equation*}
    \mathscr{S}_k := \left\{ \vec{s} \in \{0, 1\}^n \middle| \mathrm{tr} \left[ P_k \prod_{t=1}^n O_t^{s_t} \right] \neq 0 \right\}.
\end{equation*}
Assert: If $\mathscr{S}_k \neq \emptyset$, then for all $\vec{s} \in \mathscr{S}_k, s_1 = 1$.
Otherwise, consider the $l$-th position of the product $P_k \prod_{t=1}^{n} O_t^{s_t}$, it must be the product of multiple Pauli-$Y$ and one Pauli-$X$, which is traceless.
This contradicts to the definition of $\mathscr{S}_k$.
Now assume on some qubit $l_0 \neq l$, $P_k$ takes Pauli-$Y$ or identity $I$ on corresponding position, then the product $P_k \prod_{t=1}^{n} O_t^{s_t}$ is either product of multiple Pauli-$Y$s and a Hadamard gate $H$ ($l_0 = 1$) or product of multiple Pauli-$Y$s and a Pauli-$X$ ($l_0 \neq 1$).
In both cases the product is traceless, which contradicts to the definition of $\mathscr{S}_k$.
Hence on every position $P_k$ must take Pauli-$X$ or Pauli-$Z$, i.e. $P_k \in \mathcal{P}_{XZ}$.
Similarly we can prove the case that $P_k$ has Pauli-$Z$ on its $l$-th qubit.

To actually compute $r_k$ for $P_k \in \mathcal{P}_{XZ}$, we establish a deterministic construction of $\vec{s}$, and prove that $\mathscr{S}_k = \{\vec{s}\}$.
Note that $\mathrm{tr}[\sigma_x H] = \mathrm{tr}[\sigma_z H] = \sqrt{2}$, and the idea is to use operators in $\mathcal{Y}_{n}^{\mathrm{AL}}$ to transform each Pauli-$Z$ to Pauli-$X$ on each qubit except the first one.
Since each element in $\mathcal{Y}_n$ is independent, there will be one unique subset of $\mathcal{Y}_n$ such that the product $P_k \prod_{t=2}^{n} O_t^{s_t}$ takes Pauli-$X$ on all qubits except the first one, which implies the final product $P_k \prod_{t=1}^{n} O_t^{s_t}$ has trace of absolute value $2^{n-\frac{1}{2}}$. 

To compute the sign, for $P_k \in \mathcal{P}_{XZ}$, define a bit string $\vec{z} = (z_1 \cdots z_n) \in \{0, 1\}^n$, where each bit $z_t$ takes $1$ if and only if $P_k$ takes Pauli-$Z$ on the $t$-th qubit. 
For two natural numbers $a \leq b$, denote the slice of $\vec{z}$ from $z_a$ to $z_b$ by $\vec{z}_{a,b} = (z_a z_{a+1} \cdots z_b)$.
Let $w_H(\vec{z}_{a,b})$ be the Hamming weight (number of $1$s in this string) of $\vec{z}_{a, b}$, then we have 
\begin{equation*}
    \begin{aligned}
        &P_k \prod_{t=2}^{n} O_t^{z_t} \\
        =& \left( \prod_{t=2}^n (-1)^{z_t z_1} (-1)^{z_t (w_H(\vec{z}_{2, t}) - 1)} i^{z_t} (-i)^{z_t} (-1)^{z_t t} \right) \left( \frac{\sigma_x + \sigma_z}{2} + (-1)^{z_1 + w_H(\vec{z}_{2, n})} \frac{\sigma_x - \sigma_z}{2} \right) \otimes \sigma_x^{\otimes (n-1)} \\
        =& \left( \prod_{t=2}^n (-1)^{z_t (w_H(\vec{z}_{1, t}) - 1)} (-1)^{z_t t} \right) \left( \frac{\sigma_x + \sigma_z}{2} + (-1)^{w_H(\vec{z})} \frac{\sigma_x - \sigma_z}{2} \right) \otimes \sigma_x^{\otimes (n-1)} \\
        =& \left( (-1)^{\sum_{t=2}^n z_t (w_H(\vec{z}_{1, t})-1)} (-1)^{\sum_{t=2}^n z_t t} \right) \left( \frac{\sigma_x + \sigma_z}{2} + (-1)^{w_H(\vec{z})} \frac{\sigma_x - \sigma_z}{2} \right) \otimes \sigma_x^{\otimes (n-1)} \\ 
        =& \left\{ \begin{matrix*}[l]
            \left( (-1)^{\# \{t \in [n] - \{1\} | t=1\ \mathrm{mod}\ 2, z_t = 1\}} (-1)^{\left\lceil \frac{w_{H}(\vec{z}_{2, n})}{2} \right\rceil} \right) \left( \frac{\sigma_x + \sigma_z}{2} + (-1)^{w_H(\vec{z})} \frac{\sigma_x - \sigma_z}{2} \right) \otimes \sigma_x^{\otimes (n-1)}, \quad & n\ \mathrm{odd}, \\
            \left( (-1)^{\# \{t \in [n] - \{1\} | t=1\ \mathrm{mod}\ 2, z_t = 1\}} (-1)^{\left\lfloor \frac{w_{H}(\vec{z}_{2, n})}{2} \right\rfloor} \right) \left( \frac{\sigma_x + \sigma_z}{2} + (-1)^{w_H(\vec{z})} \frac{\sigma_x - \sigma_z}{2} \right) \otimes \sigma_x^{\otimes (n-1)}, \quad & n\ \mathrm{even}. 
        \end{matrix*} \right.
    \end{aligned}
\end{equation*}
Here $\#$ counts the number of elements in a set.
It is not hard to see that $\vec{s} = 1z_2z_3\cdots$ is the only one element in $\mathscr{S}_k$.
The uniqueness comes from the independency of $\mathcal{O}_n$. 
Therefore, we have 
\begin{equation}
    \mathrm{tr} \left[ P_k O_1 \prod_{t=2}^{n} O_t^{z_t} \right] = \frac{1}{\sqrt{2}} (-1)^{\sum_{t=2}^n z_t (w_H(\vec{z}_{1, t})-1)} (-1)^{\sum_{t=2}^n z_t t} = \frac{1}{\sqrt{2}} (-1)^{\sum_{t=2}^n z_t w_H(\vec{z}_{1, t})} (-1)^{\sum_{t=2}^n z_t (t + 1)}. 
    \label{eq:90}
\end{equation}

Next, we consider $P_k$ being a Pauli operator takes Pauli-$Y$ on some position $l$.
Then obviously it cannot take Pauli-$X$ or Pauli-$Z$ on any other positions, otherwise it should be in $\mathcal{P}_{XZ}$ according to discussions above, contradiction.
Hence $P_k$ must be tensor product of multiple Pauli-$Y$s and identity $I$s.
Then for all $\vec{s} \in \mathscr{S}_k$, $s_0 = 0$. Otherwise we have demonstrated that the product will be traceless on qubits that $P_k$ takes Pauli-$Y$.
Now we have $\prod_{t=1}^{n} O_t^{s_t} \in \mathcal{G}_{Y}$, which implies that $P_k \in \mathcal{G}_{Y}$.
Otherwise there will be Pauli-$Y$ that cannot be canceled, which leads to zero trace.

To show $P_k \in P_{\mathfrak{K}}$, we also calculate the Bloch coefficient of each $P_k \in \mathcal{G}_Y$.
Similar to $\mathcal{P}_{XZ}$, define a bit string $\vec{y} \in \{0, 1\}^n$, where each bit $y_t$ takes $1$ if and only if $P_k$ takes Pauli-$Y$ on the $t$-th qubit.
Then obviously $\vec{s} = 0y_2y_3 \cdots y_n$ is the only one element in $\mathscr{S}_k$, since the representation of any $P_k \in \mathcal{G}_Y$ by $\mathcal{Y}_{n}^\mathrm{AL}$ is unique. In conclusion, we have 
\begin{equation*}
    \mathrm{tr}\left[ P_k \prod_{t=2}^n O_t^{y_t}\right] = \mathrm{tr} \left[ (-1)^{\sum_{t=2}^n y_t t} I^{\otimes n} \right] = (-1)^{\sum_{t=2}^n y_t t} 2^n. 
\end{equation*}
\end{proof}

The closed form of generalized Bloch representation of dephased state $K\left( \{\theta_k\}_{k=1}^n \right)$ can be easily established now.

\begin{observation}
  \label{obv:s:dephased-bloch}
  Let $|\Phi\rangle_n \langle\Phi| = 2^{-n} \sum_k r_k P_k$ be the generalized Bloch representation of $|\Phi\rangle_n$, the dephased state $K \left( \{\theta\}_{k=1}^n \right)$ defined in Proposition~\ref{prop:main} can be written
  \begin{equation*}
    K \left( \{\theta\}_{k=1}^n \right) = 2^{-n} \sum_{P = \bigotimes_{k=1}^n P_k \in \mathcal{P}_n} \left( \bigotimes_{k=1}^n r_k g(\theta_k)^{\mathbf{1}\{[\sigma_x, P_k] \neq 0\}} P_k \right),
  \end{equation*}
  where $\mathcal{P}_n$ is the set of all $n$-qubit Pauli strings, $\mathbf{1}\{\cdot\}$ is an indicator function defined by
  \begin{equation*}
    \mathbf{1}\{[\sigma_x, P_k] \neq 0\} = \begin{cases}
      1 &\text{\ if\ } [\sigma_x, P_k] \neq 0, \\
      0 &\text{\ if\ } [\sigma_x, P_k] = 0.
    \end{cases}
  \end{equation*}
\end{observation}

\begin{figure}[htbp]
    \centering
    \includegraphics[width=0.5\textwidth]{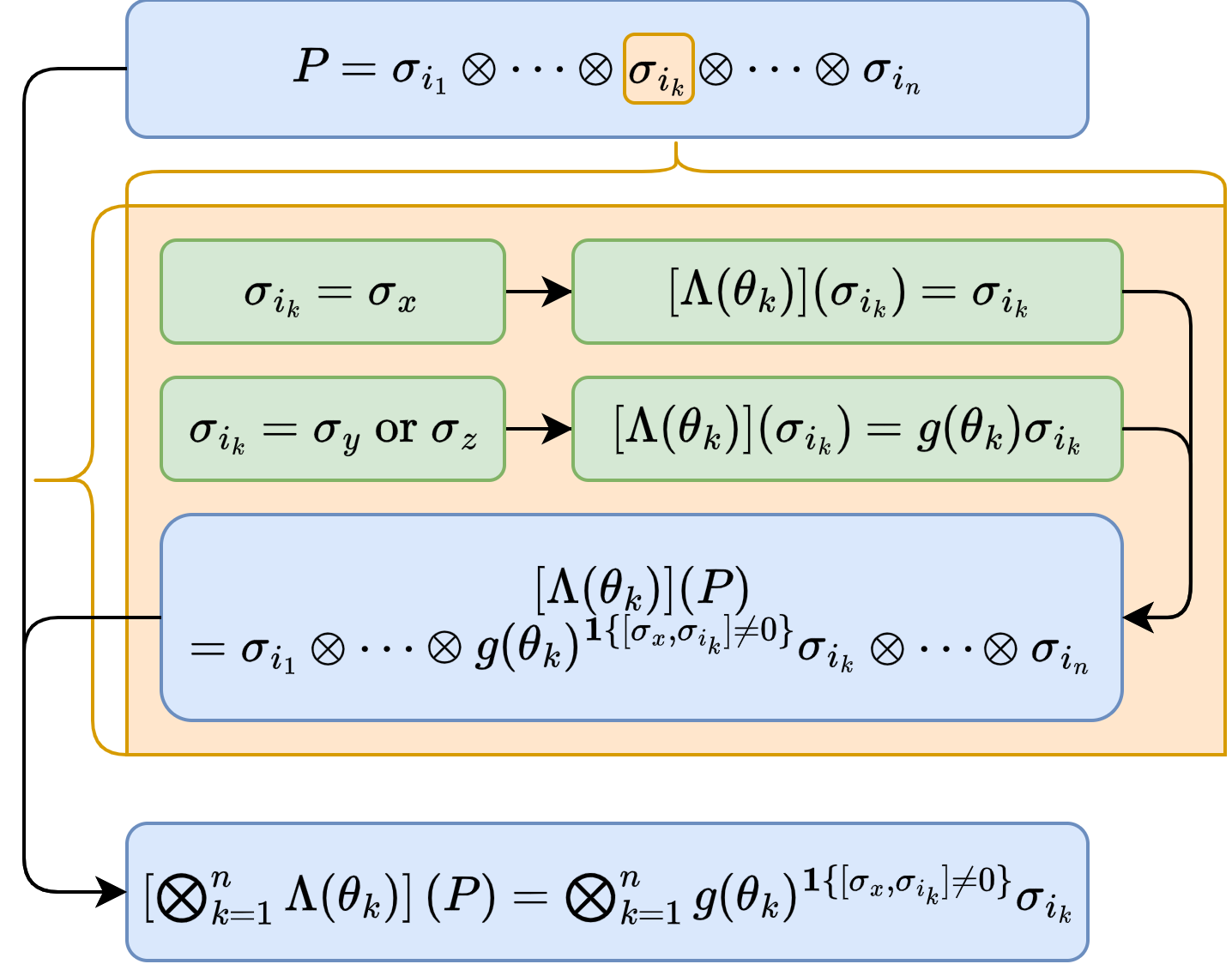}
    \caption{Sketch of proof of Observation~\protect\ref{obv:s:dephased-bloch}.}
    \label{fig:sketch-dephased-bloch}
\end{figure}

\begin{proof}
By definition of qubit dephasing channel, for single qubit Pauli operator $P \in \{\sigma_x, \sigma_y, \sigma_z\}$, we have 
\begin{equation*}
    \begin{aligned}
        [\Lambda(\theta)] (P) = &\frac{1 + g(\theta)}{2} P + \frac{1 - g(\theta)}{2} \Gamma(\theta) P \Gamma(\theta) \\
        =& \left( \frac{1 + g(\theta)}{2} P + (-1)^{\mathbf{1}\{[\Gamma(\theta), P] \neq 0\}} \frac{1 - g(\theta)}{2} \right) P = g(\theta)^{\mathbf{1}\{[\Gamma(\theta), P] \neq 0\}} P.
    \end{aligned}
\end{equation*}
By linearity of quantum channels and the completeness and orthogonality of Pauli basis, the proof is completed.
\end{proof}

\subsection{Calculation of Positivity}

Since $M_{\vec{a}} (\{\theta_k\}_{k=1}^n)$ is of rank at most $2$, we can prove it is PSD by the positivity:
\begin{equation}
  \label{eq:positivity}
  \lambda_{\vec{a}}\left( \{\theta_k\}_{k=1}^n \right) \coloneqq \left( \text{tr} \left[ M_{\vec{a}}\left( \{\theta_k\}_{k=1}^n \right) \right] \right)^2 - \text{tr} \left[ M_{\vec{a}}^2\left( \{\theta_k\}_{k=1}^n \right) \right].
\end{equation}
Then $\lambda_{\vec{a}}\left( \{\theta_k\}_{k=1}^n \right) \geq 0$ implies that $M_{\vec{a}} \left( \{\theta_k\}_{k=1}^n \right) \geq 0$, if additionally $\text{tr}\left[M_{\vec{a}}\left( \{\theta_k\}_{k=1}^n \right) \right] \geq 0$ for all $\{\theta_k\}_{k=1}^n \in [0, \pi/4]^n$.

\begin{lemma}
  \label{lemma:positive-trace}
    The following inequality holds for all $\mu \leq 2^{-(n-1)} - 2^{-1}$:
  \begin{equation*}
    \text{tr}\left[M_{\vec{a}}\left( \{\theta_k\}_{k=1}^n \right) \right] \geq 0, \qquad \forall \{\theta_k\}_{k=1}^n \in \left[0, \frac{\pi}{4}\right]^n.
  \end{equation*}
\end{lemma}

\begin{proof}
Only need to consider trace of $K\left(\left\{ \theta_k \right\}_{k=1}^n \right)$ under projection $P_{\vec{a}}$. We have
\begin{equation*}
    \begin{aligned}
        & \mathrm{tr} \left[ M_{\vec{a}} \right] = \mathrm{tr} \left[ P_{\vec{a}} T\left(\left\{ \theta_k \right\}_{k=1}^n \right) P_{\vec{a}} \right] \\ 
        =& \mathrm{tr} \left[ P_{\vec{a}} K\left(\left\{ \theta_k \right\}_{k=1}^n \right) P_{\vec{a}} \right] - \mathrm{tr} \left[ s_n W_n\left(\left\{ \theta_k \right\}_{k=1}^n \right) \right] - \mathrm{tr} \left[ P_{\vec{a}} \mu I^{\otimes n} P_{\vec{a}} \right] \\
        =& \mathrm{tr} \left[ P_{\vec{a}} K\left(\left\{ \theta_k \right\}_{k=1}^n \right) P_{\vec{a}} \right] - 0 - \mathrm{tr}[I^{\otimes n}] \cdot \frac{\mathrm{tr}[P_{\vec{a}}]}{\mathrm{tr}[I^{\otimes n}]} \cdot \mu.  
    \end{aligned}
\end{equation*}
To compute the first term, for each $Q \in \mathcal{G}_{Y_n}$,
we index it with an $n$-bit binary string $\vec{y} = y_1 y_2 \cdots y_n$ with even Hamming weight (parity $=0$), that is, $Q_{\vec{y}}$.
For each $k \geq 2$, $y_k = 1$ if and only if $Q$ takes Pauli-$Y$ on this position, and $y_1 = w_H(\vec{y}_{2, n})\text{\ mod\ }2$.
Then by definition of projector $P_{\vec{a}}$ in Eq.~\eqref{eq:s:2d-projector}, we have 
\begin{equation*}
    \begin{aligned}
        &\mathrm{tr} \left[ Q_{\vec{y}} P_{\vec{a}} \right] = \frac{1}{2^{n-1}} \mathrm{tr} \left[ \left( \prod_{l=2}^{n} Q_{l-1}^{y_l} \right) \left( \sum_{(b_1\cdots b_{n-1}) \in \{0, 1\}^{n-1}} \prod_{k=1}^{n-1} \left( (-1)^{a_k} Q_k \right)^{b_k} \right) \right] \\
        =& \frac{1}{2^{n-1}} \left( \prod_{k=1}^{n-1} \mathbf{1} \{y_{k+1} = b_k\} \right) \sum_{(b_1\cdots b_{n-1}) \in \{0, 1\}^{n-1}} \left( \mathrm{tr}[I^{\otimes n}] (-1)^{\sum_{k=1}^{n-1}a_k b_{k+1}} \right) \\
        =& 2 \cdot (-1)^{\sum_{k=1}^{n-1}a_k y_{k+1}}.
    \end{aligned}
\end{equation*}
Here the indicator function $\mathbf{1}\{\}$ is defined by
\begin{equation*}
    \mathbf{1}\{y_{k+1} = b_k\} = \begin{cases}
      1 &\text{\ if\ } y_{k+1} = b_k, \\
      0 &\text{\ if\ } y_{k+1} \neq b_k.
    \end{cases}
\end{equation*}
According to the projector-based representation of $|\Phi_n\rangle$ in Observation~\ref{obv:s:proj-based} and the Bloch representation of dephased state in Observation~\ref{obv:s:bloch-coef}, the sign of the Bloch coefficient of $|\Phi_n\rangle$ on $Q_{\vec{y}}$ is decided by the number of 1s on odd positions.
That is,
\begin{equation*}
    \mathrm{tr}\left[ |\Phi\rangle_n \langle \Phi_n| Q_{\vec{y}} \right] = (-1)^{\sum_{t=2}^n y_t t}.
\end{equation*}
Also, by Observation~\ref{obv:s:dephased-bloch}, after the dephasing channel, the absolute value of the Bloch coefficient will be multiplied with the corresponding $g^{y_{k}(\theta_k)}$ for all positions of Pauli-$Y$ in $Q_{\vec{y}}$:
\begin{equation*}
    \mathrm{tr}\left[ K\left(\{ \theta_k \}_{k=1}^n\right) Q_{\vec{y}} \right]
    = (-1)^{\sum_{t=2}^n y_t t} \prod_{k=1}^{n} g^{y_{k}} (\theta_k).
\end{equation*}
Then we have 
\begin{equation}
    \begin{aligned}
        \mathrm{tr} \left[ K\left(\{ \theta_k \}_{k=1}^n\right) P_{\vec{a}} \right] = &\frac{1}{2^{n-1}} \sum_{\substack{\vec{y} \in \{0, 1\}^n \\
        w_H(\vec{y}) = 0\text{\ mod\ 2}}} (-1)^{\sum_{k=2}^{n}a_{k-1} y_k} (-1)^{\sum_{k=2}^{n}y_k k} \prod_{k=1}^{n} g^{y_k} (\theta_{k}) \\
        = &\frac{1}{2^{n-1}} \sum_{\substack{\vec{y} \in \{0, 1\}^n \\
        w_H(\vec{y}) = 0\text{\ mod\ 2}}} (-1)^{\sum_{k=2}^{n} y_k (a_{k-1} + k)} \prod_{k=1}^{n} g^{y_k} (\theta_{k}).
    \end{aligned}
    \label{eq:91}
\end{equation}
Finally, 
\begin{align}
        \label{eq:s:linear-trace}
        \mathrm{tr} \left[ M_{\vec{a}} \right] = &\frac{1}{2^{n-1}} \sum_{\substack{\vec{y} \in \{0, 1\}^n \\
        w_H(\vec{y}) = 0\text{\ mod\ 2}}} (-1)^{\sum_{k=2}^{n} y_k (a_{k-1} + k)} \prod_{k=1}^{n} g^{y_k} (\theta_{k}) - 2\mu \\
        \geq &\frac{1}{2^{n-1}} - 2\mu - (1 - \frac{1}{2^{n-1}}) = \frac{1}{2^{n-2}} - 1 - 2\mu \geq 0. \nonumber
\end{align}
\end{proof}

Calculating positivity in Eq.~\eqref{eq:positivity} includes finding trace of cubic terms in $M_{\vec{a}}^2$, for which the following observation will be repeatedly used.

\begin{observation}
  \label{obv:s:cubic-trace}
  Let $\rho$ and $\sigma$ be two Hermitian operators with the following Bloch representation:
  \begin{equation*}
    \begin{aligned}
      \rho =& \frac{1}{2^n} \sum_{\substack{\vec{y} \in \{0, 1\}^n \\ w_H (\vec{y}) = 0\text{\ mod\ }2}} r_{\vec{y}} Q_{\vec{y}} + \frac{1}{2^n} \sum_{\vec{x} \in \{0, 1\}^n} r_{\vec{x}} R_{\vec{x}}, \\
      \sigma =& \frac{1}{2^n} \sum_{\substack{\vec{y} \in \{0, 1\}^n \\ w_H (\vec{y}) = 0\text{\ mod\ }2}} s_{\vec{y}} Q_{\vec{y}} + \frac{1}{2^n} \sum_{\vec{x} \in \{0, 1\}^n} s_{\vec{x}} R_{\vec{x}},
    \end{aligned}
  \end{equation*}
  where
  \begin{equation*}
    \begin{aligned}
      &\vec{y} = \left( y_1 y_2 \cdots y_n \right), \qquad \vec{x} = \left( x_1 x_2 \cdots x_n \right), \\
      &Q_{\vec{y}} = \prod_{k=2}^n Q_{k-1}^{y_k} \in \mathcal{Y}_n, \quad R_{\vec{x}} = \prod_{k=1}^n \left( \frac{\sigma_x + \sigma_z}{2} + (-1)^{x_k} \frac{\sigma_x - \sigma_z}{2} \right) \in \mathcal{P}_{XZ},
    \end{aligned}
  \end{equation*}
  and $w_H(\vec{y}) \coloneqq \sum_{k=1}^n y_k$ is the Hamming weight of $\vec{y}$, $w_H(\vec{y})=0$ mod 2 implies even parity.
  
  Then we have
  \begin{equation}
    \label{eq:s:cubic-trace}
    \text{tr} \left[ P_{\vec{a}} \rho \sigma \right] = \frac{1}{2^{2n-1}} \left\{ \begin{matrix*}[l]
      & \sum_{\substack{\vec{y} \in \{0, 1\}^n \\ w_H(\vec{y}) = 0 \text{\ mod\ }2}} \sum_{\substack{\vec{y}_1 \in \{0, 1\}^n \\ w_H(\vec{y}_1) = 0 \text{\ mod\ }2}} (-1)^{\sum_{k=1}^{n-1} a_k y_{k+1}} r_{\vec{y}_1} s_{\vec{y} \oplus \vec{y_1}} \\
      +& \sum_{\substack{\vec{y} \in \{0, 1\}^n \\ w_H(\vec{y}) = 0 \text{\ mod\ }2}} \sum_{\vec{x} \in \{0, 1\}^n} (-1)^{\sum_{k=1}^{n-1} a_k y_{k+1}} (-1)^{\sum_{k=1}^n y_k x_k} (-1)^{y_1} r_{\vec{x}} s_{\vec{y} \oplus \vec{x}}.
    \end{matrix*} \right.
  \end{equation}
  Here $\vec{y} \oplus \vec{x}$ is element-wise plus mod 2 (XOR) between two bit strings.
\end{observation}

\begin{proof}
Operator $P_{\vec{a}} \rho \sigma$ can be expressed as linear combination of products of 3 Pauli operators: 
\begin{equation*}
    P_{\vec{a}} \rho \sigma = \sum_l s_l P_{1, l} P_{2, l} P_{3, l}. 
\end{equation*}
Since $P_1 \in \mathcal{G}_{Y}$ commutes with all operators in $\mathcal{G}_{Y}$ or $\mathcal{P}_{XZ}$, only need to consider the following combination of $P_{2, l} P_{3, l}$: 
\begin{enumerate}
    \item $P_{2, l} \in \mathcal{G}_Y, P_{3, l} \in \mathcal{G}_Y$; 
    \item $P_{2, l} \in \mathcal{G}_Y, P_{3, l} \in \mathcal{P}_{XZ}$; 
    \item $P_{2, l} \in \mathcal{P}_{XZ}, P_{3, l} \in \mathcal{P}_{XZ}$. 
\end{enumerate}
Note that in the second case, we have $P_{2, l} P_{3, l} \in \mathcal{P}_{XZ}$ up to a scalar multiplication.
As a result, $P_{1, l} P_{2, l} P_{3, l} \in \mathcal{P}_{XZ}$ up to a scalar multiplication.
Hence $\mathrm{tr}[P_{1, l} P_{2, l} P_{3, l}] = 0$.
So it suffices to consider case 1 and case 3. 

For case 1, with $\vec{y}_1, \vec{y}_2, \vec{y}_3 \in \{0, 1\}^{n}$ satisfies $w_{H}(\vec{y}_1) = w_{H}(\vec{y}_2) = w_{H}(\vec{y}_3) = 0\text{\ mod\ }2$, we have
\begin{equation*}
    \mathrm{tr} \left[ Q_{\vec{y}_1} Q_{\vec{y}_2} Q_{\vec{y}_3} \right] = \mathrm{tr} \left[ Q_{\vec{y}_1 \oplus \vec{y}_2 \oplus \vec{y}_3} \right] = \left\{ \begin{matrix*}[l]
        2^n & \mathrm{if\ } \vec{y}_1 \oplus \vec{y}_2 \oplus \vec{y}_3 = \vec{0}, \\
        0 & \mathrm{otherwise.}
    \end{matrix*} \right. 
\end{equation*}
Here, $\oplus$ is plus mod $2$, which means this trace is non-zero if and only if $\vec{y}_3 = \vec{y}_1 \oplus \vec{y}_2$.
This gives a convenient way of calculating the trace via enumerating $\vec{y}_1$ and $\vec{y}_2$.

For case 2, with $\vec{y}, \vec{x}_1, \vec{x}_2 \in \{0, 1\}^{n}$ satisfies $w_{H}(\vec{y}) = 0\text{\ mod\ }2$, we have 
\begin{equation*}
    \left| \mathrm{tr} \left[ Q_{\vec{y}} P_{\vec{x}_1} P_{\vec{x}_2} \right] \right| = 
    \left| \mathrm{tr} \left[ Q_{\vec{y}} Q_{\vec{x}_1 \oplus \vec{x}_2} \right] \right| = 
    \left\{ \begin{matrix*}[l]
        2^n & \mathrm{if\ } \vec{y} \oplus \vec{x}_1 \oplus \vec{x}_2 = \vec{0}, \\
        0 & \mathrm{otherwise.}
    \end{matrix*} \right. 
\end{equation*}
The sign of this trace depends on the number of $\sigma_z \sigma_x$ and $\sigma_x \sigma_z$ in the multiplication.
For bit string with subscripts, denote $\vec{x}_{1} = (x_{1, 1} x_{1, 2} \cdots x_{1, n})$ and $\vec{x}_{2} = (x_{2, 1} x_{2, 2} \cdots x_{2, n})$.
Assume $\left| \mathrm{tr} \left[ Q_{\vec{y}} P_{\vec{x}_1} P_{\vec{x}_2} \right] \right| \neq 0$ for some choice of $\vec{y}$, $\vec{x}_1$ and $\vec{x}_2$, we have
\begin{equation*}
    \begin{aligned}
        &\frac{1}{2^n} \mathrm{tr} \left[ Q_{\vec{y}} P_{\vec{x}_1} P_{\vec{x}_2} \right] = \prod_{k=1}^n (-1)^{y_{k} x_{2, k}} i^{y_k} = (-1)^{\sum_{k=1}^n y_k x_{2, k}} i^{\sum_{k=1}^n y_k}\\
        =& (-1)^{\sum_{k=1}^n y_k x_{2, k}} i^{2\sum_{k=2}^n y_k} = (-1)^{\sum_{k=1}^n y_k x_{2, k}} (-1)^{\sum_{k=2}^n y_k} = (-1)^{\sum_{k=2}^n y_k x_{2, k}} (-1)^{y_1 (x_{2, 1} + 1)} \\
        =& (-1)^{\sum_{k=2}^n y_k x_{2, k}} (-1)^{y_1 x_{1, 1}}.
    \end{aligned}
\end{equation*}
To enumerate all possible $\vec{x}_1$ and $\vec{x}_2$, we use the following equality: 
\begin{equation*}
    \vec{x}_2 = \vec{y} \oplus \vec{x}_1. 
\end{equation*}
Then we have 
\begin{equation*}
    (-1)^{\sum_{k=2}^n y_k x_{2, k}} (-1)^{y_1 x_{1, 1}} = (-1)^{\sum_{k=2}^n y_k (x_{1, k} + y_k)} (-1)^{y_1 x_{1, 1}} = (-1)^{\sum_{k=1}^n y_k x_{1, k}} (-1)^{y_1}.
\end{equation*}

Now we are ready to give an analytic form of  $\mathrm{tr} \left[ P_{\vec{a}} \rho \sigma \right]$:
\begin{equation*}
    \begin{aligned}
        &\frac{1}{2^n} \mathrm{tr} \left[ P_{\vec{a}} \rho \sigma \right] \\
        = &\frac{1}{2^n \cdot 2^n \cdot 2^{n-1}} \left\{ \begin{matrix*}[l]
            &\mathrm{tr} \left[ \sum_{\substack{\vec{y} \in \{0, 1\}^{n} \\ w_H(\vec{y}) = 0\text{\ mod\ }2}} \left( (-1)^{\sum_{k=1}^{n-1} a_k y_{k+1}} Q_{\vec{y}} \sum_{\substack{\vec{y}_1 \in \{0, 1\}^{n} \\ w_H(\vec{y}_1) = 0\text{\ mod\ }2}} r_{\vec{y}_1} s_{\vec{y} \oplus \vec{y}_1} Q_{\vec{y}_1} Q_{\vec{y} \oplus \vec{y}_1} \right) \right] \\
            + & \mathrm{tr} \left[ \sum_{\substack{\vec{y} \in \{0, 1\}^{n} \\ w_H(\vec{y}) = 0\text{\ mod\ }2}} \left( (-1)^{\sum_{k=1}^{n-1} a_k y_{k+1}} Q_{\vec{y}} \sum_{\vec{x} \in \{0, 1\}^n} r_{\vec{x}} s_{\vec{x} \oplus \vec{y}} P_{\vec{x}} P_{\vec{x} \oplus \vec{y}} \right) \right]
        \end{matrix*} \right. \\
        = &\frac{1}{2^{3n-1}} \left\{ \begin{matrix*}[l]
            & \sum_{\substack{\vec{y} \in \{0, 1\}^{n} \\ w_H(\vec{y}) = 0\text{\ mod\ }2}} \sum_{\substack{\vec{y}_1 \in \{0, 1\}^{n} \\ w_H(\vec{y}_1) = 0\text{\ mod\ }2}} (-1)^{\sum_{k=1}^{n-1}a_k y_{k+1}} r_{\vec{y}_1} s_{\vec{y} \oplus \vec{y}_1} \\
            + & \sum_{\substack{\vec{y} \in \{0, 1\}^{n} \\ w_H(\vec{y}) = 0\text{\ mod\ }2}} \sum_{\vec{x} \in \{0, 1\}^n} (-1)^{\sum_{k=1}^{n-1} a_k y_{k+1}} (-1)^{\sum_{k=1}^n y_k x_{k}} (-1)^{y_1} r_{\vec{x}} s_{\vec{y} \oplus \vec{x}}
        \end{matrix*} \right.
    \end{aligned}
\end{equation*}
which is exactly Eq.~\eqref{eq:s:cubic-trace}
\end{proof}

Detailed calculation of $\lambda_{\vec{a}}$ is given in the proof of the following summarized result.

\begin{observation}
  The positivity $\lambda_{\vec{a}} \left( \{\theta_k\}_{k=1}^n \right)$ in Eq.~\eqref{eq:positivity} has the following closed form:
  \begin{equation}
    \label{eq:closed-positivity}
    \begin{aligned}
      &\lambda_{\vec{a}} \left( \{\theta_k\}_{k=1}^n \right) = \left( \text{tr} \left[ M_{\vec{a}} \left( \{\theta_k\}_{k=1}^n \right) \right] \right)^2 - \text{tr} \left[ M_{\vec{a}}^2 \left( \{\theta_k\}_{k=1}^n \right) \right] \\
      =& 2\mu^2 + \frac{1}{2^{2n}} \prod_{k=1}^n \left( 1 - g^2 (\theta_k) \right) - \frac{2\mu}{2^{n-1}} \sum_{\substack{\vec{y}\in\{0, 1\}^n \\ w_H(\vec{y}) = 0\text{\ mod\ 2}}} (-1)^{\sum_{k=2}^n y_k (a_{k-1}+k)} \prod_{k=1}^n g^{y_k} (\theta_k) \\
      &- s_n^2 \sum_{\substack{\vec{y} \in \{0, 1\}^n \\ w_H(\vec{y}) = 0\text{\ mod\ }2}} 8 (-1)^{\sum_{k=1}^{n-1} a_k y_{k+1}} (-1)^{\sum_{t=2}^n y_t t} \prod_{t=1}^n \left( \frac{1 + \sin(2\theta_t)}{2} + (-1)^{y_t} \frac{1 - \sin(2\theta_t)}{2} \right) \\
      &+ \frac{2s_n}{2^{n-2}} \sum_{\substack{\vec{y} \in \{0, 1\}^n \\ w_H(\vec{y}) = 0\text{\ mod\ }2}} \frac{1}{\sqrt{2}} (-1)^{\sum_{k=1}^{n-1} a_k y_{k+1}} (-1)^{\sum_{t=2}^n y_t t} \\
      &\qquad \times \prod_{t=1}^n \frac{1}{2} \left( \begin{matrix*}[l]
        \left[ (\cos \theta_t + g(\theta_t)\sin \theta_t) + (g(\theta_t) \cos \theta_t + \sin \theta_t) \right] \\
        + (-1)^{y_t} \left[ (\cos \theta_t + g(\theta_t)\sin\theta_t) - (g(\theta_t) \cos \theta_t + \sin \theta_t) \right]
      \end{matrix*} \right)
    \end{aligned}
  \end{equation}
\end{observation}

\begin{proof}
Only need to calculate analytic expression of $\text{tr} \left[ M^2_{\vec{a}}(\{\theta_k\}_{k=1}^n) \right]$ then combine with result in Eq.~\eqref{eq:s:linear-trace}.
We have
\begin{equation}
  \label{eq:s:square-trace}
  \begin{aligned}
  &\text{tr} \left[ M^2_{\vec{a}}(\{\theta_k\}_{k=1}^n) \right] \\
  =& \mathrm{tr} \left[ P_{\vec{a}} K^2 \left( \{ \theta_k \}_{k=1}^{n} \right) \right] + s_n^2 \mathrm{tr} \left[ P_{\vec{a}} W^2 \left( \{ \theta_k \}_{k=1}^{n} \right) \right] + \mu^2 \mathrm{tr} \left[ P_{\vec{a}} I^{\otimes n} \right] \\
  & - 2s_n \mathrm{tr} \left[ P_{\vec{a}} K \left( \{ \theta_k \}_{k=1}^{n} \right) W \left( \{ \theta_k \}_{k=1}^{n} \right) \right] - 2\mu \mathrm{tr} \left[ P_{\vec{a}} K \left( \{ \theta_k \}_{k=1}^{n} \right) \right] + 2s_n\mu \mathrm{tr} \left[ P_{\vec{a}} W \left( \{ \theta_k \}_{k=1}^{n} \right) \right].
  \end{aligned}
\end{equation}
\begin{enumerate}[label=(\roman*)]
  \item $\mathrm{tr} \left[ P_{\vec{a}} K^2\left(\left\{ \theta_k \right\}_{k=1}^n \right) \right]$.
  Assume $K\left( \{\theta_k\}_{k=1}^n \right)$ has the following generalized Bloch representation:
  \begin{equation}
    \label{eq:s:dephased-bloch}
    K\left( \{\theta_k\}_{k=1}^n \right) = \frac{1}{2^n} \sum_{\substack{\vec{y} \in \{0, 1\}^{n} \\ w_H(\vec{y}) = 0\text{\ mod\ }2}} r_{\vec{y}} Q_{\vec{y}} + \frac{1}{2^n} \sum_{\vec{x} \in \{0, 1\}^{n}} r_{\vec{x}} R_{\vec{x}}.
  \end{equation}
  Then according to Observation~\ref{obv:s:bloch-coef} and \ref{obv:s:dephased-bloch}, the analytic form of $r_{\vec{x}}$ and $r_{\vec{y}}$ can be written
  \begin{equation*}
    \begin{aligned}
      r_{\vec{x}} &= \frac{1}{\sqrt{2}} (-1)^{\sum_{t=2}^n x_t \left( \sum_{k=1}^t x_k \right)} (-1)^{\sum_{t=2}^n x_t (t + 1)} \prod_{t=1}^n g^{x_t} (\theta_t), \\
      r_{\vec{y}} &= (-1)^{\sum_{t=2}^n y_t t} \prod_{t=1}^n g^{y_t} (\theta_t).
    \end{aligned}
  \end{equation*}
  Use Observation~\ref{obv:s:cubic-trace} to calculate the trace, we have
  \begin{equation*}
      \begin{aligned}
          &2^{2n-1} \mathrm{tr} \left[ P_{\vec{a}} K^2\left(\left\{ \theta_k \right\}_{k=1}^n \right) \right] \\
          =& \sum_{\substack{\vec{y} \in \{0,1 \}^{n} \\ w_H(\vec{y}) = 0\text{\ mod\ }2}} \sum_{\substack{\vec{y}_1 \in \{0,1 \}^{n} \\ w_H(\vec{y}_1) = 0\text{\ mod\ }2}} (-1)^{\sum_{k=1}^{n-1} a_k y_{k+1}} r_{\vec{y}_1} r_{\vec{y} \oplus \vec{y}_1} \\
          & + \sum_{\substack{\vec{y} \in \{0,1 \}^{n} \\ w_H(\vec{y}) = 0\text{\ mod\ }2}} \sum_{\vec{x} \in \{0,1\}^n} (-1)^{\sum_{k=1}^{n-1} a_k y_{k+1}} (-1)^{\sum_{k=1}^n y_k x_k} (-1)^{y_1} r_{\vec{x}} r_{\vec{y} \oplus \vec{x}} \\
          =& \sum_{\substack{\vec{y} \in \{0,1 \}^{n} \\ w_H(\vec{y}) = 0\text{\ mod\ }2}} \sum_{\substack{\vec{y}_1 \in \{0,1 \}^{n} \\ w_H(\vec{y}_1) = 0\text{\ mod\ }2}} (-1)^{\sum_{k=1}^{n-1} a_k y_{k+1}} (-1)^{\sum_{t=2}^n y_{1, t} t} (-1)^{\sum_{t=2}^n(y_{1, t} + y_t) t} \prod_{t=1}^n \left[ g(\theta_t) \right]^{y_{1, t} + (y_{1, t} + y_{t} \mathrm{\ mod\ 2})} \\
          &+ \sum_{\substack{\vec{y} \in \{0,1 \}^{n} \\ w_H(\vec{y}) = 0\text{\ mod\ }2}} \sum_{\vec{x} \in \{0,1\}^n} \frac{1}{2} (-1)^{\sum_{k=1}^{n-1} a_k y_{k+1}} (-1)^{\sum_{k=2}^n y_k (x_1 + x_k + 1)} \\
          &\quad\times \left\{ \begin{matrix*}[l]
              & (-1)^{\sum_{t=2}^n x_t \left( \sum_{k=1}^t x_k \right)} (-1)^{\sum_{t=2}^n x_t (t+1)} \prod_{t=1}^n g^{x_t}(\theta_t) \\
              \times & (-1)^{\sum_{t=2}^n \left( (x_t + y_t) \left( \sum_{k=1}^t (x_k + y_k) \right) \right)} (-1)^{\sum_{t=2}^n \left( (x_t + y_t) (t+1) \right)} \prod_{t=1}^n \left[g(\theta_t)\right]^{(x_t + y_t \mathrm{\ mod\ }2)}
          \end{matrix*} \right. 
      \end{aligned}
  \end{equation*}
  To simplify this expression, by inspection one can establish the following equality:
  \begin{equation*}
    \begin{aligned}
        &(-1)^{\sum_{k=2}^n y_k(x_1 + x_k + 1)} \\
            &\times (-1)^{\sum_{t=2}^n x_t(\sum_{k=1}^t x_k)} (-1)^{\sum_{t=2}^n x_t (t+1)} \\
            &\times (-1)^{\sum_{t=2}^n \left( (x_t + y_t) \left( \sum_{k=1}^t (x_k + y_k) \right) \right)} (-1)^{\sum_{t=2}^n \left( (x_t + y_t) (t + 1) \right)} \\
            =& (-1)^{\sum_{t=2}^n y_t t}.
    \end{aligned}
  \end{equation*}
  As a result,
  \begin{equation}
    \label{eq:s:positivity-term-1}
      \begin{aligned}
          &\mathrm{tr} \left[ P_{\vec{a}} K^2\left(\left\{ \theta_k \right\}_{k=1}^n \right) \right] \\
          =& \frac{1}{2^{2n-1}} \sum_{\substack{\vec{y} \in \{0,1 \}^{n} \\ w_H(\vec{y}) = 0\text{\ mod\ }2}} \sum_{\substack{\vec{y}_1 \in \{0,1 \}^{n} \\ w_H(\vec{y}_1) = 0\text{\ mod\ }2}} (-1)^{\sum_{k=1}^{n-1} a_k y_{k+1}} (-1)^{\sum_{t=2}^n y_t t} \prod_{t=1}^n \left[ g(\theta_t) \right]^{y_{1, t} + (y_{1, t} + y_{t} \mathrm{\ mod\ 2})} \\
          &+ \frac{1}{2^{2n}} \sum_{\substack{\vec{y} \in \{0,1 \}^{n} \\ w_H(\vec{y}) = 0\text{\ mod\ }2}} \sum_{\vec{x} \in \{0,1\}^n} (-1)^{\sum_{k=1}^{n-1} a_k y_{k+1}} (-1)^{\sum_{t=2}^n y_t t} \prod_{t=1}^n \left[g(\theta_t)\right]^{x_t + (x_t + y_t \mathrm{\ mod\ }2)}.
      \end{aligned}
  \end{equation}
  \item $s_n^2 \mathrm{tr} \left[ P_{\vec{a}} W^2\left(\left\{ \theta_k \right\}_{k=1}^n \right) \right]$.
  By definition of Bell operator $W_n \left( \{\theta_k\}_{k=1}^n \right)$, we can write its generalized Bloch representation as:
  \begin{equation}
    \label{eq:s:dephased-bell}
    \begin{aligned}
      &W_n \left( \{\theta_k\}_{k=1}^n \right) = \frac{1}{2^n} \sum_{\vec{x} \in \{0, 1\}^n} p_{\vec{x}} R_{\vec{x}}, \\
      &p_{\vec{x}} = 2^{n+1} (-1)^{\sum_{t=2}^n x_t \left( \sum_{k=1}^t x_k \right)} (-1)^{\sum_{t=2}^n x_t\left( t+1 \right)} \prod_{t=1}^n \left( \frac{\cos \theta_t + \sin \theta_t}{2} + (-1)^{x_t} \frac{\cos \theta_t - \sin \theta_t}{2} \right).
    \end{aligned}
  \end{equation}
  Then, according to Observation~\ref{obv:s:cubic-trace}, we have
  \begin{equation*}
      \begin{aligned}
          &2^{2n-1} \mathrm{tr} \left[ P_{\vec{a}} W^2\left(\left\{ \theta_k \right\}_{k=1}^n \right) \right] \\
          =& \sum_{\substack{\vec{y} \in \{0, 1\}^{n} \\ w_H(\vec{y}) = 0\text{\ mod\ }2}} \sum_{\vec{x} \in \{0, 1\}^n} (-1)^{\sum_{k=1}^{n-1} a_k y_{k+1}} (-1)^{\sum_{k=1}^n y_k x_k} (-1)^{y_1} p_{\vec{x}} p_{\vec{y} \oplus \vec{x}} \\
          =& \sum_{\substack{\vec{y} \in \{0, 1\}^{n} \\ w_H(\vec{y}) = 0\text{\ mod\ }2}} \sum_{\vec{x} \in \{0, 1\}^n} (-1)^{\sum_{k=1}^{n-1} a_k y_{k+1}} (-1)^{\sum_{k=1}^n y_k x_k} (-1)^{y_1} \\
          & \times \left\{ \begin{matrix*}[l]
              &2^{n+1} (-1)^{\sum_{t=2}^n x_t \left( \sum_{k=1}^t x_k \right)} (-1)^{\sum_{t=2}^n x_t (t + 1)} \prod_{t=1}^n \left( \frac{\cos \theta_t + \sin \theta_t}{2} + (-1)^{x_t} \frac{\cos \theta_t - \sin \theta_t}{2} \right) \\
              \times &2^{n+1} (-1)^{\sum_{t=2}^n \left( (x_t + y_t) \left( \sum_{k=1}^t (x_k + y_k) \right) \right)} (-1)^{\sum_{t=2}^n \left( (x_t + y_t) (t+1) \right)} \\
              \times &\prod_{t=1}^n \left( \frac{\cos \theta_t + \sin \theta_t}{2} + (-1)^{x_t + y_t} \frac{\cos \theta_t - \sin \theta_t}{2} \right)
          \end{matrix*} \right. \\
          =& \sum_{\substack{\vec{y} \in \{0, 1\}^{n} \\ w_H(\vec{y}) = 0\text{\ mod\ }2}} \sum_{\vec{x} \in \{0, 1\}^n} 2^{2(n+1)} (-1)^{\sum_{k=1}^{n-1} a_k y_{k+1}} (-1)^{\sum_{t=2}^n y_t t} \left\{ \begin{matrix*}
              &\prod_{t=1}^n \left( \frac{\cos \theta_t + \sin \theta_t}{2} + (-1)^{x_t} \frac{\cos \theta_t - \sin \theta_t}{2} \right) \\
              \times &\prod_{t=1}^n \left( \frac{\cos \theta_t + \sin \theta_t}{2} + (-1)^{x_t + y_t} \frac{\cos \theta_t - \sin \theta_t}{2} \right)
          \end{matrix*} \right.  
      \end{aligned}
  \end{equation*}
  Here the last equality holds by simplifying the coefficients according to procedure similar to the first term.
  Fix $\vec{y}$ and look into the sum over $\vec{x}$, we have 
  \begin{equation}
    \label{eq:s:positivity-term-3}
      \begin{aligned}
          &\mathrm{tr} \left[ P_{\vec{a}} W^2\left(\left\{ \theta_k \right\}_{k=1}^n \right) \right] = \sum_{\vec{y} \in \{0, 1\}^{n-1}} 8 (-1)^{\sum_{k=1}^{n-1} a_k y_{k+1}} (-1)^{\sum_{t=2}^n y_t t} \prod_{t=1}^n \left( \frac{1 + \sin (2\theta_t)}{2} + (-1)^{y_t} \frac{1 - \sin 2\theta_{t}}{2} \right).
      \end{aligned}
  \end{equation}
  \item $\mu^2 \mathrm{tr} \left[ P_{\vec{a}} I^{\otimes n} \right] = 2\mu^2$.
  \item $2s_n \mathrm{tr} \left[ P_{\vec{a}} K \left( \{\theta_k\}_{k=1}^{n} \right) W_n \left( \{\theta_k\}_{k=1}^{n} \right) \right]$.
  By Eq.~\eqref{eq:s:dephased-bloch} and Eq.~\eqref{eq:s:dephased-bell}, combined with Observation~\ref{obv:s:cubic-trace}, we have
  \begin{equation*}
      \begin{aligned}
          &\mathrm{tr} \left[ P_{\vec{a}} K \left( \{\theta_k\}_{k=1}^{n} \right) W_n \left( \{\theta_k\}_{k=1}^{n} \right) \right] \\
          =& \frac{1}{2^{2n-1}} \sum_{\substack{\vec{y} \in \{0, 1\}^{n} \\ w_H(\vec{y}) = 0\text{\ mod\ }2}} \sum_{\vec{x} \in \{0, 1\}^n} (-1)^{\sum_{k=1}^{n-1}a_k y_{k+1}} (-1)^{\sum_{k=1}^n y_k x_k} (-1)^{y_1} r_{\vec{x}} p_{\vec{y} \oplus \vec{x}} \\
          =& \frac{1}{2^{2n-1}} \sum_{\substack{\vec{y} \in \{0, 1\}^{n} \\ w_H(\vec{y}) = 0\text{\ mod\ }2}} \sum_{\vec{x} \in \{0, 1\}^n} (-1)^{\sum_{k=1}^{n-1}a_k y_{k+1}} (-1)^{\sum_{k=1}^n y_k x_k} (-1)^{y_1} \\
          &\times \left\{\begin{matrix*}[l]
              &\frac{1}{\sqrt{2}} (-1)^{\sum_{t=2}^n x_t \left( \sum_{k=1}^t x_k \right) (-1)^{\sum_{t=2}^n x_t (t + 1)}} \prod_{t=1}^n g^{x_t} (\theta_t) \\
              \times & 2^{n+1}(-1)^{\sum_{t=2}^n \left( \left( x_t + y_t \right) \left( \sum_{k=1}^t \left( x_k + y_k \right) \right) \right) (-1)^{\sum_{t=2}^n \left( \left( x_t + y_t \right) (t + 1) \right)}} \\
              \times & \prod_{t=1}^n \left( \frac{\cos \theta_t + \sin \theta_t}{2} + (-1)^{x_t + y_t} \frac{\cos \theta_t - \sin \theta_t}{2} \right)
          \end{matrix*} \right. \\
          =& \frac{1}{2^{n-2}} \sum_{\substack{\vec{y} \in \{0, 1\}^{n} \\ w_H(\vec{y}) = 0\text{\ mod\ }2}} \sum_{\vec{x} \in \{0, 1\}^n} \frac{1}{\sqrt{2}} (-1)^{\sum_{k=1}^{n-1} a_k y_{k+1}} (-1)^{\sum_{t=2}^n y_t t} \\ 
          &\times \prod_{t=1}^n g^{x_t} (\theta_t) \left( \frac{\cos \theta_t + \sin \theta_t}{2} + (-1)^{x_t + y_t} \frac{\cos \theta_t - \sin \theta_t}{2} \right).
      \end{aligned}
  \end{equation*}
  Similarly, by fixing $\vec{y}$ and look into the summation over $\vec{x}$, we have 
  \begin{equation}
    \label{eq:s:positivity-term-2}
      \begin{aligned}
          &\mathrm{tr} \left[ P_{\vec{a}} K W \left( \{\theta_k\}_{k=1}^{n} \right) \right] \\
          =& \frac{1}{2^{n-2}} \sum_{\substack{\vec{y} \in \{0, 1\}^{n} \\ w_H(\vec{y}) = 0\text{\ mod\ }2}} \frac{1}{\sqrt{2}} (-1)^{\sum_{k=1}^{n-1} a_k y_{k+1}} (-1)^{\sum_{t=2}^n y_t t} \\
          & \times \prod_{t=1}^n \frac{1}{2} \left( \begin{matrix*}
              \left[ (\cos \theta_t + g(\theta_t) \sin \theta_t) + (g(\theta_t) \cos \theta_t + \sin \theta_t) \right] \\
              + (-1)^{y_t} \left[ (\cos \theta_t + g(\theta_t) \sin \theta_t) - (g(\theta_t) \cos \theta_t + \sin \theta_t) \right]
          \end{matrix*} \right)
      \end{aligned}
  \end{equation}

  Then consider the quadratic terms.
  \item $2\mu \mathrm{tr} \left[ P_{\vec{a}} K \left( \{\theta_k\}_{k=1}^{n} \right) \right]$.
  We already have Eq.~\eqref{eq:91}.
  \item $2\mu \mathrm{tr} \left[ P_{\vec{a}} W \left( \{\theta_k\}_{k=1}^{n} \right) \right] = 0$. This is because $W_n \left( \{\theta_k\}_{k=1}^{n} \right)$ is linear combination of tensor products of Pauli-$X$s and Pauli-$Z$s, and $P_{\vec{a}}$ is linear combination of tensor products of even number of Pauli-$Y$s.
\end{enumerate}
Also, from Eq.~\eqref{eq:s:linear-trace} we have
\begin{equation}
    \label{eq:s:trace-square}
    \begin{aligned}
        \left( \mathrm{tr} \left[ M_{\vec{a}} \right] \right)^2 =& \frac{1}{2^{2(n-1)}} \sum_{\substack{\vec{y} \in \{0, 1\}^{n} \\ w_H(\vec{y}) = 0\text{\ mod\ }2}} \sum_{\substack{\vec{y} \in \{0, 1\}^{n} \\ w_H(\vec{y}) = 0\text{\ mod\ }2}} (-1)^{\sum_{k=2}^n \left( (y_k + y_{1, k}) (a_{k-1} + k) \right)} \prod_{k=1}^{n} g^{y_k + y_{1, k}}(\theta_k) \\
        &- \frac{4\mu}{2^{n-1}} \sum_{\substack{\vec{y} \in \{0, 1\}^{n} \\ w_H(\vec{y}) = 0\text{\ mod\ }2}} (-1)^{\sum_{k=2}^n y_k (a_{k-1} + k)} \prod_{k=1}^n g^{y_k} (\theta_k) + 4\mu^2
    \end{aligned}
\end{equation}

Take Eq.~\eqref{eq:s:square-trace} into Eq.~\eqref{eq:positivity}, now we expand and simplify the following expression:
\begin{align}
    &\lambda_{\vec{a}} = \left( \mathrm{tr} \left[ M_{\vec{a}} \right] \right)^2 - \mathrm{tr} \left[ M_{\vec{a}}^2 \right] \nonumber\\
    =& \left( \mathrm{tr} \left[ P_{\vec{a}} T \left( \{ \theta_k \}_{k=1}^{n} \right) \right] \right)^2 - \mathrm{tr} \left[ P_{\vec{a}} K^2 \left( \{ \theta_k \}_{k=1}^{n} \right) \right] - s_n^2 \mathrm{tr} \left[ P_{\vec{a}} W^2 \left( \{ \theta_k \}_{k=1}^{n} \right) \right] - \mu^2 \mathrm{tr} \left[ P_{\vec{a}} I^{\otimes n} \right] \nonumber \\
    &+ 2s_n \mathrm{tr} \left[ P_{\vec{a}} K W \left( \{ \theta_k \}_{k=1}^{n} \right) \right] + 2\mu \mathrm{tr} \left[ P_{\vec{a}} K \left( \{ \theta_k \}_{k=1}^{n} \right) \right] - 2s_n\mu \mathrm{tr} \left[ P_{\vec{a}} W \left( \{ \theta_k \}_{k=1}^{n} \right) \right] \nonumber \\
    =& \left( \mathrm{tr} \left[ P_{\vec{a}} T \left( \{ \theta_k \}_{k=1}^{n} \right) \right] \right)^2 + 2\mu \mathrm{tr} \left[ P_{\vec{a}} K \left( \{ \theta_k \}_{k=1}^{n} \right) \right] - 2\mu^2 - \mathrm{tr} \left[ P_{\vec{a}} K^2 \left( \{ \theta_k \}_{k=1}^{n} \right) \right] \label{eq:s:positivity-f4-term} \\
    &+ 2s_n \mathrm{tr} \left[ P_{\vec{a}} K W \left( \{ \theta_k \}_{k=1}^{n} \right) \right] - s_n^2 \mathrm{tr} \left[ P_{\vec{a}} W^2 \left( \{ \theta_k \}_{k=1}^{n} \right) \right]. \nonumber 
\end{align}

Combining Eq.~\eqref{eq:s:trace-square}, \eqref{eq:91}, \eqref{eq:s:positivity-term-1}, we have
\begin{align}
    & \left( \mathrm{tr} \left[ P_{\vec{a}} T \left( \{ \theta_k \}_{k=1}^{n} \right) \right] \right)^2 + 2\mu \mathrm{tr} \left[ P_{\vec{a}} K \left( \{ \theta_k \}_{k=1}^{n} \right) \right] - 2\mu^2 - \mathrm{tr} \left[ P_{\vec{a}} K^2 \left( \{ \theta_k \}_{k=1}^{n} \right) \right] \nonumber \\
    =& \frac{4}{2^{2n}} \sum_{\substack{\vec{y} \in \{0, 1\}^{n} \\ w_H(\vec{y}) = 0\text{\ mod\ }2}} \sum_{\substack{\vec{y}_1 \in \{0, 1\}^{n} \\ w_H(\vec{y}_1) = 0\text{\ mod\ }2}} (-1)^{\sum_{k=2}^n \left( (y_k + y_{1, k}) (a_{k-1} + k) \right)} \prod_{k=1}^{n} g^{y_k + y_{1, k}}(\theta_k) \label{eq:s:expand-1} \\
    &- \frac{4\mu}{2^{n-1}} \sum_{\substack{\vec{y} \in \{0, 1\}^{n} \\ w_H(\vec{y}) = 0\text{\ mod\ }2}} (-1)^{\sum_{k=2}^n y_k (a_{k-1} + k)} \prod_{k=1}^n g^{y_k} (\theta_k) + 4\mu^2 \nonumber\\
    &+ \frac{2\mu}{2^{n-1}} \sum_{\substack{\vec{y} \in \{0, 1\}^{n} \\ w_H(\vec{y}) = 0\text{\ mod\ }2}} (-1)^{\sum_{k=2}^n y_k (a_{k-1} + k)} \prod_{k=1}^n g^{y_k} (\theta_k) - 2\mu^2 \nonumber \\
    &- \frac{2}{2^{2n}} \sum_{\substack{\vec{y} \in \{0, 1\}^{n} \\ w_H(\vec{y}) = 0\text{\ mod\ }2}} \sum_{\substack{\vec{y}_1 \in \{0, 1\}^{n} \\ w_H(\vec{y}_1) = 0\text{\ mod\ }2}} (-1)^{\sum_{k=2}^{n} a_{k-1} y_{k}} (-1)^{\sum_{t=2}^n y_t t} \prod_{t=1}^n g^{y_{1, t} + (y_{1, t} + y_t \mathrm{\ mod\ 2})} (\theta_t) \label{eq:s:expand-2}\\
    &- \frac{1}{2^{2n}} \sum_{\substack{\vec{y} \in \{0, 1\}^{n} \\ w_H(\vec{y}) = 0\text{\ mod\ }2}} \sum_{\vec{x} \in \{0, 1\}^{n}} (-1)^{\sum_{k=2}^n a_{k-1} y_k} (-1)^{\sum_{t=2}^n y_t t} \prod_{t=1}^n g^{x_t + (x_t + y_t \mathrm{\ mod\ 2})}(\theta_t). \label{eq:s:expand-3}
\end{align}
Fix $k$ in each term in the summation of Eq.~\eqref{eq:s:expand-1} and inspect its value for different $y_k$ and $y_{1, k}$, we have
\begin{equation*}
    \begin{aligned}
        &(-1)^{(y_k + y_{1, k})(a_{k-1} + k)} g^{y_k + y_{1, k}}(\theta_k) = \left\{ \begin{matrix*}[c]
            y_k = 0, y_{1, k} = 0: & 1\\
            y_k = 0, y_{1, k} = 1: & (-1)^{a_{k-1}}(-1)^k g(\theta_k)\\
            y_k = 1, y_{1, k} = 0: & (-1)^{a_{k-1}}(-1)^k g(\theta_k)\\
            y_k = 1, y_{1, k} = 1: & g^2(\theta_k)\\
        \end{matrix*} \right.
    \end{aligned}
\end{equation*}
Compare with each term in the second-to-last term in Eq.~\eqref{eq:s:expand-2}:
\begin{equation*}
    \begin{aligned}
        &(-1)^{a_{t-1} y_{t}} (-1)^{y_t t} g^{y_{1, t} + (y_{1, t} + y_t \mathrm{\ mod\ 2})} (\theta_t) = \left\{ \begin{matrix*}[c]
            y_t = 0, y_{1, t} = 0: & 1\\
            y_t = 0, y_{1, t} = 1: & g^2(\theta_t)\\
            y_k = 1, y_{1, t} = 0: & (-1)^{a_{t-1}}(-1)^t g(\theta_t)\\
            y_t = 1, y_{1, t} = 1: & (-1)^{a_{t-1}}(-1)^t g(\theta_t)\\
        \end{matrix*} \right.
    \end{aligned}
\end{equation*}
This implies that
\begin{equation*}
    \begin{aligned}
        &\sum_{\substack{\vec{y} \in \{0, 1\}^{n} \\ w_H(\vec{y}) = 0\text{\ mod\ }2}} \sum_{\substack{\vec{y}_1 \in \{0, 1\}^{n} \\ w_H(\vec{y}_1) = 0\text{\ mod\ }2}} (-1)^{\sum_{k=2}^n \left( (y_k + y_{1, k}) (a_{k-1} + k) \right)} \prod_{k=1}^{n} g^{y_k + y_{1, k}}(\theta_k) \\
        =& \sum_{\substack{\vec{y} \in \{0, 1\}^{n} \\ w_H(\vec{y}) = 0\text{\ mod\ }2}} \sum_{\substack{\vec{y}_1 \in \{0, 1\}^{n} \\ w_H(\vec{y}_1) = 0\text{\ mod\ }2}} (-1)^{\sum_{k=2}^{n} a_{k-1} y_{k}} (-1)^{\sum_{t=2}^n y_t t} \prod_{t=1}^n g^{y_{1, t} + (y_{1, t} + y_t \mathrm{\ mod\ 2})} (\theta_t),
    \end{aligned}
\end{equation*}
Moreover, the last term in Eq.~\eqref{eq:s:expand-3} can be written:
\begin{equation*}
    \begin{aligned}
        &\sum_{\substack{\vec{y} \in \{0, 1\}^{n} \\ w_H(\vec{y}) = 0\text{\ mod\ }2}} \sum_{\vec{x} \in \{0, 1\}^{n}} (-1)^{\sum_{k=2}^n a_{k-1} y_k} (-1)^{\sum_{t=2}^n y_t t} \prod_{t=1}^n g^{x_t + (x_t + y_t \mathrm{\ mod\ 2})}(\theta_t) \\
        =&\left\{ \begin{matrix*}[l]
            & \sum_{\substack{\vec{y} \in \{0, 1\}^{n} \\ w_H(\vec{y}) = 0\text{\ mod\ }2}} \sum_{\substack{\vec{y}_1 \in \{0, 1\}^{n} \\ w_H(\vec{y}_1) = 0\text{\ mod\ }2}} (-1)^{\sum_{k=2}^{n} a_{k-1} y_{k}} (-1)^{\sum_{t=2}^n y_t t} \prod_{t=1}^n g^{y_{1, t} + (y_{1, t} + y_t \mathrm{\ mod\ 2})} (\theta_t) \\
            + & \sum_{\substack{\vec{y} \in \{0, 1\}^{n} \\ w_H(\vec{y}) = 0\text{\ mod\ }2}} \sum_{\substack{\vec{y}_1 \in \{0, 1\}^{n} \\ w_H(\vec{y}_1) = 0\text{\ mod\ }2}} (-1)^{\sum_{k=2}^{n} a_{k-1} y_{k}} (-1)^{\sum_{t=2}^n y_t t} \prod_{t=2}^n g^{y_{1, t} + (y_{1, t} + y_t \mathrm{\ mod\ 2})} (\theta_t) \\
            & \times g^{(y_{1, 1} + 1 \mathrm{\ mod\ 2}) + ((y_{1, 1} + 1 \mathrm{\ mod\ 2}) + y_1 \mathrm{\ mod\ 2})} (\theta_1) 
        \end{matrix*}\right. 
    \end{aligned}
\end{equation*}
So we can merge 3 terms given in Eq.~\eqref{eq:s:expand-1}, Eq.~\eqref{eq:s:expand-2} and Eq.~\eqref{eq:s:expand-3}, then obtain:
\begin{equation}
    \label{eq:s:simplified-f4-term}
    \begin{aligned}
        & \left( \mathrm{tr} \left[ P_{\vec{a}} T \left( \{ \theta_k \}_{k=1}^{n} \right) \right] \right)^2 + 2\mu \mathrm{tr} \left[ P_{\vec{a}} K \left( \{ \theta_k \}_{k=1}^{n} \right) \right] - 2\mu^2 - \mathrm{tr} \left[ P_{\vec{a}} K^2 \left( \{ \theta_k \}_{k=1}^{n} \right) \right] \\
        =& \frac{1}{2^{2n}} \sum_{\substack{\vec{y} \in \{0, 1\}^{n} \\ w_H(\vec{y}) = 0\text{\ mod\ }2}} \sum_{\substack{\vec{y}_1 \in \{0, 1\}^{n} \\ w_H(\vec{y}_1) = 0\text{\ mod\ }2}} \left\{ \begin{matrix*}[l]
            (-1)^{\sum_{k=2}^{n} a_{k-1} y_{k}} (-1)^{\sum_{t=2}^n y_t t} \left[ \prod_{t=2}^n g^{y_{1, t} + (y_{1, t} + y_t \mathrm{\ mod\ 2})} (\theta_t) \right] \\
            \times \left[ \begin{matrix*}[l]
                & g^{y_{1, 1} + (y_{1, 1} + y_1 \mathrm{\ mod\ 2})} (\theta_1) \\
                -& g^{(y_{1, 1} + 1 \mathrm{\ mod\ 2}) + ((y_{1, 1} + 1 \mathrm{\ mod\ 2}) + y_1 \mathrm{\ mod\ 2})} (\theta_1)
            \end{matrix*} \right]
        \end{matrix*} \right\} \\
        &- \frac{2\mu}{2^{n-1}} \sum_{\substack{\vec{y} \in \{0, 1\}^{n} \\ w_H(\vec{y}) = 0\text{\ mod\ }2}} (-1)^{\sum_{k=2}^n y_k (a_{k-1} + k)} \prod_{k=1}^n g^{y_k} (\theta_k) + 2\mu^2 \\
        =& \frac{1}{2^{2n}} \prod_{k=1}^n \left( 1 - g^2 (\theta_k) \right) - \frac{2\mu}{2^{n-1}} \sum_{\substack{\vec{y} \in \{0, 1\}^{n} \\ w_H(\vec{y}) = 0\text{\ mod\ }2}} (-1)^{\sum_{k=2}^n y_k (a_{k-1} + k)} \prod_{k=1}^n g^{y_k} (\theta_k) + 2\mu^2,
    \end{aligned}
\end{equation}
where the last equality follows Observation~\ref{obv:s:simplify-positivity-1}.
Combining Eq.~\eqref{eq:s:simplified-f4-term}, Eq.~\eqref{eq:s:positivity-term-2} and Eq.~\eqref{eq:s:positivity-term-3}, we obtain Eq.~\eqref{eq:closed-positivity}.
\end{proof}

\begin{observation}
    \label{obv:s:simplify-positivity-1}
    \begin{equation}
        \label{eq:s:simplify-positivity-1}
        \begin{aligned}
            &\sum_{\substack{\vec{y} \in \{0, 1\}^{n} \\ w_H(\vec{y}) = 0\text{\ mod\ }2}} \sum_{\substack{\vec{y}_1 \in \{0, 1\}^{n} \\ w_H(\vec{y}_1) = 0\text{\ mod\ }2}} \left\{ \begin{matrix*}[l]
            (-1)^{\sum_{k=2}^{n} a_{k-1} y_{k}} (-1)^{\sum_{t=2}^n y_t t} \left[ \prod_{t=2}^n g^{y_{1, t} + (y_{1, t} + y_t \mathrm{\ mod\ 2})} (\theta_t) \right] \\
            \times \left[ \begin{matrix*}[l]
                & g^{y_{1, 1} + (y_{1, 1} + y_1 \mathrm{\ mod\ 2})} (\theta_1) \\
                -& g^{(y_{1, 1} + 1 \mathrm{\ mod\ 2}) + ((y_{1, 1} + 1 \mathrm{\ mod\ 2}) + y_1 \mathrm{\ mod\ 2})} (\theta_1)
            \end{matrix*} \right]
        \end{matrix*} \right\} \\
            =& \prod_{k=1}^n \left(1 - g^2(\theta_k) \right).
        \end{aligned}
    \end{equation}
\end{observation}
\begin{proof}
We first prove that, for a fixed choice of $\vec{y} = (y_1 y_2 \cdots y_n) \in \{0, 1\}^n$ with $w_H(\vec{y}) = 0 \text{\ mod\ }2$, we have
\begin{equation*}
    \begin{aligned}
        &\sum_{\substack{\vec{y}_1 \in \{0, 1\}^{n} \\ w_H(\vec{y}_1) = 0\text{\ mod\ }2}} \left\{ \begin{matrix*}[l]
            (-1)^{\sum_{k=2}^{n} a_{k-1} y_{k}} (-1)^{\sum_{t=2}^n y_t t} \left[ \prod_{t=2}^n g^{y_{1, t} + (y_{1, t} + y_t \mathrm{\ mod\ 2})} (\theta_t) \right] \\
            \times \left[ \begin{matrix*}[l]
                & g^{y_{1, 1} + (y_{1, 1} + y_1 \mathrm{\ mod\ 2})} (\theta_1) \\
                -& g^{(y_{1, 1} + 1 \mathrm{\ mod\ 2}) + ((y_{1, 1} + 1 \mathrm{\ mod\ 2}) + y_1 \mathrm{\ mod\ 2})} (\theta_1)
            \end{matrix*} \right]
        \end{matrix*} \right\} = 0 \mathrm{\ if\ }\vec{y} \neq \vec{0}.
    \end{aligned}
\end{equation*}
Discuss by cases of value of $y_1$.
If $y_1 = 0$, then since $\vec{y} \neq \vec{0}$, then there exists some integer $k_0 \geq 2$ such that $y_{k_0} = 1$.
Without loss of generality we can assume $k_0 = 2$.
Then Eq.~\eqref{eq:s:simplify-positivity-1} can be written
\begin{equation}
    \label{eq:s:inproof-1}
    \begin{aligned}
        &\sum_{\substack{\vec{y}_1 \in \{0, 1\}^{n} \\ w_H(\vec{y}_1) = 0\text{\ mod\ }2}} \left\{ \begin{matrix*}[l]
            (-1)^{\sum_{k=2}^{n} a_{k-1} y_{k}} (-1)^{\sum_{t=2}^n y_t t} \left[ \prod_{t=2}^n g^{y_{1, t} + (y_{1, t} + y_t \mathrm{\ mod\ 2})} (\theta_t) \right] \\
            \times \left[ g^{2y_{1, 1}} (\theta_1) - g^{2(y_{1, 1} + 1 \mathrm{\ mod\ 2})} (\theta_1) \right]
        \end{matrix*} \right\} \\
        =&\sum_{\substack{\vec{y}_1 \in \{0, 1\}^{n} \\ w_H(\vec{y}_1) = 0\text{\ mod\ }2}}
        \left\{ (-1)^{\sum_{k=3}^{n} a_{k-1} y_{k}} (-1)^{\sum_{t=3}^n y_t t} \left[ \prod_{t=3}^n g^{y_{1, t} + (y_{1, t} + y_t \mathrm{\ mod\ 2})} (\theta_t) \right] (-1)^{a_1 + 2} g(\theta_2) (-1)^{y_{1, 1}} \left( 1 - g^2 (\theta_1) \right) \right\} \\
        =& \sum_{\substack{\vec{y}_1 \in \{0, 1\}^{n} \\ w_H(\vec{y}_1) = 0\text{\ mod\ }2 \\ y_{1, 2} = 0}}
        \left\{ (-1)^{\sum_{k=3}^{n} a_{k-1} y_{k}} (-1)^{\sum_{t=3}^n y_t t} \left( \prod_{t=3}^n g^{y_{1, t} + (y_{1, t} + y_t \mathrm{\ mod\ 2})} (\theta_t) \right) (-1)^{a_1 + 2} g(\theta_2) (-1)^{y_{1, 1}} \left( 1 - g^2 (\theta_1) \right) \right\} \\
        &+ \sum_{\substack{\vec{y}_1 \in \{0, 1\}^{n} \\ w_H(\vec{y}_1) = 0\text{\ mod\ }2 \\ y_{1, 2} = 1}}
        \left\{ (-1)^{\sum_{k=3}^{n} a_{k-1} y_{k}} (-1)^{\sum_{t=3}^n y_t t} \left( \prod_{t=3}^n g^{y_{1, t} + (y_{1, t} + y_t \mathrm{\ mod\ 2})} (\theta_t) \right) (-1)^{a_1 + 2} g(\theta_2) (-1)^{y_{1, 1}} \left( 1 - g^2 (\theta_1) \right) \right\} \\
        =& 0.
    \end{aligned}
\end{equation}
Here the last equality holds by considering a bijection between $n$-bit strings:
\begin{equation*}
    \mathcal{F} : \vec{y}_1 = (y_{1, 1} y_{1, 2} y_{1, 3}\cdots y_{1, n}) \mapsto \vec{y}_2 = ((y_{1, 1} \oplus 1)(y_{1, 2} \oplus 1)y_{1, 3} \cdots y_{1, n}),
\end{equation*}
which flips the first two bits of a bit string and does not affect parity.
Then for $\vec{y}_2 = \mathcal{F}(\vec{y}_1)$, we have
\begin{equation*}
    \begin{aligned}
        &(-1)^{\sum_{k=3}^{n} a_{k-1} y_{k}} (-1)^{\sum_{t=3}^n y_t t} \left( \prod_{t=3}^n g^{y_{2, t} + (y_{2, t} + y_t \mathrm{\ mod\ 2})} (\theta_t) \right) (-1)^{a_1 + 2} g(\theta_2) (-1)^{y_{2, 1}} \left( 1 - g^2 (\theta_1) \right) \\
        =& (-1)^{\sum_{k=3}^{n} a_{k-1} y_{k}} (-1)^{\sum_{t=3}^n y_t t} \left( \prod_{t=3}^n g^{y_{1, t} + (y_{1, t} + y_t \mathrm{\ mod\ 2})} (\theta_t) \right) (-1)^{a_1 + 2} g(\theta_2) (-1)^{y_{2, 1}} \left( 1 - g^2 (\theta_1) \right) \\
        =& (-1) \times (-1)^{\sum_{k=3}^{n} a_{k-1} y_{k}} (-1)^{\sum_{t=3}^n y_t t} \left( \prod_{t=3}^n g^{y_{1, t} + (y_{1, t} + y_t \mathrm{\ mod\ 2})} (\theta_t) \right) (-1)^{a_1 + 2} g(\theta_2) (-1)^{y_{1, 1}} \left( 1 - g^2 (\theta_1) \right).
    \end{aligned}
\end{equation*}
This gives the last equality in Eq.~\eqref{eq:s:inproof-1}.

Now suppose $y_1 = 1$.
Take into Eq.~\eqref{eq:s:simplify-positivity-1}, it can be easily verified that
\begin{equation*}
    \begin{aligned}
        &\sum_{\substack{\vec{y}_1 \in \{0, 1\}^{n} \\ w_H(\vec{y}_1) = 0\text{\ mod\ }2}} \left\{ \begin{matrix*}[l]
            (-1)^{\sum_{k=2}^{n} a_{k-1} y_{k}} (-1)^{\sum_{t=2}^n y_t t} \left[ \prod_{t=2}^n g^{y_{1, t} + (y_{1, t} + y_t \mathrm{\ mod\ 2})} (\theta_t) \right] \\
            \times \left[ \begin{matrix*}[l]
                & g^{y_{1, 1} + (y_{1, 1} + 1 \mathrm{\ mod\ 2})} (\theta_1) \\
                -& g^{(y_{1, 1} + 1 \mathrm{\ mod\ 2}) + ((y_{1, 1} + 1 \mathrm{\ mod\ 2}) + 1 \mathrm{\ mod\ 2})} (\theta_1)
            \end{matrix*} \right]
        \end{matrix*} \right\} \\
        =&\sum_{\substack{\vec{y}_1 \in \{0, 1\}^{n} \\ w_H(\vec{y}_1) = 0\text{\ mod\ }2}}
        \left\{ (-1)^{\sum_{k=2}^{n} a_{k-1} y_{k}} (-1)^{\sum_{t=2}^n y_t t} \left[ \prod_{t=2}^n g^{y_{1, t} + (y_{1, t} + y_t \mathrm{\ mod\ 2})} (\theta_t) \right] \left[ g(\theta_1) - g (\theta_1) \right] \right\} \\
        =& 0.
    \end{aligned}
\end{equation*}

Lastly, it remains to consider the case $\vec{y} = 0$, for which Eq.~\eqref{eq:s:simplify-positivity-1} transforms into:
\begin{equation*}
    \begin{aligned}
        &\sum_{\substack{\vec{y}_1 \in \{0, 1\}^{n} \\ w_H(\vec{y}_1) = 0\text{\ mod\ }2}} \left\{ \begin{matrix*}[l]
            \left[ \prod_{t=2}^n g^{y_{1, t} + (y_{1, t} + 0 \mathrm{\ mod\ 2})} (\theta_t) \right] \\
            \times \left[ \begin{matrix*}[l]
                & g^{y_{1, 1} + (y_{1, 1} + 0 \mathrm{\ mod\ 2})} (\theta_1) \\
                -& g^{(y_{1, 1} + 1 \mathrm{\ mod\ 2}) + ((y_{1, 1} + 1 \mathrm{\ mod\ 2}) + 0 \mathrm{\ mod\ 2})} (\theta_1)
            \end{matrix*} \right]
        \end{matrix*} \right\} \\
        =&\sum_{\substack{\vec{y}_1 \in \{0, 1\}^{n} \\ w_H(\vec{y}_1) = 0\text{\ mod\ }2}} \left( \prod_{t=2}^n g^{2y_{1, t}}(\theta_t) \right) (-1)^{y_{1, 1}} \left( 1 - g^2 (\theta_1)\right).
    \end{aligned}
\end{equation*}
Writing $(-1)^{y_{1, 1}} = \prod_{k=2}^n(-1)^{y_{1, k}}$, we have
\begin{equation*}
    \begin{aligned}
        &\sum_{\substack{\vec{y}_1 \in \{0, 1\}^{n} \\ w_H(\vec{y}_1) = 0\text{\ mod\ }2}} \left( \prod_{t=2}^n g^{2y_{1, t}}(\theta_t) \right) (-1)^{y_{1, 1}} \left( 1 - g^2 (\theta_1)\right) = \sum_{\substack{\vec{y}_1 \in \{0, 1\}^{n} \\ w_H(\vec{y}_1) = 0\text{\ mod\ }2}} \left( \prod_{t=2}^n (-1)^{y_{1, t}} g^{2y_{1, t}}(\theta_t) \right) \left( 1 - g^2 (\theta_1)\right) \\
        =& \prod_{t=1}^n \left( 1 - g^2 (\theta_t) \right).
    \end{aligned}
\end{equation*}
\end{proof}

\begin{remark}
  \label{remark:sign}
  The sign
  \begin{equation}
    \label{eq:sign}
    (-1)^{\sum_{k=1}^{n-1} a_k y_{k+1}} (-1)^{\sum_{t=2}^n y_t t} = (-1)^{\left[\sum_{k=1}^{n-1} (a_k + k) y_{k+1}\right] + y_1} = (-1)^{\sum_{k=2}^n (a_{k-1}+k)y_k}
  \end{equation}
  in Eq.~\eqref{eq:closed-positivity} plays crucial role in the following estimation of $\lambda_{\vec{a}}$. The binary string $\vec{a}$ controls which elements of $\vec{y}_{k+1}$ can affect the sign.
  Consider the index set
  \begin{equation*}
    \mathcal{T}_{\vec{a}} \coloneqq \left\{ k \in \{1, 2, \cdots, n\} - \{1\} \middle| a_{k-1} + k = 1 \text{\ mod\ 2} \right\},
  \end{equation*}
  which is the positions that affects the sign.
  Then since $w_H(\vec{y})$ is restricted to be even, we have the symmetry between $\mathcal{T}_{\vec{a}}$ and its complement $\mathcal{P}_{\vec{a}} = \{1, 2, \cdots, n\} - \mathcal{T}_{\vec{a}}$.
  That is, the sign in Eq.~\eqref{eq:sign} is negative if and only if $\vec{y}$ has odd number of `1' at positions in $\mathcal{T}_{\vec{a}}$, or equivalently $\mathcal{P}_{\vec{a}}$.

  For convenience, in the following discussion we denote $\mathcal{T}_{\vec{a}} = \{t_l\}_{l=1}^h$ and $\mathcal{P}_{\vec{a}} = \{j_l\}_{l=1}^d$, where $h+d=n$.
\end{remark}

\subsection{Lower Bound Estimation}

To have a lower-bound of positivity in Eq.~\eqref{eq:closed-positivity} uniformly in $\{\theta_k\}_{k=1}^n$, we first apply the following substitution for all $t = 1, \cdots n$:
\begin{equation*}
  v_t = 1 - \sin \left( \theta_t + \frac{\pi}{4} \right),
\end{equation*}
which maps the domain $[0, \pi/4]^n$ to $[0, \kappa]^n$, where $\kappa = 1 - 1/\sqrt{2}$.
More specifically, under this substitution we have
\begin{equation*}
    \begin{aligned}
        &g(\theta_t) \to 1 - \sqrt{2} \left( 1 + \sqrt{2} \right) v_t, \\
        &\sin(2\theta_t) \to 1 + 2v_t (v_t - 2), \\
        &\cos \theta_t + \sin \theta_t + g(\theta_t) \left[ \cos \theta_t + \sin \theta_t \right] =
        \sqrt{2} \left(1 + g(\theta_t) \right) \sin \left( \theta_t + \frac{\pi}{4} \right) \to 2\sqrt{2} (1 - v_t) \left( 1 - \frac{1 + \sqrt{2}}{\sqrt{2}} v_t \right) \\
        &\cos \theta_t - \sin \theta_t - g(\theta_t) (\cos \theta_t - \sin \theta_t) =
        \sqrt{2} \left( 1 - g(\theta_t) \right) \cos \left( \theta_t + \frac{\pi}{4} \right) \to 2 \left( 1 + \sqrt{2} \right) v_t \sqrt{2v_t - v_t^2}.
    \end{aligned}
\end{equation*}

The following observation further simplifies the summations in Eq.~\eqref{eq:closed-positivity}.

\begin{observation}
  \label{obv:no-bistring}
  For arbitrary expressions $\{A_t\}_{t=1}^n$ and $\{B_t\}_{t=1}^n$, we have
  \begin{equation*}
    \begin{aligned}
      &\sum_{\substack{\vec{y} \in \{0, 1\}^{n} \\ w_H(\vec{y}) = 0\text{\ mod\ }2}} \prod_{t=1}^n \left( A_t + (-1)^{y_t} B_t \right) = 2^{n-1} \left( \prod_{t=1}^n A_t + \prod_{t=1}^n B_t \right), \\
      &\sum_{\substack{\vec{y} \in \{0, 1\}^{n} \\ w_H(\vec{y}) = 1\text{\ mod\ }2}} \prod_{t=1}^n \left( A_t + (-1)^{y_t} B_t \right) = 2^{n-1} \left( \prod_{t=1}^n A_t - \prod_{t=1}^n B_t \right).
    \end{aligned}
  \end{equation*}
\end{observation}

After the substitution and by Observation~\ref{obv:no-bistring}, we have
\begin{equation}
  \label{eq:simp-positivity}
    \begin{aligned}
        &\lambda_{\vec{a}} \left( \{ v_t \}_{t=1}^n \right) \\
        =& 2 \mu^2 + \left( \frac{1 + \sqrt{2}}{\sqrt{2}} \right)^n \prod_{t=1}^n v_t \left( 1 - \frac{1 + \sqrt{2}}{\sqrt{2}} v_t \right) \\
        &- 2\mu \left\{ \left[ \prod_{l=1}^h \left( 1 - \frac{1 + \sqrt{2}}{\sqrt{2}} v_{t_l} \right) \right] \left[ \prod_{l=1}^d \frac{1 + \sqrt{2}}{\sqrt{2}} v_{j_l} \right] + \left[ \prod_{l=1}^h \frac{1 + \sqrt{2}}{\sqrt{2}} v_{t_l} \right] \left[ \prod_{l=1}^d \left( 1 - \frac{1 + \sqrt{2}}{\sqrt{2}} v_{j_l} \right) \right] \right\} \\
        &- 2^{n+2} s_n^2 \left\{ \left[ \prod_{l=1}^h (1 - v_{t_l})^2 \right] \left[ \prod_{l=1}^d v_{j_l} (2 - v_{j_l}) \right] + \left[ \prod_{l=1}^h v_{t_l} (2 - v_{t_l}) \right] \left[ \prod_{l=1}^d (1 - v_{j_l})^2 \right] \right\} \\
        &+ 2^{\frac{n+3}{2}} s_n \left\{ \begin{matrix*}[l]
            \left[ \prod_{l=1}^h \left( 1 - \frac{1 + \sqrt{2}}{\sqrt{2}} v_{t_l} \right) \left( 1 - v_{t_l} \right) \right] \left[ \prod_{l=1}^d \frac{1 + \sqrt{2}}{\sqrt{2}} v_{j_l} \sqrt{2 v_{j_l} - v_{j_l}^2} \right] \\
            + \left[ \prod_{l=1}^h \frac{1 + \sqrt{2}}{\sqrt{2}} v_{t_l} \sqrt{2 v_{t_l} - v_{t_l}^2} \right] \left[ \prod_{l=1}^{d} \left( 1 - \frac{1 + \sqrt{2}}{\sqrt{2}} v_{j_l} \right) \left( 1 - v_{j_l} \right) \right]
        \end{matrix*} \right\},
    \end{aligned}
\end{equation}
where we use the notation $\mathcal{T}_{\vec{a}} = \{t_l\}_{l=1}^h$ and $\mathcal{P}_{\vec{a}} = \{j_l\}_{l=1}^d$ in Remark~\ref{remark:sign}.

\begin{figure}[htbp]
    \centering
    \includegraphics[width=\textwidth]{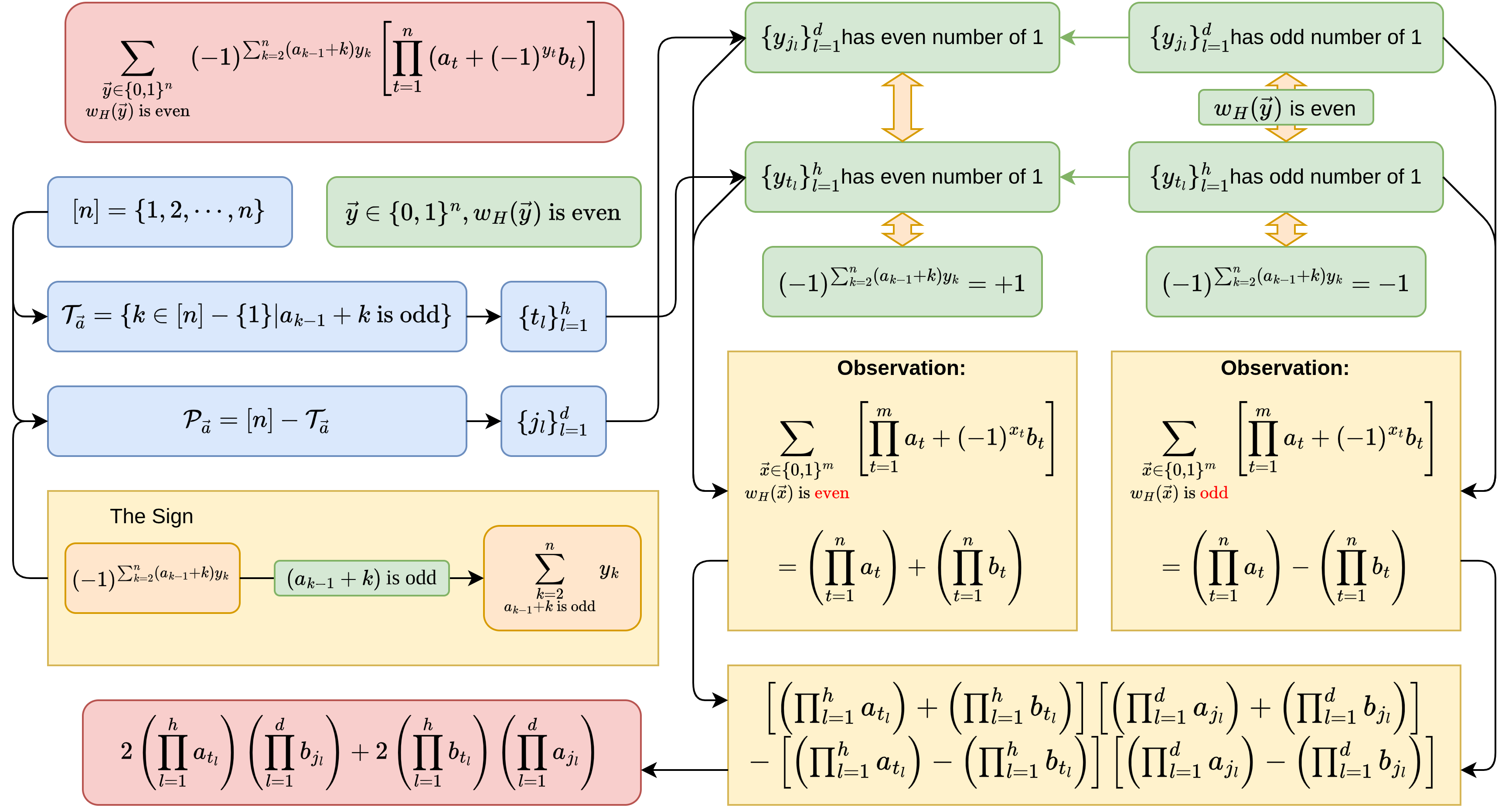}
    \caption{Sketch of simplification procedure of $\lambda_{\vec{a}}$.}
    \label{fig:simplify-sketch}
\end{figure}

Now we discuss in cases by value of $h$.

\noindent \textbf{Case 1}: $h \geq 1$.
We apply a rough estimation on the only one term with negative coefficient $-2^{n+2} s_n^2$.

\begin{observation}
  Suppose we have two sets of variables $\{c_t\}_{t=1}^j \in [0, 1/2]^j$ and $\{v_t\}_{t=1}^n \in [1/2, 1]^n$ for some integers $j$ and $n$.
  Then we have the following inequality:
  \begin{equation*}
    \left( \prod_{t=1}^j c_t \right) \left( \prod_{k=1}^n v_k \right) + \left( \prod_{t=1}^j (1 - c_t) \right) \left( \prod_{k=1}^n (1 - v_k) \right) \leq \max \left\{ \frac{1}{2^n}, \frac{1}{2^j} \right\}.
  \end{equation*}
\end{observation}
\noindent \emph{Proof}. The observation is a straightforward corollary of fundamental theorem of linear programming, that the extreme values can only be taken at vertices.
\hfill $\square$

Then, after discarding all non-constant positive terms in Eq.~\ref{eq:simp-positivity}, we have
\begin{equation*}
  \lambda_{\vec{a}} \left( \{ v_t \}_{t=1}^n \right) \geq 2\mu^2 - 2^{n+2} s_n^2 \max \left\{ \frac{1}{2^h}, \frac{1}{2^{n-h}} \right\},
\end{equation*}
Substituting $s_n = (\sqrt{2}+1)/2^{n/2}$ and $\mu = -(1+\sqrt{2})$,
$\displaystyle 2\mu^{2} - 2^{n+2}s_n^{2}\max\{2^{-h},\,2^{-(n-h)}\}
= 2(1+\sqrt{2})^{2}\bigl[1 - 2\cdot\max\{2^{-h},\,2^{-(n-h)}\}\bigr]
\geq 0$ for all $n \geq 3$ and $h \geq 1$.

\noindent \textbf{Case 2}: $h = 0$.
In this case, $\lambda_{\vec{a}}$ is a symmetric function about $\{v_t\}_{t=1}^n$.
We apply the following inequality to $\lambda_{\vec{a}}$:
\begin{equation}
  \label{eq:h0-ineq}
  \prod_{t=1}^n v_t \left( 2 - v_t \right) \leq \prod_{t=1}^n 2v_t, \qquad \forall \{v_t\}_{t=1}^n \in \left[0, \kappa\right].
\end{equation}
Then discard the product with square roots in Eq.~\eqref{eq:simp-positivity}.
Since the discarded term is manifestly non-negative for all $\{v_t\}_{t=1}^n \in [0, \kappa]^n$, we have
\begin{equation*}
  \begin{aligned}
    &\lambda_{\vec{a}} \geq P \left( \{v_t\}_{t=1}^{n} \right) \\
    \coloneqq& 2\mu^2 + \left( \frac{1 + \sqrt{2}}{\sqrt{2}} \right)^n \prod_{t=1}^n v_t \left( 1 - \frac{1 + \sqrt{2}}{\sqrt{2}} v_t \right) - 2\mu \left[ \prod_{t=1}^n \left( 1 - \frac{1 + \sqrt{2}}{\sqrt{2}} v_t \right) + \prod_{t=1}^n \frac{1 + \sqrt{2}}{\sqrt{2}} v_t \right] \\
    &- 2^{n+2} s_n^2 \left[ \prod_{t=1}^n (1 - v_t)^2 + 2^n \prod_{t=1}^n v_t \right] + 2^{\frac{n+3}{2}} s_n \prod_{t=1}^n \left( 1 - \frac{1 + \sqrt{2}}{\sqrt{2}} v_t \right) \left( 1 - v_t \right).
  \end{aligned}
\end{equation*}
The right-hand-side, $P\left( \{v_t\}_{t=1}^n \right)$, is a symmetric polynomial about $\{v_t\}_{t=1}^n$.
Take the second-order derivative about arbitrary one variable:
\begin{equation*}
  \begin{aligned}
    \frac{\partial^2 P\left( \{v_t\}_{t=1}^n \right)}{\partial v_n^2} =& - \sqrt{2} \left( 1 + \sqrt{2} \right) \left( \frac{1 + \sqrt{2}}{\sqrt{2}} \right)^n \prod_{t=1}^{n-1} \left[ v_t \left( 1 - \frac{1 + \sqrt{2}}{\sqrt{2}} v_t \right) \right] \\
    &- 2^{n+3} s_n^2 \prod_{t=1}^{n-1} (1 - v_t)^2 + 2^{\frac{n+4}{2}} \left( 1 + \sqrt{2} \right) s_n \prod_{t=1}^{n-1} \left( 1 - \frac{1 + \sqrt{2}}{\sqrt{2}} v_t \right) \left( 1 - v_t \right),
  \end{aligned}
\end{equation*}
which is independent about $v_n$.
Take the value of $\mu$ and $s_n$ in Proposition~\ref{prop:main} into this partial derivative, we have
\begin{equation*}
  \left. \frac{\partial^2 P\left( \{v_t\}_{t=1}^n \right)}{\partial v_n^2} \right|_{\mu, s_n} \leq 4 \left( 1 + \sqrt{2} \right)^2 \left\{ - 2 \left[ \prod_{t=1}^{n-1} (1 - v_t)^2 \right] + \left[ \prod_{t=1}^{n-1} \left( 1 - \frac{1 + \sqrt{2}}{\sqrt{2}} v_t \right) \left( 1 - v_t \right) \right] \right\} \leq 0.
\end{equation*}
Hence $P\left( \{v_t\}_{t=1}^n \right)$ is concave in each one variable, which implies that for arbitrary $\{v_t\}_{t=1}^n$ we have the following inequalities:
\begin{equation*}
  \label{eq:concave-decomp}
  \begin{aligned}
    &\left. P \left( \{v_t\}_{t=1}^n \right) \right|_{\mu, s_n} \\
    \geq& \frac{\kappa - v_n}{\kappa} P \left( \{v_t\}_{t=1}^{n-1} \cup \{0\} \right) + \frac{v_n}{\kappa} P \left( \{v_t\}_{t=1}^{n-1} \cup \{\kappa\} \right) \\
    \geq& \frac{(\kappa - v_{n-1}) (\kappa - v_n)}{\kappa^2} P \left( \{v_t\}_{t=1}^{n-2} \cup \{0\}^2 \right) + \frac{v_{n-1} (\kappa - v_n)}{\kappa^2} P \left( \{v_t\}_{t=1}^{n-2} \cup \{\kappa, 0\} \right) \\
    &+ \frac{(\kappa - v_{n-1})v_n}{\kappa^2} P \left( \{v_t\}_{t=1}^{n-2} \cup \{0, \kappa\} \right) + \frac{v_{n-1}v_n}{\kappa^2} P \left( \{v_t\}_{t=1}^{n-2} \cup \{\kappa\}^2 \right) \\
    \geq& \cdots
  \end{aligned}
\end{equation*}
Eventually, $P \left( \{v_t\}_{t=1}^{n} \right)$ can be lower-bounded by convex combination of values on vertices.
So it suffices to verify the non-negativity of $P \left( \{v_t\}_{t=1}^{n}\right)$ on vertices for $\mu$ and $s_n$ given in Proposition~\ref{prop:main}.
Note that its symmetry about variables simplifies the verification.
For integer $1 \leq j \leq n-1$ (i.e.\ at least one $v_t = 0$), the terms containing $\prod v_t$ or $\prod \alpha v_t$ vanish, yielding
\begin{equation*}
   \left. P \left( \{\kappa\}^j \cup \{0\}^{n-j} \right) \right|_{\mu, s_n} = 2 \left( 1 + \sqrt{2} \right)^2 + 2 \left( 1 + \sqrt{2} \right) \cdot 2^{-j} - 4 \left( 1 + \sqrt{2} \right)^2 \cdot 2^{-j} + 2 \left( 2 + \sqrt{2} \right) \cdot 2^{-\frac{3j}{2}}.
\end{equation*}
To verify its monotonicity, regard $j$ as a real variable and differentiate:
\begin{equation*}
  \begin{aligned}
    \frac{d}{dj}\,
    \Bigl. P \bigl( \{\kappa\}^j \cup \{0\}^{n-j} \bigr) \Bigr|_{\mu,s_n}
    &=\ln(2)\,2^{-j}\Bigl[
      4(1+\sqrt{2})^{2} - 2(1+\sqrt{2})
      - 3\bigl(2+\sqrt{2}\bigr)\,2^{-j/2}
    \Bigr] \\
    &\geq \ln(2)\,2^{-j}\Bigl[
      4(1+\sqrt{2})^{2} - 2(1+\sqrt{2})
      - 3\bigl(2+\sqrt{2}\bigr)
    \Bigr],
  \end{aligned}
\end{equation*}
where the inequality uses $2^{-j/2} \leq 1$ for $j \geq 0$.
The bracketed quantity simplifies to
$4 + 3\sqrt{2} > 0$, hence the derivative is positive for all $j \geq 0$, establishing strict monotonic increase on $[1,+\infty)$.
Also, we have the following boundary values:
\begin{equation*}
  \begin{aligned}
    \left. P \left( \{0\}^n \right) \right|_{\mu, s_n} = 0,\quad \left. P \left( \{\kappa\} \cup \{0\}^{n-1} \right) \right|_{\mu, s_n} = 2 \left( 1 + \sqrt{2} \right) > 0,
  \end{aligned}
\end{equation*}
The monotonicity argument together with $P(\{\kappa\}\cup\{0\}^{n-1})>0$ establishes non-negativity for all $1 \leq j \leq n-1$.

For the remaining vertex $j=n$ (all $v_t = \kappa$), the terms that vanished previously now contribute.
A direct evaluation gives
\begin{equation}
  \label{eq:P-eta-n}
  \begin{aligned}
    \left. P\left( \{\kappa\}^n \right) \right|_{\mu, s_n}
    =&\ 2\left(1+\sqrt{2}\right)^{2} + 4^{-n} + 4\left(1+\sqrt{2}\right)2^{-n} - 4\left(1+\sqrt{2}\right)^{2}2^{-n} \\
    &- 4\left(1+\sqrt{2}\right)^{2}\left(2-\sqrt{2}\right)^{n} + 2\left(2+\sqrt{2}\right)2^{-\frac{3n}{2}} \\
    =&\ 2\left(1+\sqrt{2}\right)^{2} + 4^{-n} - 4\left(1+\sqrt{2}\right)\sqrt{2}\cdot 2^{-n} - 4\left(1+\sqrt{2}\right)^{2}\left(2-\sqrt{2}\right)^{n} + 2\left(2+\sqrt{2}\right)2^{-\frac{3n}{2}}.
  \end{aligned}
\end{equation}
Discarding the manifestly positive terms $4^{-n}$ and $2(2+\sqrt{2})\,2^{-3n/2}$ gives the lower bound
\begin{equation*}
  P\bigl(\{\kappa\}^{n}\bigr) > 2\left(1+\sqrt{2}\right)^{2} - 4\left(1+\sqrt{2}\right)\sqrt{2}\cdot 2^{-n} - 4\left(1+\sqrt{2}\right)^{2}\left(2-\sqrt{2}\right)^{n}.
\end{equation*}
For $n \geq 3$, the decreasing functions $2^{-n}$ and $(2-\sqrt{2})^{n}$ attain their maxima at $n=3$.
Hence
\begin{equation*}
  P\bigl(\{\kappa\}^{n}\bigr) > 2\left(1+\sqrt{2}\right)^{2} - 4\left(1+\sqrt{2}\right)\sqrt{2}\cdot 2^{-3} - 4\left(1+\sqrt{2}\right)^{2}\left(2-\sqrt{2}\right)^{3} = -11 + \frac{23}{2}\sqrt{2}.
\end{equation*}
Since $(23\sqrt{2}/2)^{2}=529/2 = 264.5 > 121 = 11^{2}$, we have $23\sqrt{2}/2 > 11$, and therefore $P(\{\kappa\}^{n}) > 0$ for all $n \geq 3$.
Together with the boundary values $j=0,1$ and the monotonicity argument for $1 \leq j \leq n-1$, all vertex values are non-negative, which completes the proof of Proposition~\ref{prop:main}.

\section{Comparison with Previous Analytic Robust Self-Testing Bounds}
\label{sec:comparison}

This section supplies the formal derivation of the comparison displayed in Fig.~\ref{fig:tolerance_bound_qubits} of the main text: the largest relative violation deficiency $\eta \coloneqq \varepsilon/\beta_{nQ}$ for which a certified extractability $\Xi \geq 1/2$ of the $n$-qubit GHZ state can be guaranteed, plotted as a function of the number of parties $n$ for the bound of Result~\ref{result:main} and for McKague's graph-state self-tests~\cite{McKague_2011,McKague_2016}.

\subsection{White-Noise Benchmark and Identification of the Deviation Parameters}

All three bounds are evaluated on the same one-parameter family of white-noise states
\begin{equation}
  \label{eq:sm:comparison:white-noise}
  \rho_p \coloneqq p\, |\text{GHZ}_n\rangle\langle\text{GHZ}_n| + (1-p)\,\frac{I}{2^n}, \qquad p \in [0,1].
\end{equation}
For this benchmark the parties are equipped with the measurement setup attaining the maximal violation, for which every correlation term entering the Bell operator $W_n$ of Section~\ref{sec:proof-main} is a product of traceless observables; hence $\operatorname{tr}[W_n I/2^{n}] = 0$, $\operatorname{tr}[W_n\rho_p] = p\,\beta_{nQ}$, and the relative violation deficiency reads
\begin{equation}
  \label{eq:sm:comparison:def-eta}
  \eta \coloneqq \frac{\beta_{nQ} - \operatorname{tr}[W_n \rho_p]}{\beta_{nQ}} = 1 - p .
\end{equation}

To instantiate McKague's tests on the same family, recall that the graph state $|K_n\rangle$ of the complete graph on $n$ vertices is locally Clifford-equivalent to $|\text{GHZ}_n\rangle$: local complementation at any vertex $v$ maps $K_n$ to the star graph $K_{1,n-1}$, and the star-graph state is transformed into $|\text{GHZ}_n\rangle$ by a Hadamard on each leaf qubit. Hence
\begin{equation}
  \label{eq:sm:comparison:lc}
  |K_n\rangle = U\,|\text{GHZ}_n\rangle, \qquad U = \bigotimes_{k=1}^n U_k ,
\end{equation}
with $U$ a product of single-qubit Clifford unitaries. McKague's complete-graph test estimates the $n$ stabilizer correlations
\begin{equation*}
  S_v \coloneqq X_v Z^{N_v} = X_v \otimes \Bigl( \bigotimes_{u \neq v} Z_u \Bigr), \qquad v \in V,
\end{equation*}
together with a triangle correlation: for a triangle $V' = \{a, b, c\}$ one measures $X^{V'} Z^{N(V')} = X_a X_b X_c \otimes \bigl( \bigotimes_{u \notin V'} Z_u \bigr)$, since every vertex outside $V'$ has all three triangle vertices as neighbours. The stabilizer operators satisfy $S_v |K_n\rangle = |K_n\rangle$ for every $v$ by definition, and Lemma~1 of~\cite{McKague_2011}, applied to the triangle $V'$ (each of whose vertices has even degree in the induced subgraph, with $|E'| = 3$), gives $X^{V'} Z^{N(V')} |K_n\rangle = -|K_n\rangle$. Every tested operator is a traceless tensor product of Pauli $X$ and $Z$ observables; consequently, on the locally rotated state
\begin{equation*}
  \rho_p' \coloneqq U\rho_p U^\dagger = p\, |K_n\rangle\langle K_n| + (1-p)\,\frac{I}{2^n},
\end{equation*}
each tested correlation equals $\pm p$, and the per-correlation deviation $\epsilon_M$ that governs McKague's robustness theorems reads
\begin{equation}
  \label{eq:sm:comparison:eps-mck}
  \epsilon_M = 1 - p = \eta .
\end{equation}
Moreover, extractability is invariant under local unitaries, $\Xi(\rho_p \to |\text{GHZ}_n\rangle\langle\text{GHZ}_n|) = \Xi(\rho_p' \to |K_n\rangle\langle K_n|)$, so the three bounds are directly comparable as functions of the same parameter $\eta$.
We stress that Result~\ref{result:main} is not restricted to this family: it holds for \emph{every} state attaining the observed violation $\beta$.
The white-noise family~\eqref{eq:sm:comparison:white-noise} merely provides the common benchmark on which all three tests can be evaluated exactly.
We also note that the comparison is naturally phrased in terms of the \emph{relative} deficiency $\eta$: in absolute terms every certified tolerance grows exponentially in $n$ ($\Theta(2^{n/2})$ for the bound of Result~\ref{result:main} and as $\Theta(n^{-4}2^{n/2})$ for the McKague tests) and only the relative deficiency is a scale-invariant measure of how close to the maximal violation the observed statistics must lie.

\subsection{Bound of This Work}

From Result~\ref{result:main}, $\Xi \geq s_n \beta + \mu$ with $s_n = \left(\sqrt{2}+1\right)/2^{n/2}$ and $\mu = -\left(1+\sqrt{2}\right)$. Substituting $s_n \beta_{nQ} + \mu = 1$ and $s_n = \left(2+\sqrt{2}\right)/\beta_{nQ}$,
\begin{equation}
  \label{eq:sm:comparison:ours}
  1 - \Xi\bigl(\rho_p \to |\text{GHZ}_n\rangle\langle\text{GHZ}_n|\bigr) \leq \left(2+\sqrt{2}\right) \eta .
\end{equation}
For the conjectured-optimal parameters $\bar{s}_n$, $\bar{\mu}$ of Eq.~\eqref{eq:upper-bound-params} (proved optimal for $n = 3, 4, 5$~\cite{Kaniewski_2016,Jha_Singh_Pan_2026}, numerically established up to $n \leq 7$~\cite{Baccari_Augusiak_Supic_Tura_Acin_2020}, and supported by the evidence of Section~\ref{sec:numeric} up to $n = 100$), the coefficient halves:
\begin{equation}
  \label{eq:sm:comparison:ours-opt}
  1 - \Xi\bigl(\rho_p \to |\text{GHZ}_n\rangle\langle\text{GHZ}_n|\bigr) \leq \frac{2+\sqrt{2}}{2}\, \eta .
\end{equation}
Hence extractability $\Xi \geq 1/2$ is certified for every $n$ as soon as
\begin{equation}
  \label{eq:sm:comparison:thr-ours}
  \eta \leq \eta_{\rm prov}^{\star} \coloneqq \frac{1}{2\left(2+\sqrt{2}\right)} = \frac{2-\sqrt{2}}{4} \approx 0.1464,
  \qquad
  \eta \leq \eta_{\rm opt}^{\star} \coloneqq \frac{1}{2+\sqrt{2}} = \frac{2-\sqrt{2}}{2} \approx 0.2929,
\end{equation}
respectively; both thresholds are independent of $n$.

\subsection{McKague's Complete-Graph Tests}

\emph{Bound of~\cite{McKague_2011} (reference experiment~1).}
Theorem~2 of~\cite{McKague_2011} states that if a compatible physical experiment $\epsilon_M$-simulates reference experiment~1 on a connected graph $G$ with $n$ vertices (that is, every tested correlation deviates from its ideal value by at most $\epsilon_M$), then it is $\delta$-equivalent to the reference experiment: there exist a local isometry $\Phi = \bigotimes_{v \in V} \Phi_v$ and a state $|\text{junk}\rangle$ such that $\bigl\| \Phi(|\psi'\rangle) - |\text{junk}\rangle \otimes |\psi_G\rangle \bigr\|_1 \leq \delta$ (and likewise for the tested measurements). The proof in Section~3.2 of~\cite{McKague_2011} provides the graph-specific estimate
\begin{equation}
  \label{eq:sm:comparison:mck11-delta}
  \delta = \Bigl( \left(4l + c + 1\right)\Bigl( n + \frac{|E|}{2} \Bigr) + n \Bigr) \sqrt{\epsilon_M},
\end{equation}
where $c$ is the size of the distinguished odd induced cycle and $l$ the largest distance from it to any vertex. For $K_n$: $c = 3$ (a triangle), $l = 1$, and $|E| = n(n-1)/2$, so
\begin{equation}
  \label{eq:sm:comparison:mck11-delta-kn}
  \delta_{11} = \left(2n^{2} + 7n\right) \sqrt{\eta}.
\end{equation}
For a vector the trace norm coincides with the Euclidean norm, so with $|a\rangle \coloneqq \Phi(|\psi'\rangle)$ and $|b\rangle \coloneqq |\text{junk}\rangle \otimes |K_n\rangle$ one has $\operatorname{Re}\langle a|b\rangle \geq 1 - \delta^2/2$ and hence
\begin{equation}
  1 - F = 1 - \bigl|\langle a|b\rangle\bigr|^{2} \leq 1 - \Bigl( 1 - \frac{\delta^{2}}{2} \Bigr)^{\!2} = \delta^{2} - \frac{\delta^{4}}{4} \leq \delta^{2}.
\end{equation}
Purifying $\rho_p'$, applying the theorem to the purification, and tracing out the junk register, which converts each local isometry into a local channel and cannot decrease the fidelity, yields
\begin{equation}
  \label{eq:sm:comparison:mck11-fid}
  1 - \Xi\bigl(\rho_p \to |\text{GHZ}_n\rangle\langle\text{GHZ}_n|\bigr) \leq \delta_{11}^{2} = \left(2n^{2} + 7n\right)^{2} \eta .
\end{equation}
The curve plotted in Fig.~\ref{fig:tolerance_bound_qubits} follows from the sharp condition $\delta_{11}^{2} \leq 1/2$, i.e.\
\begin{equation}
  \label{eq:sm:comparison:thr-mck11}
  \eta \leq \eta_{\rm M11}^{\star}(n) \coloneqq \frac{1}{2\left(2n^{2} + 7n\right)^{2}} = \Theta\bigl(n^{-4}\bigr).
\end{equation}
(An even simpler sufficient condition, $\delta_{11} \leq 1/2$, certifies the stronger target $\Xi \geq 3/4$ through $1 - F \leq \delta_{11}^{2} \leq 1/4$ and corresponds to the smaller threshold $\eta \leq 1/\bigl(4\left(2n^{2}+7n\right)^{2}\bigr)$.)

\emph{Bound of~\cite{McKague_2016} (Theorem~3).}
Theorem~3 of~\cite{McKague_2016} is stated for triangular-lattice graphs, but as noted in Section~3.1.1 of that paper, its proof only requires a family of triangles covering the vertex set; it therefore applies to every graph in which each vertex lies in a triangle, in particular to $K_n$ for all $n \geq 3$. Let $G$ be a graph on $n$ vertices with adjacency matrix $\mathbf{A}$ in which every vertex lies in a triangle. If $\langle\psi'|S'_v|\psi'\rangle \geq 1 - \epsilon_M$ for every $v$ and $-\langle\psi'|X'^{\tau}Z'^{\mathbf{A}\tau}|\psi'\rangle \geq 1 - \epsilon_M$ for every triangle $\tau$ with characteristic vector $\tau$, then there exist a local isometry $\Phi$ and a state $|\text{junk}\rangle$ such that
\begin{equation}
  \label{eq:sm:comparison:mck16-thm}
  \Bigl\| \Phi\bigl( X'^{q} Z'^{p} |\psi'\rangle \bigr) - |\text{junk}\rangle \otimes X^{q} Z^{p} |G\rangle \Bigr\|_{2}
  \leq \Bigl( 2\sqrt{p\cdot p} + 2\sqrt{2n} + \sqrt{|E| + n} \,\Bigr) \left( 2\epsilon_M \right)^{\frac{1}{4}}
\end{equation}
for all $p, q \in \{0, 1\}^{n}$. For the state ($p = q = 0$) and $G = K_n$, with $|E| = n(n-1)/2$, set
\begin{equation}
  \label{eq:sm:comparison:def-A}
  A_n \coloneqq 2\sqrt{2n} + \sqrt{\frac{n(n+1)}{2}},
  \qquad
  \delta_{3} \coloneqq A_n \left( 2\eta \right)^{\frac{1}{4}},
\end{equation}
so that $1 - F \leq \delta_{3}^{2}$ and, after purification and tracing out the junk register,
\begin{equation}
  \label{eq:sm:comparison:mck16-fid}
  1 - \Xi\bigl(\rho_p \to |\text{GHZ}_n\rangle\langle\text{GHZ}_n|\bigr) \leq \delta_{3}^{2} = A_n^{2} \left( 2\eta \right)^{\frac{1}{2}} .
\end{equation}
The sharp threshold for $\Xi \geq 1/2$ is
\begin{equation}
  \label{eq:sm:comparison:thr-mck16}
  \eta \leq \eta_{\rm M16}^{\star}(n) \coloneqq \frac{1}{8 A_n^{4}} = \Theta\bigl(n^{-4}\bigr),
\end{equation}
(An even simpler sufficient condition, $\delta_3 \leq 1/2$, certifies the stronger target $\Xi \geq 3/4$ and corresponds to the smaller threshold $\eta \leq 1/\bigl(32 A_n^{4}\bigr)$, a constant factor of $4$ below the sharp value.) We note that the hypotheses of Theorem~3, as stated, involve the $n$ stabilizer correlations and one triangle correlation per triangle of the graph (for $K_n$ this is $\binom{n}{3} = \Theta(n^{3})$ correlations), although the triangle-cover remark shows that $O(n)$ of them suffice for the proof. In the plotted range $3 \leq n \leq 40$ the 2011 bound certifies the same state from only $n + 1$ correlations and is tighter than the 2016 bound; the two curves cross only around $n \approx 84$, where $1/\bigl(2(2n^{2}+7n)^{2}\bigr) = 1/\bigl(8 A_n^{4}\bigr)$.

\emph{The star-graph test.}
For completeness, reference experiment~2 of~\cite{McKague_2011} applied to the star graph $K_{1,n-1}$, which is the GHZ state itself up to leaf Hadamards, carries the bound
\begin{equation}
  \label{eq:sm:comparison:star}
  \delta_{\rm star} = \left(7n - 2\right) \sqrt{\eta} + \frac{13}{2}\left(3n - 1\right) \eta^{\frac{1}{4}},
\end{equation}
obtained from the estimate $\delta = \bigl(2l\left(2n+|E|\right) + n\bigr)\sqrt{\epsilon_M} + 13\left(n + |E|/2\right)\epsilon_M^{1/4}$ of Section~3.3 of~\cite{McKague_2011} with the distinguished vertex taken as the centre of the star, so that $l = 1$, and $|E| = n - 1$. This test is strictly dominated by the complete-graph test: certification requires $\delta_{\rm star} \leq 1/\sqrt{2}$, which forces $\eta \leq 4/\bigl(13(3n-1)\bigr)^{4}$ since $\delta_{\rm star} \geq \tfrac{13}{2}(3n-1)\,\eta^{1/4}$, and
\begin{equation}
  \frac{1}{2\left(2n^{2}+7n\right)^{2}} \cdot \frac{\bigl(13(3n-1)\bigr)^{4}}{4}
  = \frac{\bigl(13(3n-1)\bigr)^{4}}{8\left(2n^{2}+7n\right)^{2}} > 10^{3}
\end{equation}
for every $n \geq 3$ (the ratio is increasing in $n$, with value $\approx 9.6 \times 10^{3}$ at $n = 3$). The star-graph curve is therefore omitted from Fig.~\ref{fig:tolerance_bound_qubits}.

\subsection{Summary}

Combining Eqs.~\eqref{eq:sm:comparison:thr-ours}, \eqref{eq:sm:comparison:thr-mck11} and \eqref{eq:sm:comparison:thr-mck16}, Table~\ref{tab:sm:comparison} lists the certified tolerance for a few representative values of $n$. The gap between the $n$-independent tolerance of Result~\ref{result:main} and the tolerances of the McKague tests grows quartically in $n$:
\begin{equation}
  \label{eq:sm:comparison:ratio}
  \frac{\eta_{\rm prov}^{\star}}{\eta_{\rm M11}^{\star}(n)} = \frac{2-\sqrt{2}}{2}\,\bigl(2n^{2}+7n\bigr)^{2},
  \qquad
  \frac{\eta_{\rm prov}^{\star}}{\eta_{\rm M16}^{\star}(n)} = 2\left(2-\sqrt{2}\right) A_n^{4},
\end{equation}
both of order $\Theta(n^{4})$.
For instance, at $n = 10$ the proved bound certifies extractability $\Xi \geq 1/2$ for a relative deficiency $\eta \approx 0.146$, which is over four orders of magnitude larger than what either McKague bound tolerates, and at $n = 40$ the gap exceeds six orders of magnitude.
In particular, at relative deficiency $\eta = 10^{-4}$ the 2011 bound certifies $\Xi \geq 1/2$ only for $n \leq 4$, and the 2016 bound for no $n \geq 3$ at all, consistent with the statement in the main text that under this noise level none of the previously known analytic bounds certifies entanglement for large $n$.

\begin{table}[htbp]
  \centering
  \caption{Largest relative violation deficiency $\eta = \varepsilon/\beta_{nQ}$ for which extractability $\Xi \geq 1/2$ of the $n$-qubit GHZ state is certified, for the proved bound of Result~\protect\ref{result:main} ($\eta_{\rm prov}^{\star}$), the conjectured-optimal bound ($\eta_{\rm opt}^{\star}$), and the sharp forms $\eta_{\rm M11}^{\star}(n)$, $\eta_{\rm M16}^{\star}(n)$ of McKague's complete-graph bounds of Eqs.~\protect\eqref{eq:sm:comparison:thr-mck11} and \protect\eqref{eq:sm:comparison:thr-mck16}. The curves of Fig.~\protect\ref{fig:tolerance_bound_qubits} plot exactly these sharp thresholds.}
  \label{tab:sm:comparison}
  \begin{tabular}{cccccc}
    \toprule
    $n$ & $\eta_{\rm prov}^{\star}$ & $\eta_{\rm opt}^{\star}$ & $\eta_{\rm M11}^{\star}(n)$ & $\eta_{\rm M16}^{\star}(n)$ & $\eta_{\rm prov}^{\star}/\eta_{\rm M16}^{\star}(n)$ \\
    \midrule
    3  & $0.1464$ & $0.2929$ & $3.287\times10^{-4}$ & $4.287\times10^{-5}$ & $3.42\times10^{3}$ \\
    5  & $0.1464$ & $0.2929$ & $6.920\times10^{-5}$ & $1.156\times10^{-5}$ & $1.27\times10^{4}$ \\
    10 & $0.1464$ & $0.2929$ & $6.859\times10^{-6}$ & $1.745\times10^{-6}$ & $8.39\times10^{4}$ \\
    20 & $0.1464$ & $0.2929$ & $5.659\times10^{-7}$ & $2.304\times10^{-7}$ & $6.36\times10^{5}$ \\
    40 & $0.1464$ & $0.2929$ & $4.129\times10^{-8}$ & $2.668\times10^{-8}$ & $5.49\times10^{6}$ \\
    \bottomrule
  \end{tabular}
\end{table}

\section{Efficient Verification of Optimal Lower Bound}
\label{sec:numeric}

In this section, we discuss the bottleneck of the proof of operator inequality~\ref{eq:operator-ineq-main-sm} with optimal parameters $\overline{s}_n$ and $\overline{\mu}$ defined by
\begin{equation}
  \label{eq:upper-bound-params-sm}
  \bar{s}_n \coloneqq \frac{1 + \sqrt{2}}{2^{\frac{n}{2} + 1}},\quad \bar{\mu} \coloneqq - \frac{1}{\sqrt{2}},
\end{equation}

To do this, simply replace $s_n$ and $\mu$ with $\overline{s}_n$ and $\overline{\mu}$ in the simplified expression of $\lambda_{\vec{a}} \left( \{v_t\}_{t=1}^n \right)$ in Eq.~\eqref{eq:simp-positivity}.
We have
\begin{equation}
  \begin{aligned}
    &\lambda_{\vec{a}}\left( \left\{ v_t \right\}_{t=1}^n \right) \\
    =&1 + \left( \frac{1 + \sqrt{2}}{\sqrt{2}} \right)^n \prod_{t=1}^n v_t \left( 1 - \frac{1 + \sqrt{2}}{\sqrt{2}} v_t \right) \\
    &+ \sqrt{2} \left\{ \left[ \prod_{l=1}^h \left( 1 - \frac{1 + \sqrt{2}}{\sqrt{2}} v_{t_l} \right) \right] \left[ \prod_{l=1}^d \frac{1 + \sqrt{2}}{\sqrt{2}} v_{j_l} \right] + \left[ \prod_{l=1}^h \frac{1 + \sqrt{2}}{\sqrt{2}} v_{t_l} \right] \left[ \prod_{l=1}^d \left( 1 - \frac{1 + \sqrt{2}}{\sqrt{2}} v_{j_l} \right) \right] \right\} \\
    &- \left( \sqrt{2} + 1 \right)^2 \left\{ \left[ \prod_{l=1}^h (1 - v_{t_l})^2 \right] \left[ \prod_{l=1}^d v_{j_l} (2 - v_{j_l}) \right] + \left[ \prod_{l=1}^h v_{t_l} (2 - v_{t_l}) \right] \left[ \prod_{l=1}^d (1 - v_{j_l})^2 \right] \right\} \\
    &+ \left( 2 + \sqrt{2} \right) \left\{ \begin{matrix*}[l]
        \left[ \prod_{l=1}^h \left( 1 - \frac{1 + \sqrt{2}}{\sqrt{2}} v_{t_l} \right) \left( 1 - v_{t_l} \right) \right] \left[ \prod_{l=1}^d \frac{1 + \sqrt{2}}{\sqrt{2}} v_{j_l} \sqrt{2 v_{j_l} - v_{j_l}^2} \right] \\
        + \left[ \prod_{l=1}^h \frac{1 + \sqrt{2}}{\sqrt{2}} v_{t_l} \sqrt{2 v_{t_l} - v_{t_l}^2} \right] \left[ \prod_{l=1}^{d} \left( 1 - \frac{1 + \sqrt{2}}{\sqrt{2}} v_{j_l} \right) \left( 1 - v_{j_l} \right) \right]
    \end{matrix*} \right\} \\
    \geq& 1 - \left( \sqrt{2} + 1 \right)^2 \left\{ \left[ \prod_{l=1}^h (1 - v_{t_l})^2 \right] \left[ \prod_{l=1}^d v_{j_l} (2 - v_{j_l}) \right] + \left[ \prod_{l=1}^h v_{t_l} (2 - v_{t_l}) \right] \left[ \prod_{l=1}^d (1 - v_{j_l})^2 \right] \right\} \\
    \geq& 1 - \left( \sqrt{2} + 1 \right)^2 \max \left\{ \frac{1}{2^h}, \frac{1}{2^d} \right\}.
  \end{aligned}
\end{equation}

Then we can still discuss by cases of values of $h$.
For $n \le 5$ the optimal bound is already established analytically~\cite{Kaniewski_2016,Jha_Singh_Pan_2026}, so in the following we assume $n \ge 6$.
Moreover, $\lambda_{\vec{a}}$ in Eq.~\eqref{eq:simp-positivity} is symmetric under exchanging the two index sets $\mathcal{T}_{\vec{a}}$ and $\mathcal{P}_{\vec{a}}$ (cf.~Remark~\ref{remark:sign}), and the domain $[0, \kappa]^n$ is invariant under permutations of the variables, so the minimum of $\lambda_{\vec{a}}$ over $[0, \kappa]^n$ is invariant under exchanging $\mathcal{T}_{\vec{a}}$ and $\mathcal{P}_{\vec{a}}$. Without loss of generality we may therefore assume $h \le d = n - h$.

\noindent \textbf{Case 1}: $h \geq 3$.
By the assumption $h \le d$ this implies $d \ge 3$ as well, so $\max\{2^{-h}, 2^{-d}\} \le 2^{-3} = 1/8$; since $1 \ge \left( \sqrt{2} + 1 \right)^2 / 8$, in this case we have $\lambda_{\vec{a}} \left( \{v_t\}_{t=1}^n \right) \geq 1 - \left( \sqrt{2} + 1 \right)^2 / 8 \geq 0$.

\noindent \textbf{Case 2}: $h = 0$. In this case, we can still apply the inequality in Eq.~\eqref{eq:h0-ineq} and lower bound $\lambda_{\vec{a}} \left( \{v_t\}_{t=1}^n \right)$ with a symmetric polynomial $P\left( \{v_t\}_{t=1}^n \right)$:
\begin{equation*}
  \begin{aligned}
    &\lambda_{\vec{a}} \left( \{v_t\}_{t=1}^n \right) \geq \overline{P} \left( \{ v_t \}_{t=1}^n \right) \\
    \coloneqq& 1 + \left( \frac{1 + \sqrt{2}}{\sqrt{2}} \right)^n \prod_{t=1}^n v_t \left( 1 - \frac{1 + \sqrt{2}}{\sqrt{2}} v_t \right) +\sqrt{2} \left[ \prod_{t=1}^n \left( 1 - \frac{1 + \sqrt{2}}{\sqrt{2}} v_t \right) + \prod_{t=1}^n \frac{1 + \sqrt{2}}{\sqrt{2}}v_t \right] \\
    &-\left( \sqrt{2} + 1 \right)^2 \left[ \prod_{t=1}^n \left( 1 - v_{t} \right)^2 + 2^n \prod_{t=1}^n v_t \right] + \left( 2 + \sqrt{2} \right) \prod_{t=1}^n \left( 1 - \frac{1 + \sqrt{2}}{\sqrt{2}} v_t \right)\left( 1 - v_t \right).
  \end{aligned}
\end{equation*}
Calculate the second-order partial derivative about one variable:
\begin{equation*}
  \begin{aligned}
    \frac{\partial^2 \overline{P} \left( \{v_t\}_{t=1}^n \right)}{\partial v_n^2} =& -\sqrt{2} \left( 1 + \sqrt{2} \right) \left( \frac{1 + \sqrt{2}}{\sqrt{2}} \right)^n \prod_{t=1}^{n-1} v_t \left( 1 - \frac{1 + \sqrt{2}}{\sqrt{2}} v_t \right) \\
    &- 2\left( \sqrt{2} + 1 \right)^2 \prod_{t=1}^{n-1} \left( 1 - v_t \right)^2 + 2 \left( \sqrt{2} + 1 \right)^2 \prod_{t=1}^{n-1} \left( 1 - \frac{1 + \sqrt{2}}{\sqrt{2}} v_t \right) \left( 1 - v_t \right) \leq 0.
  \end{aligned}
\end{equation*}
Still, since $1 - v_t \geq 1 - \frac{1 + \sqrt{2}}{\sqrt{2}} v_t$ for all $v_t \in [0, \kappa]$, this second-order partial derivative is still non-positive.
So we can follows the decomposition in Eq.~\eqref{eq:concave-decomp} and lower bound $\overline{P} \left( \{v_t\}_{t=1}^n \right)$ by convex combination of values on vertices.
Straightforward calculation shows $\overline{P} \left( \{0\}^n \right) = 0$ and, for all $1 \leq j \leq n-1$,
\begin{equation*}
  \begin{aligned}
    &\overline{P}(\{\kappa\}^j \cup \{0\}^{n-j}) \\
    =& 1 + \sqrt{2} \cdot 2^{-j} - \left( 1 + \sqrt{2} \right)^2 \cdot 2^{-j} + \left( 2 + \sqrt{2} \right) \cdot 2^{-\frac{3j}{2}}.
  \end{aligned}
\end{equation*}
This function about $j$ takes value of $0$ at $j=1$, and monotonically increasing in $j$ on $(1, + \infty)$.

For the remaining vertex $j=n$ (all $v_t = \kappa$), the terms that vanished for $j \leq n-1$ now contribute.
A direct evaluation of the original $\lambda_{\vec{a}}$ at this vertex gives
\begin{equation}
  \label{eq:P-bar-eta-n}
  \begin{aligned}
    \lambda_{\vec{a}}\left( \{\kappa\}^n \right)
    =&\ 1 + 4^{-n} + 2\sqrt{2}\cdot 2^{-n}
       - 2\left(1+\sqrt{2}\right)^{2} 2^{-n}
       + 2\left(2+\sqrt{2}\right) 2^{-\frac{3n}{2}} \\
    =&\ 1 + 4^{-n}
       - \left(6 + 2\sqrt{2}\right) 2^{-n}
       + 2\left(2+\sqrt{2}\right) 2^{-\frac{3n}{2}}.
  \end{aligned}
\end{equation}
Discarding the manifestly positive terms $4^{-n}$ and $2(2+\sqrt{2})\,2^{-3n/2}$ gives the lower bound
\begin{equation*}
  \lambda_{\vec{a}}\bigl(\{\kappa\}^{n}\bigr) > 1 - \left(6 + 2\sqrt{2}\right) 2^{-n}.
\end{equation*}
For $n \geq 4$, the decreasing function $2^{-n}$ attains its maximum at $n = 4$, yielding
\begin{equation*}
  \lambda_{\vec{a}}\bigl(\{\kappa\}^{n}\bigr) > 1 - \frac{6 + 2\sqrt{2}}{2^{4}} = \frac{5 - \sqrt{2}}{8} > 0.
\end{equation*}
Since $n = 3$ is covered by the known optimal result of~\cite{Kaniewski_2016}, all vertex values are non-negative for $n \geq 3$, hence $\lambda_{\vec{a}}\left( \{v_t\}_{t=1}^n \right) \geq 0$ if $h=0$.

\noindent \textbf{Case 3}: $h = 1$ or $2$.
The proof strategy on small systems proposed in \cite{Kaniewski_2016} tries to establish inequality $\lambda_{\vec{a}} - \lambda_{\vec{a}^+} \geq 0$, where $\vec{a}^+$ is the binary string satisfies $h = 0$.
However, this inequality could fail on larger systems for $\vec{a}$ satisfies $h=1$.

The remaining obstacle is therefore to verify non-negativity of $\lambda_{\vec{a}}$ for $h = 1, 2$ and arbitrary $n \geq 3$.
Since $\lambda_{\vec{a}}$ is block-symmetric in $\{v_{t_l}\}_{l=1}^h$ and $\{v_{j_l}\}_{l=1}^d$, the verification domain can be reduced to a \emph{fundamental domain} $\mathcal{D}_{h, n}^{\rm fund}$ defined by ordering constraints within each block:
\begin{equation*}
  \mathcal{D}_{h, n}^{\rm fund} = \left\{ 
    \{v_{t_l}\}_{l=1}^h \cup \{v_{j_l}\}_{l=1}^d \in [0, \kappa]^n 
    \;\middle|\;
    v_{t_1} \leq \cdots \leq v_{t_h},\;
    v_{j_1} \leq \cdots \leq v_{j_d}
  \right\}.
\end{equation*}
The full hypercube $[0, \kappa]^n$ is the union of $h! \cdot d!$ congruent copies of $\mathcal{D}_{h, n}^{\rm fund}$ under the action of the block-symmetry group.
Consequently, a function attaining its global minimum at some point in the full domain must attain the same minimum at the corresponding ordered representative in $\mathcal{D}_{h, n}^{\rm fund}$.

\emph{Grid generation.}
Let $m$ be the number of uniformly spaced grid values per dimension on $[0, \kappa]$:
$\bigl\{0, \frac{1}{m-1}\kappa, \frac{2}{m-1}\kappa, \dots, \kappa\bigr\}$.
The number of grid points in the fundamental domain can be obtained via combinations with repetitions:
\begin{equation*}
  \begin{aligned}
    |\mathcal{G}|(n, h, m) &=
    \begin{cases}
      m \cdot \displaystyle\binom{m + n - 2}{n - 1}, & h = 1,\\[10pt]
      \displaystyle\frac{m(m+1)}{2} \cdot \binom{m + n - 3}{n - 2}, & h = 2.
    \end{cases}
  \end{aligned}
\end{equation*}
For fixed $m$, $|\mathcal{G}| = O(n^{m-1})$, a polynomial in $n$, as opposed to the full hypercube grid $|\mathcal{G}_{\rm full}| = m^n$ which grows exponentially.
As an illustration, for $n = 40$ and $m = 7$, $|\mathcal{G}| \approx 5.7 \times 10^7$ ($h=1$) or $\approx 2.0 \times 10^8$ ($h=2$), whereas $m^n \approx 7^{40} \approx 6.4 \times 10^{33}$.

\emph{Verification algorithm.}
The verification proceeds in three phases.
\begin{enumerate}[label=(\arabic*)]
  \item \textbf{Boundary check.}
  Evaluate $\lambda_{\vec{a}}$ at all distinct boundary configurations where each $v_t \in \{0, \kappa\}$.
  Due to block symmetry, only $(h+1)(n-h+1) = O(n)$ combinations are distinct.
  \item \textbf{Grid evaluation.}
  Generate all $|\mathcal{G}|$ grid points in $\mathcal{D}_{h, n}^{\rm fund}$ and evaluate $\lambda_{\vec{a}}$ at \emph{every} point.
  No random sampling is employed.
  \item \textbf{Multi-start local optimization.}
  Using each grid point as an initial iterate, run a bounded quasi-Newton solver (L-BFGS-B, $\texttt{ftol}=10^{-12}$, $\texttt{maxiter}=200$) to locate nearby local minima.
  Phase~2 already provides a lower bound from the grid itself; Phase~3 refines it by descending from every grid point.
\end{enumerate}


\emph{Numerical results.}
The algorithm is implemented with GPU parallelization using \texttt{JAX}.
Computations run at \texttt{float64} precision throughout.
Table~\ref{tab:benchmark} reports the optimized minimum values for a range of $(n, h, m)$ combinations.
We take fine-grid $m \geq 7$ for small number of qubits ($n \leq 40$) to exhaustively chart the fundamental domain at low dimension, confirming the non-negativity of minimum value at high resolution.
The optimized minimum is stable across tested values of $n$ and $m$ for each fixed $h$.
For $h=1$, the optimized minimum is approximately $-4.44 \times 10^{-16}$, equivalently $0$ within \texttt{fp64} numerical precision.
The result is consistent with definition of $\lambda_{\vec{a}}$ with $h=1$, for which the polynomial touches $0$ at some boundary points such as $(\kappa, 0, \cdots, 0)$.
For $h=2$, the optimized minimum is approximately $0.3232$.
By definition of $\lambda_{\vec{a}}$, such value can be attained at boundary configurations of the form $(\kappa, \kappa, 0, \cdots, 0)$.
We extend the verification to higher party numbers with more coarse grid ($n \geq 60, m \leq 6$), where the constancy of the minimum holds, supporting the conjecture that $\lambda_{\vec{a}}$ has no negative values in the domain for $h=1$ and $h=2$.

Together, the numerical evidence supports the optimal robustness bound in Eq.~\eqref{eq:upper-bound-params-sm} for all $n \geq 3$.
For $n \le 5$ this bound is already established analytically~\cite{Kaniewski_2016,Jha_Singh_Pan_2026}; for $n \geq 6$ the two remaining cases ($h = 1, 2$) are an open conjecture supported by the numerical evidence above.

\begin{table}[htbp]
\centering
\caption{Numerical results of the verification algorithm.}
\label{tab:benchmark}
\begin{tabular}{ccccc}
\toprule
$n$ & $h$ & $m$ & Number of Grid Points & Optimized Minimum \\
\midrule
20  & 1 & 9 & 19980675 & $-4.44\times10^{-16}$ \\
20  & 2 & 9 & 70302375 & $0.3232$ \\
40  & 1 & 7 & 57015420 & $-4.44\times10^{-16}$ \\
40  & 2 & 7 & 197653456 & $0.3232$ \\
60  & 1 & 6 & 45747072  & $-4.44\times10^{-16}$ \\
60  & 2 & 6 & 147605787 & $0.3232$ \\
80  & 1 & 5 & 9188110 & $-4.44\times10^{-16}$ \\
80  & 2 & 5 & 26235900 & $0.3232$ \\
100 & 1 & 5 & 22106375 & $-4.44\times10^{-16}$ \\
100 & 2 & 5 & 63743625 & $0.3232$ \\
\bottomrule
\end{tabular}
\end{table}

\section{Finite Sampling Complexity Analysis}
\label{sec:finite_sampling}

In this section, we analyze the finite-sampling complexity of the $n$-qubit MABK self-testing protocol.
Given a finite number $N$ of experimental rounds, we derive a confidence interval for the MABK violation $\beta$, combine it with the operator inequality in Proposition~\ref{prop:main} to obtain a certified extractability bound, and determine the number of rounds required to achieve a target extractability with prescribed confidence.
The analysis reveals that for near-maximal violation the required number of rounds scales as $N \propto 2^{n}$, consistent with the $2^{n}$ measurement settings of the $n$-qubit MABK inequality.

\subsection{Estimation Protocol}

\label{sm:fs:protocol}

Consider $N$ independent experimental rounds.
In each round $i \in \{1, \ldots, N\}$, a measurement setting $x_i \in \{0, 1\}^n$ is chosen uniformly at random, and the $n$ parties each perform the projective measurement corresponding to their respective observable $O_{x_{i,k}}^{(k)}$.
The outcomes $a_{i,1}, \ldots, a_{i,n} \in \{\pm 1\}$ are recorded and the product $z_i = \prod_{k=1}^n a_{i,k}$ is computed.

The $n$-qubit MABK Bell operator admits the expansion
\begin{equation*}
  W_n = \sum_{\vec{x} \in \{0, 1\}^n} c_{\vec{x}} \, \bigotimes_{k=1}^n O_{x_k}^{(k)},
\end{equation*}
where the coefficients are determined recursively as described in the main text and satisfy $|c_{\vec{x}}| \le 1$ for $n \ge 3$ (for odd $n$ exactly half of them vanish).
The maximum quantum violation is $\beta_{nQ} = 2^{(n+1)/2}$.

For a state $\rho$ and fixed measurement observables, define the \emph{correlator} for setting $\vec{x} \coloneqq (x_1 x_2 \cdots x_n)$ as
\begin{equation*}
  E(\vec{x}) = \operatorname{tr}\!\bigl[ \rho \, \bigotimes_{k=1}^n O_{x_k}^{(k)} \bigr].
\end{equation*}
The MABK violation value is then
\begin{equation*}
  \beta = \sum_{\vec{x} \in \{0, 1\}^n} c_{\vec{x}} \, E(\vec{x}).
\end{equation*}
Since each $O_{x_k}^{(k)}$ is a Hermitian operator with eigenvalues $\pm 1$, each correlator satisfies $|E(\vec{x})| \leq 1$.

We consider two complementary estimators for $\beta$.
The \emph{importance-weighted (IW) estimator} uses a fixed reweighting factor $2^n$ to correct for the uniform sampling probability:
\begin{equation*}
  \hat{\beta}_{\rm IW} = \frac{1}{N} \sum_{i=1}^N 2^n \, c_{x_i} \, z_i.
\end{equation*}
Since each setting is chosen with probability $1/2^n$, we have $\mathbb{E}[2^n c_{x_i} z_i] = \sum_{\vec{x}} c_{\vec{x}} E(\vec{x}) = \beta$, so $\hat{\beta}_{\rm IW}$ is unbiased.
Each term $2^n c_{x_i} z_i$ lies in $[-2^n, 2^n]$ (range $2^{n+1}$), reflecting the variance inflation inherent in the per-round reweighting.

The \emph{per-setting (PS) estimator} partitions the $N$ rounds by setting and computes the empirical mean for each:
\begin{equation*}
  \hat{E}(\vec{x}) = \begin{cases}
    \displaystyle \frac{1}{N_{\vec{x}}} \sum_{i: \vec{x}_i = \vec{x}} z_i, & N_{\vec{x}} > 0, \\[12pt]
    0, & N_{\vec{x}} = 0,
  \end{cases}
  \qquad
  \hat{\beta}_{\rm PS} = \sum_{\vec{x} \in \{0, 1\}^n} c_{\vec{x}} \, \hat{E}(\vec{x}),
\end{equation*}
where $N_{\vec{x}} = |\{i : \vec{x}_i = \vec{x}\}|$ is the number of rounds in which setting $\vec{x}$ was selected.
The random variables $\{N_{\vec{x}}\}$ are multinomially distributed with $N_{\vec{x}} \sim \mathrm{Binomial}(N, 1/2^n)$ and $\sum_{\vec{x}} N_{\vec{x}} = N$.
Conditioned on $N_{\vec{x}} > 0$, $\hat{E}(\vec{x})$ is the sample mean of $N_{\vec{x}}$ independent draws of $z_i$, each bounded in $[-1, 1]$ (range $2$).
In contrast to the IW estimator, each summand in the PS formulation remains in $[-1, 1]$ without amplification, yielding lower per-round variance.
This makes the PS estimator more suitable when the sample size $N$ is small.
However, the random sample sizes $\{N_{\vec{x}}\}$ are correlated through the multinomial constraint $\sum_{\vec{x}} N_{\vec{x}} = N$, which complicates the concentration analysis.
Specifically, after conditioning and union bounds, the resulting confidence interval scales as $2^{3n/2} / \sqrt{N}$, as detailed below.

Both estimators are unbiased for $\beta$ and process the same experimental data.
For the certification analysis we adopt the IW formulation, whose i.i.d.\ summands yield the tighter $2^n / \sqrt{N}$ confidence interval and a correspondingly tighter measurement budget.

\subsection{Hoeffding Concentration}

\label{sm:fs:hoeffding}

We begin with the classic concentration inequality for bounded random variables.

\begin{lemma}[Hoeffding's Inequality~\cite{Hoeffding_1963}]
  \label{lem:sm:fs:hoeffding}
  Let $Y_1, \ldots, Y_m$ be independent random variables such that $Y_i \in [a_i, b_i]$ almost surely.
  Define $\bar{Y} = \frac{1}{m} \sum_{i=1}^m Y_i$ and $\mu = \mathbb{E}[\bar{Y}]$.
  Then for any $\tau > 0$,
  \begin{equation*}
    \Pr\!\bigl( |\bar{Y} - \mu| \geq \tau \bigr) \leq 2 \exp\!\left( -\frac{2 m \tau^2}{\frac{1}{m} \sum_{i=1}^m (b_i - a_i)^2} \right).
  \end{equation*}
  For the special case where all $Y_i \in [-1, 1]$, the bound simplifies to
  \begin{equation*}
    \Pr\!\bigl( |\bar{Y} - \mu| \geq \tau \bigr) \leq 2 \exp\!\left( -\frac{m \tau^2}{2} \right).
  \end{equation*}
\end{lemma}

We now apply this lemma to establish concentration of the (IW) estimator.
Since each per-round contribution $2^n c_{x_i} z_i$ lies in $[-2^n, 2^n]$, the range of each term is $R = 2^{n+1}$.
The sum of ranges squared is $N \cdot (2^{n+1})^2 = N \cdot 4^{n+1}$.
By Lemma~\ref{lem:sm:fs:hoeffding} with $a_i = -2^n$ and $b_i = 2^n$,
\begin{equation}
  \Pr\!\bigl( |\hat{\beta}_{\rm IW} - \beta| \geq \tau \bigr)
  \leq 2 \exp\!\left( -\frac{2 N \tau^2}{4^{n+1}} \right)
  = 2 \exp\!\left( -\frac{N \tau^2}{2 \cdot 4^{n}} \right).
  \label{sm:fs:iw_hoeffding}
\end{equation}

For the purpose of certification, we require a one-sided confidence bound $\beta \geq \hat{\beta}_{\rm IW} - \tau$ with confidence $1 - \delta$.
Since $|\hat{\beta}_{\rm IW} - \beta| \geq \tau$ implies either $\beta \leq \hat{\beta}_{\rm IW} - \tau$ or $\beta \geq \hat{\beta}_{\rm IW} + \tau$, the two-sided bound in~\eqref{sm:fs:iw_hoeffding} directly yields
\begin{equation*}
  \Pr\!\bigl( \beta \leq \hat{\beta}_{\rm IW} - \tau \bigr)
  \leq \Pr\!\bigl( |\hat{\beta}_{\rm IW} - \beta| \geq \tau \bigr)
  \leq 2 \exp\!\left( -\frac{N \tau^2}{2 \cdot 4^{n}} \right).
\end{equation*}
Setting this tail probability equal to $\delta$ and solving for $\tau$ yields the following concentration result.

\begin{observation}
  \label{obv:sm:fs:iw_concentration}
  For the IW estimator $\hat{\beta}_{\rm IW}$, with confidence $1 - \delta$,
  \begin{equation}
    \beta \geq \hat{\beta}_{\rm IW} - 2^n \sqrt{\frac{2 \ln(2 / \delta)}{N}}.
    \label{sm:fs:iw_bound}
  \end{equation}
\end{observation}

For the PS estimator, we need to handle the random sample sizes $\{N_{\vec{x}}\}$.
We proceed by first conditioning on the multinomial draw, applying Hoeffding per setting, and then using a Chernoff bound to control the event that some $N_{\vec{x}}$ is too small.

\begin{lemma}[Multiplicative Chernoff Bound \cite{mitzenmacher2017probability}]
  \label{lem:sm:fs:chernoff}
  Let $X \sim \mathrm{Binomial}(m, p)$ with $\mu = \mathbb{E}[X] = mp$.
  For any $\gamma \in (0, 1)$,
  \begin{equation*}
    \Pr\!\bigl( X \leq (1 - \gamma) \mu \bigr) \leq \exp\!\left( -\frac{\gamma^2 \mu}{2} \right).
  \end{equation*}
\end{lemma}

\begin{observation}
  \label{obv:sm:fs:nx_lower}
  Let $N_{\min} = \min_{\vec{x}} N_{\vec{x}}$.
  With probability at least $1 - \delta_1$,
  \begin{equation*}
    N_{\min} \geq \frac{N}{2^{n+1}},
  \end{equation*}
  provided that $N \geq 8 \cdot 2^n \ln(2^n / \delta_1)$.
\end{observation}

\noindent \emph{Proof.}
For each $\vec{x}$, $N_{\vec{x}} \sim \mathrm{Binomial}(N, 1/2^n)$ with $\mu = N / 2^n$.
Taking $\gamma = 1/2$ in Lemma~\ref{lem:sm:fs:chernoff},
\begin{equation*}
  \Pr\!\left(N_{\vec{x}} \leq \frac{N}{2^{n+1}}\right) \leq \exp\!\left(-\frac{N}{8 \cdot 2^n}\right).
\end{equation*}
By the union bound over all $2^n$ settings,
\begin{equation*}
  \Pr\!\left(\bigcup_{\vec{x}} \left\{ N_{\vec{x}} \leq \frac{N}{2^{n+1}} \right\}\right)
  \leq 2^n \exp\!\left(-\frac{N}{8 \cdot 2^n}\right) \leq \delta_1,
\end{equation*}
where the last inequality follows from the condition $N \geq 8 \cdot 2^n \ln(2^n / \delta_1)$.
\hfill $\square$

We now combine the Chernoff guarantee on $N_{\min}$ with Hoeffding's inequality applied conditionally to each setting.

\begin{observation}
  \label{obv:sm:fs:ps_concentration}
  For the PS estimator $\hat{\beta}_{\rm PS}$, with confidence $1 - \delta$,
  \begin{equation}
    |\hat{\beta}_{\rm PS} - \beta| \leq 2^n \sqrt{\frac{2^{n+2} \ln(4 \cdot 2^n / \delta)}{N}},
    \label{sm:fs:ps_bound}
  \end{equation}
  provided that $N \geq 8 \cdot 2^n \ln(2 \cdot 2^n / \delta)$.
\end{observation}

\noindent \emph{Proof.}
Condition on the multinomial draw $\{N_{\vec{x}}\}$.
For each setting $\vec{x}$ with $N_{\vec{x}} > 0$, applying Lemma~\ref{lem:sm:fs:hoeffding} to the $N_{\vec{x}}$ i.i.d.\ observations $z_i \in [-1, 1]$ yields
\begin{equation}
  \Pr\!\bigl( |\hat{E}(\vec{x}) - E(\vec{x})| \geq \tau_{\vec{x}} \mid N_{\vec{x}} \bigr) \leq 2 \exp\!\left( -\frac{N_{\vec{x}} \tau_{\vec{x}}^2}{2} \right).
  \label{sm:fs:cond_hoeff}
\end{equation}
If $N_{\vec{x}} = 0$, then $\hat{E}(\vec{x}) = 0$ and $|\hat{E}(\vec{x}) - E(\vec{x})| \leq 1$, which can be absorbed into the subsequent analysis with a trivial bound.

Apply Observation~\ref{obv:sm:fs:nx_lower} with $\delta_1 = \delta/2$; this requires $N \geq 8 \cdot 2^n \ln(2 \cdot 2^n / \delta)$.
Condition on the event $\mathcal{E}_{\rm C} = \{N_{\min} \geq N / 2^{n+1}\}$, which holds with probability at least $1 - \delta/2$.

Choose a uniform error threshold $\tau$ for all settings and set
\begin{equation*}
  \tau = \sqrt{ \frac{2 \ln(4 \cdot 2^n / \delta)}{N_{\min}} }.
\end{equation*}
For each $\vec{x}$, Eq.~\eqref{sm:fs:cond_hoeff} gives
\begin{equation*}
  \Pr\!\bigl( |\hat{E}(\vec{x}) - E(\vec{x})| \geq \tau \mid N_{\vec{x}}, \mathcal{E}_{\rm C} \bigr)
  \leq 2 \exp\!\left( -\frac{N_{\vec{x}} \tau^2}{2} \right)
  \leq 2 \exp\!\left( -\frac{N_{\min} \tau^2}{2} \right)
  = \frac{\delta}{2 \cdot 2^n}.
\end{equation*}
By the union bound over all $2^n$ settings,
\begin{equation*}
  \Pr\!\left( \bigcup_{\vec{x}} \bigl\{ |\hat{E}({\vec{x}}) - E({\vec{x}})| \geq \tau \bigr\} \;\Big|\; \mathcal{E}_{\rm C} \right)
  \leq 2^n \cdot \frac{\delta}{2 \cdot 2^n} = \frac{\delta}{2}.
\end{equation*}
Therefore, unconditionally,
\begin{equation*}
  \Pr\!\left( \bigcup_{\vec{x}} \bigl\{ |\hat{E}(\vec{x}) - E(\vec{x})| \geq \tau \bigr\} \right)
  \leq 1 - \Pr(\mathcal{E}_{\rm C}) + \frac{\delta}{2} \leq \delta.
\end{equation*}
With probability at least $1 - \delta$,
\begin{equation*}
  |\hat{\beta}_{\rm PS} - \beta|
  = \Bigl| \sum_{\vec{x}} c_{\vec{x}} \bigl( \hat{E}(\vec{x}) - E(\vec{x}) \bigr) \Bigr|
  \leq \sum_{\vec{x}} |\hat{E}(\vec{x}) - E(\vec{x})|
  \leq 2^n \tau
  = 2^n \sqrt{ \frac{2 \ln(4 \cdot 2^n / \delta)}{N_{\min}} }.
\end{equation*}
Substituting $N_{\min} \geq N / 2^{n+1}$ yields Eq.~\eqref{sm:fs:ps_bound}.
\hfill $\square$

Comparing the IW bound~\eqref{sm:fs:iw_bound} with the PS bound~\eqref{sm:fs:ps_bound}, we note that the IW bound scales as $2^n / \sqrt{N}$ (up to logarithmic factors), whereas the PS bound scales as $2^{3n/2} / \sqrt{N}$, the additional factor $\sqrt{2^{n+2}} = 2^{n/2 + 1}$ reflecting that the effective sample size per setting is $N / 2^n$ rather than $N$.
For small to moderate $n$, the PS estimator may nevertheless be preferable in practice because its per-round contribution remains bounded in $[-1, 1]$ (rather than $[-2^n, 2^n]$), which can yield lower variance in finite samples.
For the subsequent certification analysis we adopt the IW bound~\eqref{sm:fs:iw_bound}, whose $2^n / \sqrt{N}$ scaling yields a tighter and simpler measurement budget formula.

\subsection{Confidence Interval for the MABK Violation}

\label{sm:fs:interval}

From Observation~\ref{obv:sm:fs:iw_concentration}, with confidence $1 - \delta$,
\begin{equation}
  \beta \geq \hat{\beta}_{\rm IW} - \Delta(N, \delta),
  \label{sm:fs:beta_lower}
\end{equation}
where the width of the confidence interval is
\begin{equation}
  \Delta(N, \delta) = 2^n \sqrt{\frac{2 \ln(2 / \delta)}{N}}.
  \label{sm:fs:delta_def}
\end{equation}

For what follows, we drop the subscript and write $\hat{\beta} \equiv \hat{\beta}_{\rm IW}$.



Proposition~\ref{prop:main} establishes the operator inequality~\eqref{eq:operator-ineq-main-sm}
with
\begin{equation}
  s_n = \frac{\sqrt{2} + 1}{2^{n/2}}, \qquad \mu = -(1 + \sqrt{2}),
  \label{sm:fs:params}
\end{equation}
holding for all measurement angles $\{\theta_k\} \in [0, \pi/2]^n$ and for all $n \geq 3$.
Consequently, for any $n$-partite state $\rho$ and any choice of local measurement observables, the extractability of the target state $\Phi_n$ (which is locally equivalent to the $n$-qubit GHZ state) satisfies
\begin{equation}
  \Xi(\rho \to \Phi_n) \geq s_n \, \beta + \mu,
  \label{sm:fs:extractability}
\end{equation}
where $\beta = \operatorname{tr}[\rho W_n]$ is the MABK violation produced by $\rho$.

Combining the concentration inequality~\eqref{sm:fs:beta_lower} with the extractability bound~\eqref{sm:fs:extractability}, we obtain the finite-sample certified extractability.

\begin{proposition}[Finite-Sample Certified Extractability]
  \label{prop:sm:fs:certified}
  Let $\hat{\beta}$ be the MABK violation estimated from $N$ experimental rounds with the protocol of Section~\ref{sm:fs:protocol}.
  With confidence at least $1 - \delta$,
  \begin{equation*}
    \Xi(\rho \to \Phi_n) \geq s_n \bigl( \hat{\beta} - \Delta(N, \delta) \bigr) + \mu,
  \end{equation*}
  where $\Delta(N, \delta)$ is defined in Eq.~\eqref{sm:fs:delta_def}, $s_n$ and $\mu$ are given in Eq.~\eqref{sm:fs:params}, and $\Phi_n$ is the $n$-qubit state locally equivalent to $|\mathrm{GHZ}_n\rangle$.
\end{proposition}


\noindent Thus, from a finite number of experimental rounds, we can certify a lower bound on the extractability of the $n$-qubit GHZ state with a rigorously quantified confidence level.

\subsection{Measurement Budget}

\label{sm:fs:budget}

We now determine the number of rounds $N$ required to certify extractability at least $\Xi_{\rm target}$ with confidence $1 - \delta$.

From Proposition~\ref{prop:sm:fs:certified}, the requirement $\Xi \geq \Xi_{\rm target}$ is satisfied whenever
\begin{equation*}
  s_n \bigl( \hat{\beta} - \Delta(N, \delta) \bigr) + \mu \geq \Xi_{\rm target}.
\end{equation*}
Rearranging,
\begin{equation}
  \Delta(N, \delta) \leq \hat{\beta} + \frac{\mu}{s_n} - \frac{\Xi_{\rm target}}{s_n}.
  \label{sm:fs:delta_requirement}
\end{equation}
Substituting the explicit form of $\Delta(N, \delta)$ from Eq.~\eqref{sm:fs:delta_def} and solving for $N$ yields the following result.

\begin{proposition}[Measurement Budget]
  \label{prop:sm:fs:budget}
  To certify extractability $\Xi(\rho \to \Phi_n) \geq \Xi_{\rm target}$ with confidence $1 - \delta$ given an estimated MABK violation $\hat{\beta}$, it suffices to perform
  \begin{equation}
    N \geq \frac{2 \cdot 2^{2n} \cdot \ln(2 / \delta)}{\bigl( \hat{\beta} + \mu / s_n - \Xi_{\rm target} / s_n \bigr)^2}.
    \label{sm:fs:budget_formula}
  \end{equation}
\end{proposition}

\begin{proof}
  From Eq.~\eqref{sm:fs:delta_def},
  \begin{equation*}
    \Delta(N, \delta) = 2^n \sqrt{\frac{2 \ln(2 / \delta)}{N}}.
  \end{equation*}
  Inserting this into the condition $\Delta(N, \delta) \leq \hat{\beta} + \mu/s_n - \Xi_{\rm target}/s_n$ and squaring both sides,
  \begin{equation*}
    \frac{2 \cdot 2^{2n} \cdot \ln(2 / \delta)}{N}
    \leq \left( \hat{\beta} + \frac{\mu}{s_n} - \frac{\Xi_{\rm target}}{s_n} \right)^{\!2}.
  \end{equation*}
  Solving for $N$ gives Eq.~\eqref{sm:fs:budget_formula}.
  For the bound to be meaningful, the denominator must be positive, which requires $\hat{\beta} \geq \Xi_{\rm target} / s_n - \mu / s_n$.
  Using the explicit parameter values from~\eqref{sm:fs:params},
  \begin{equation*}
    \frac{\mu}{s_n} = -\frac{(1+\sqrt{2}) \cdot 2^{n/2}}{\sqrt{2}+1} = -2^{n/2}.
  \end{equation*}
  Hence the condition reads $\hat{\beta} \geq (\Xi_{\rm target} + 2^{n/2}) / s_n$, which is always satisfied for sufficiently large $\hat{\beta}$ approaching the maximal violation $\beta_{nQ} = 2^{(n+1)/2}$.
\end{proof}

The denominator in Eq.~\eqref{sm:fs:budget_formula} has a clear physical interpretation: the term $\hat{\beta} + \mu/s_n$ represents the slack between the estimated violation and the intercept of the extractability bound.
Specifically, $\hat{\beta} + \mu/s_n = \hat{\beta} - 2^{n/2}$ must exceed $\Xi_{\rm target}/s_n$ for certification to be possible.
At the maximal quantum violation $\beta_{nQ} = \sqrt{2} \cdot 2^{n/2}$, the slack is $(\sqrt{2} - 1) \, 2^{n/2}$, which grows exponentially with $n$.
Since $s_n \propto 2^{-n/2}$, the required precision $\Xi_{\rm target}/s_n$ also grows with $n$; the net effect is that the denominator scales proportional to $2^{n/2}$ times constants determined by the target and the estimated violation.

\subsection{Discussion}

\label{sm:fs:discussion}

Equation~\eqref{sm:fs:budget_formula} reveals the fundamental scaling of the measurement budget with the number of parties.
For the near-maximal violation regime, the denominator acquires a compensating $2^{n}$ factor from $(\mu/s_n)^2$, canceling one factor of $2^{2n}$ in the numerator and yielding $N \propto 2^{n}$.
Explicitly, taking $\hat{\beta} = \beta_{nQ}(1 - \eta)$, where $\eta \coloneqq \varepsilon/\beta_{nQ}$ is the relative violation deficiency, with $\beta_{nQ} = \sqrt{2} \cdot 2^{n/2}$, $\mu/s_n = -2^{n/2}$, and $s_n = (\sqrt{2}+1)/2^{n/2}$, the denominator becomes $2^{n/2}\bigl[\sqrt{2}(1-\eta) - 1 - \Xi_{\rm target}/(\sqrt{2}+1)\bigr]$, whence
\begin{equation*}
  N \geq \frac{2 \cdot 2^{n} \cdot \ln(2/\delta)}{\bigl[\sqrt{2}(1-\eta) - 1 - \Xi_{\rm target}/(\sqrt{2}+1)\bigr]^2},
\end{equation*}
that is, $N \propto 2^{n}$ for fixed $\eta$ and $\Xi_{\rm target}$.

The $2^{2n}$ in the numerator originates from the $2^n$ range of each IW summand, squared via Hoeffding; the compensating $2^n$ in the denominator arises because $|\mu|/s_n = 2^{n/2}$ (squared thereafter), reflecting the growing slack between violation and certification threshold.
If the violation does not scale with $n$, the term $\mu/s_n \propto -2^{n/2}$ dominates and makes the RHS of~\eqref{sm:fs:delta_requirement} negative for large $n$; certification of a fixed $\Xi_{\rm target} > 0$ then requires $N \to \infty$, so the protocol is only practically meaningful in the near-maximal regime.

The budget $N \propto 2^{n}$ limits applicability to $n = O(\log N)$ parties, consistent with the $2^{n}$ measurement settings of the MABK family.
If an a priori bound on per-round variance is available, Bernstein's inequality~\cite{Bernstein_1924} can tighten the concentration analysis; a sequential testing procedure (e.g.\ the law of the iterated logarithm~\cite{Balsubramani_2014}) could further reduce the required rounds by adaptive stopping.
We leave these refinements for future work.

\section{Certified Randomness from the Extractability Bound}
\label{sec:dirg}

In a DIRG protocol based on the $n$-qubit MABK inequality, the certified randomness rate, that is, the conditional entropy $H(R|E)$ of the measurement outcomes given Eve's side information, can deviate from its ideal value due to two independent sources:
\begin{enumerate}[label=(\roman*)]
  \item \textbf{State imperfection}: the shared state $\rho_Q$ is not exactly the $n$-qubit GHZ state $|\Phi_n\rangle$;
  \item \textbf{Measurement imperfection}: the local observables of the $n$ parties are not exactly the optimal Pauli measurements (angles $\theta_k \neq \pi/4$).
\end{enumerate}
At maximal MABK violation, self-testing guarantees that both are ideal and the randomness rate attains its maximal value~\cite{Wooltorton_Brown_Colbeck_2025}.
At non-maximal violation, the imperfect state and measurements could both contribute to the deficit $H_{\rm ideal} - H(R|E)$.

In this section, we formulate the sources of randomness deficit and quantify their contributions.
Using the extractability bound of Proposition~\ref{prop:main}, which holds uniformly for \emph{all} measurement angles, we first derive a certified lower bound on $H(R|E)$ that accounts for the fidelity gap between the physical state and the ideal GHZ state (state-imperfection contribution).
The measurement-imperfection contribution, arising from the difference between the physical and ideal measurement channels, together with the measurement-translation residual arising from the non-commutation of the physical measurement with the extraction isometry, is then bounded in Section~\ref{sm:dirg:completion} via a commutation-based analysis, yielding a combined, fully device-independent bound.

\subsection{Background on DIRG and MABK Randomness}
\label{sm:dirg:background}

We consider a device-independent scenario with $n$ spatially separated parties, each holding an uncharacterized quantum device.
In each experimental round, the $n$ parties receive binary inputs $x_k \in \{0, 1\}$ and produce binary outputs $a_k \in \{\pm 1\}$.
The correlations are described by an unknown $n$-partite quantum state $\rho_Q$ measured with local binary observables $\{O^{(k)}_{x_k}\}$.
We assume the existence of an adversary, Eve, who may be entangled with the devices; the global state $|\Psi\rangle_{QE}$ on the device registers $Q = Q_1 \cdots Q_n$ and Eve's register $E$ is thus taken to be pure, with $\rho_Q = \operatorname{tr}_E |\Psi\rangle\langle\Psi|_{QE}$.

In a DIRG protocol, the $n$ parties fix a common input setting, conventionally $\vec{x} = \mathbf{0} \coloneqq (0,0,\dots,0)$, and collect the raw output string $\vec{a} = (a_1, \dots, a_n)$.
The device-independent randomness rate certified by a Bell inequality $W_n$ (expressed in terms of the measurement observables) that achieves violation $\beta = \langle W_n \rangle$ is defined as~\cite{Wooltorton_Brown_Colbeck_2025}
\begin{equation*}
  R(\beta) \coloneqq \inf_{\substack{\text{quantum strategies with} \\ \langle W_n \rangle \geq \beta}} H(R \mid X = \mathbf{0}, E)_{\rho},
\end{equation*}
where $\rho_{RE|X=\mathbf{0}}$ is the classical-quantum (cq) state after measuring input $\mathbf{0}$:
\begin{equation*}
  \rho_{RE|\mathbf{0}} = \sum_{\vec{a} \in \{\pm 1\}^n} |\vec{a}\rangle\langle \vec{a}|_R \otimes \operatorname{tr}_Q\!\bigl[ (P_{\vec{a}}^{(\mathbf{0})} \otimes I_E) \, |\Psi\rangle\langle\Psi|_{QE} \bigr],
\end{equation*}
with $P_{\vec{a}}^{(\mathbf{0})} = \bigotimes_{k=1}^n P_{a_k}^{(k,0)}$ the projector onto outcome $a_k$ for the $0$-th measurement of party $k$.
The register $R$ holds the classical output string $\vec{a}$, and $E$ is Eve's quantum side information.
The conditional entropy $H(R|E) = H(\rho_{RE|\mathbf{0}}) - H(\rho_E)$ quantifies the extractable randomness in the presence of an adversary.

\begin{proposition}[Wooltorton~et~al.\ \cite{Wooltorton_Brown_Colbeck_2025}]
  \label{sm:dirg:wooltorton}
  For the $n$-qubit MABK inequality, the maximal quantum violation $\beta_{nQ} = 2^{(n+1)/2}$ certifies, on input $\mathbf{0}$,
  \begin{equation}
    H_{\rm ideal}(n) = \begin{cases}
      n, & n \text{ odd and } n \equiv 3 \pmod 4,\\[4pt]
      n - 1, & n \text{ odd and } n \equiv 1 \pmod 4,\\[4pt]
      n + \frac{1}{2} - \frac{\log_2(1+\sqrt{2})}{\sqrt{2}} \approx n - 0.4, & n \text{ even},
    \end{cases}
    \label{sm:dirg:hideal_def}
  \end{equation}
  where $H_{\rm ideal}(n) \coloneqq H(R|E)$ evaluated at maximal violation.
  For odd $n \equiv 1 \pmod 4$ the uniformly random input is $\mathbf{1}$ rather than $\mathbf{0}$: the maximal violation then certifies the full $n$ bits, whereas on input $\mathbf{0}$ the outcomes are uniform over only half of the strings, giving $n-1$ bits.
\end{proposition}

The state and measurements that achieve $\beta_{nQ}$ are known: the state is locally equivalent to the $n$-qubit GHZ state, and the optimal measurement angles are $\theta_k = \pi/4$ for all $k$, making the two observables of each party maximally complementary.
At this optimal point the measurement outcomes are decoupled from any external system; the resulting conditional entropy $H(R|E)$ on input $\mathbf{0}$ equals $n$ for $n \equiv 3 \pmod 4$, $n - 1$ for $n \equiv 1 \pmod 4$, and $n - 0.4$ for even $n$, as stated in Proposition~\ref{sm:dirg:wooltorton}.
The parity dependence originates from the structure of the MABK correlators: for even $n$ every $n$-party correlator appears with non-zero coefficient, preventing full uniform distribution on a single input setting; for odd $n$ exactly half of the correlators appear, and the input setting that yields the uniform distribution is determined by $n \bmod 4$.
In the DIRG protocol below we take input $\mathbf{0}$ as the randomness input; for $n \equiv 1 \pmod 4$ the same analysis applies verbatim to input $\mathbf{1}$ (the state- and measurement-imperfection bounds of Sections~\ref{sm:dirg:post_measurement} and \ref{sm:dirg:completion} do not depend on which of the two observables of each party is measured), for which the ideal entropy is $H_{\rm ideal}(n) = n$.

Our goal is to bound the randomness rate $R(\beta_{nQ} - \varepsilon)$, for both parities, when the observed violation falls short of the maximum by $\varepsilon > 0$, using the robust extractability bound of Proposition~\ref{prop:main}.

\subsection{From Extractability to Trace Distance}
\label{sm:dirg:trace_distance}

Proposition~\ref{prop:main} provides a lower bound on the extractability $\Xi(\rho_Q \to \Phi_n)$ in terms of the MABK violation:
\begin{equation*}
  \Xi(\rho_Q \to \Phi_n) \geq s_n \beta + \mu,
  \qquad
  s_n = \frac{\sqrt{2} + 1}{2^{n/2}}, \quad \mu = -(1+\sqrt{2}).
\end{equation*}
The extractability is defined with respect to a \emph{specific} local extraction channel $\Lambda = \bigotimes_{k=1}^n \Lambda_k$, constructed from the measurement observables via Jordan rotation and dephasing (see the main text).
For a violation $\beta = \beta_{nQ} - \varepsilon$, we obtain
\begin{equation}
  \Xi(\rho_Q \to \Phi_n) \geq s_n(\beta_{nQ} - \varepsilon) + \mu
  = (s_n\beta_{nQ} + \mu) - s_n\varepsilon
  = 1 - s_n \varepsilon,
  \label{sm:dirg:extractability_eps}
\end{equation}
where we used $s_n\beta_{nQ} + \mu = (\sqrt{2}+1)\sqrt{2} - (1+\sqrt{2}) = 1$.

Proposition~\ref{prop:main} is established via the operator inequality
\begin{equation*}
  K\bigl(\{\theta_k\}_{k=1}^n\bigr) \geq s_n\, W_n\bigl(\{\theta_k\}_{k=1}^n\bigr) + \mu I,
  \qquad
  K\bigl(\{\theta_k\}_{k=1}^n\bigr) \coloneqq \Lambda^\dagger\bigl(\{\theta_k\}_{k=1}^n\bigr)\bigl(|\Phi_n\rangle\langle\Phi_n|\bigr),
\end{equation*}
which holds for the \emph{specific} local extraction channel $\Lambda = \bigotimes_{k=1}^n \Lambda_k$ constructed from the measurement observables via Jordan rotation and dephasing (see the main text).
By the duality of $\Lambda$, $F\bigl(\Lambda(\rho_Q), \Phi_n\bigr) = \operatorname{tr}\!\bigl[\rho_Q\, K\bigl(\{\theta_k\}_{k=1}^n\bigr)\bigr]$ (with $\Phi_n$ pure); tracing the operator inequality with $\rho_Q$ and substituting $\beta = \beta_{nQ} - \varepsilon$ as in Eq.~\eqref{sm:dirg:extractability_eps} yields
\begin{equation}
  F\!\bigl( \Lambda(\rho_Q), \, \Phi_n \bigr) \geq 1 - s_n \varepsilon,
  \label{sm:dirg:fidelity_bound}
\end{equation}
where $F(\rho,\sigma) = \|\sqrt{\rho}\sqrt{\sigma}\|_1^2$ is the Uhlmann fidelity.

We now lift this fidelity bound on the reduced state to a trace-distance bound on a global purification.
Let $V_k: \mathcal{H}_{Q_k} \to \mathcal{H}_{A_k} \otimes \mathcal{H}_{J_k}$ be the Stinespring dilation of the local channel $\Lambda_k$, so that $\Lambda_k(\cdot) = \operatorname{tr}_{J_k}[V_k \,(\cdot)\, V_k^\dagger]$.
The combined isometry $V = \bigotimes_{k=1}^n V_k$ dilates $\Lambda$:
\begin{equation}
  \Lambda(\rho_Q) = \operatorname{tr}_{J}\!\bigl[ V \rho_Q V^\dagger \bigr],
  \label{sm:dirg:stinespring}
\end{equation}
with $J = J_1 \cdots J_n$ the joint junk system.

Define the dilated global extracted state
\begin{equation*}
  |\Psi'\rangle_{AJE} \coloneqq (V \otimes I_E) |\Psi\rangle_{QE},
\end{equation*}
where registers $A = A_1 \cdots A_n$ hold the logical qubits (the output of the isometry) and $J$ the junk.
Tracing out $J$, we recover
\begin{equation*}
  \operatorname{tr}_J\!\bigl[ |\Psi'\rangle\langle\Psi'|_{AJE} \bigr]
  = (\Lambda \otimes I_E)\bigl(|\Psi\rangle\langle\Psi|_{QE}\bigr)
  = (\Lambda \otimes I_E)(\rho_{QE}).
\end{equation*}

\begin{lemma}
  \label{lem:sm:dirg:purification_trace_distance}
  There exists a pure state $|\zeta\rangle_{JE}$ on the junk and Eve registers such that
  \begin{equation}
    \bigl\|\, |\Psi'\rangle\langle\Psi'|_{AJE}
    - |\Phi_n\rangle\langle\Phi_n|_A \otimes |\zeta\rangle\langle\zeta|_{JE} \,\bigr\|_1
    \leq 2\sqrt{s_n \varepsilon}.
    \label{sm:dirg:trace_distance_bound}
  \end{equation}
\end{lemma}

\begin{proof}
  From Eq.~\eqref{sm:dirg:fidelity_bound} and Eq.~\eqref{sm:dirg:stinespring}, the fidelity between the reduced states satisfies
  \begin{equation*}
    F\!\bigl( \operatorname{tr}_{JE}[|\Psi'\rangle\langle\Psi'|], \, \Phi_n \bigr)
    = F(\Lambda(\rho_Q), \Phi_n) \geq 1 - s_n\varepsilon.
  \end{equation*}
  By Uhlmann's theorem~\cite{Uhlmann_1976}, the fidelity between two mixed states equals the maximum squared overlap between their purifications.
  Since $|\Psi'\rangle_{AJE}$ is a purification of $\operatorname{tr}_{JE}[|\Psi'\rangle\langle\Psi'|]$ and any state $|\Phi_n\rangle_A \otimes |\zeta\rangle_{JE}$ is a purification of $\Phi_n$ (with arbitrary $|\zeta\rangle$), there exists a purification $|\zeta\rangle_{JE}$ achieving the fidelity:
  \begin{equation*}
    \bigl| \langle \Phi_n |_A \langle \zeta |_{JE} \, |\Psi'\rangle_{AJE} \bigr|^2
    = F(\Lambda(\rho_Q), \Phi_n) \geq 1 - s_n\varepsilon.
  \end{equation*}
  The Fuchs--van~de~Graaf inequality~\cite{Fuchs_van_de_Graaf_1999} relates the trace distance between two pure states to their fidelity:
  \begin{equation*}
    \frac{1}{2} \bigl\| |\psi\rangle\langle\psi| - |\phi\rangle\langle\phi| \bigr\|_1
    = \sqrt{1 - |\langle \psi | \phi \rangle|^2}.
  \end{equation*}
  Applying this to $|\psi\rangle = |\Psi'\rangle_{AJE}$ and $|\phi\rangle = |\Phi_n\rangle_A \otimes |\zeta\rangle_{JE}$, and using $1 - |\langle\phi|\psi\rangle|^2 \leq s_n\varepsilon$, we obtain~\eqref{sm:dirg:trace_distance_bound}.
\end{proof}

The isometry $V = \bigotimes_k V_k$ is a product of local isometries, reflecting the fact that the extraction channel $\Lambda$ respects the tensor-product structure of the $n$ parties.
This locality essentially guarantees that the logical registers $A_k$ remain associated with the respective parties, so that measurements on the physical systems translate (up to the isometry, and up to the measurement-translation residual quantified in Section~\ref{sm:dirg:completion}) to measurements on the logical qubits.

\subsection{Post-Measurement State Analysis}
\label{sm:dirg:post_measurement}

We now analyze the cq-state obtained after measuring input $\mathbf{0}$.
In the DIRG protocol the physical measurement is performed on the physical registers $Q_k$; under the local isometry $V_k$ the physical projector $P_{a_k}^{(k,\mathrm{phys})}$ is mapped to $\widetilde{A}_{a_k}^{(k)} \coloneqq V_k P_{a_k}^{(k,\mathrm{phys})} V_k^\dagger$, which acts on $A_k \otimes J_k$.
As shown in Section~\ref{sm:dirg:completion}, $\widetilde{A}_{a_k}^{(k)}$ does \emph{not} coincide with the logical projector $P_{a_k}^{(k,0)} \otimes I_{J_k}$ unless $\theta_k = \pi/4$; the resulting \emph{measurement-translation residual} is quantified in Lemma~\ref{lem:sm:dirg:translation} below.
Since the junk registers are held by the honest parties and are irrelevant for the certified randomness, we discard them, so that all post-measurement cq-states below live on the registers $R \otimes E$.
The post-measurement analysis then bounds the physical cq-state by the ideal one (where $|\Phi_n\rangle$ is measured with the optimal Pauli observables), up to the state- and measurement-imperfection contributions quantified below.

Let $\mathcal{M}$ denote the CPTP map that effects the $n$-partite measurement of input $\mathbf{0}$ and records the outcomes in a classical register $R$:
\begin{equation}
  \mathcal{M}(\rho_{AJE}) = \sum_{\vec{a} \in \{\pm 1\}^n} |\vec{a}\rangle\langle \vec{a}|_R \otimes \operatorname{tr}_A\!\bigl[ (P_{\vec{a}}^{(\mathbf{0})} \otimes I_{JE}) \, \rho_{AJE} \bigr],
  \label{sm:dirg:measurement_map}
\end{equation}
where $P_{\vec{a}}^{(\mathbf{0})}$ are the projectors for the (transformed) measurement operators acting on the logical qubits $A$.
The physical post-measurement cq-state is
\begin{equation}
  \rho^{\rm phys}_{RE|\mathbf{0}} \coloneqq \sum_{\vec{a} \in \{\pm 1\}^n} |\vec{a}\rangle\langle \vec{a}|_R \otimes \operatorname{tr}_Q\!\bigl[ (P_{\vec{a}}^{(\rm phys)} \otimes I_E) \, |\Psi\rangle\langle\Psi|_{QE} \bigr],
  \label{sm:dirg:phys_cq}
\end{equation}
with $P_{\vec{a}}^{(\rm phys)} = \bigotimes_{k=1}^n P_{a_k}^{(k,\rm phys)}$ the physical projectors, and the \emph{extracted-picture} cq-state (the logical measurement on the extracted state, with the junk discarded) is
\begin{equation}
  \rho'_{RE|\mathbf{0}} \coloneqq \operatorname{tr}_J\!\bigl[ \mathcal{M}\!\bigl( |\Psi'\rangle\langle\Psi'|_{AJE} \bigr) \bigr],
  \label{sm:dirg:actual_cq}
\end{equation}
The ideal post-measurement cq-state (achieved at $\beta = \beta_{nQ}$) is defined analogously,
\begin{equation*}
  \rho^{\rm ideal}_{RE|\mathbf{0}} \coloneqq \operatorname{tr}_J\!\bigl[ \mathcal{M}_{\rm ideal}\!\bigl( |\Phi_n\rangle\langle\Phi_n|_A \otimes |\zeta\rangle\langle\zeta|_{JE} \bigr) \bigr],
\end{equation*}

At maximal MABK violation, self-testing forces the measurement angles to $\theta_k = \pi/4$ for all $k$, irrespective of the parity of $n$.
At this optimal point, the measurement outcomes are decoupled from the junk register and from Eve; on the uniformly random input of Proposition~\ref{sm:dirg:wooltorton} (input $\mathbf{0}$ for $n \equiv 3 \pmod 4$ and input $\mathbf{1}$ for $n \equiv 1 \pmod 4$) the outcomes are uniformly distributed:
\begin{equation}
  \rho^{\rm ideal}_{RE|\mathbf{0}}
  = \frac{1}{2^n} \sum_{\vec{a} \in \{\pm 1\}^n} |\vec{a}\rangle\langle \vec{a}|_R \otimes \sigma_{JE},
  \label{sm:dirg:ideal_form}
\end{equation}
where $\sigma_{JE} = \operatorname{tr}_A[|\Phi_n\rangle\langle\Phi_n| \otimes |\zeta\rangle\langle\zeta|] = |\zeta\rangle\langle\zeta|_{JE}$ is independent of the outcome $\vec{a}$.
For even $n$ the ideal cq-state is decoupled but not uniform:
\begin{equation*}
  \rho^{\rm ideal}_{RE|\mathbf{0}}
  = \sum_{\vec{a} \in \{\pm 1\}^n} p_{\rm ideal}(\vec{a})\,|\vec{a}\rangle\langle \vec{a}|_R \otimes \sigma_{JE},
\end{equation*}
with $H\bigl(\{p_{\rm ideal}(\vec{a})\}\bigr) = H_{\rm ideal}(n)$ of Proposition~\ref{sm:dirg:wooltorton}.
Since only the \emph{value} of the ideal conditional entropy enters the entropy-continuity argument below, the parity (and the class modulo $4$) of $n$ affects the final bound solely through $H_{\rm ideal}(n)$.

\emph{Decomposition of the randomness deficit.}
The extracted-picture cq-state $\rho'_{RE|\mathbf{0}}$ of Eq.~\eqref{sm:dirg:actual_cq} is obtained by first applying the extraction isometry $V$ to $|\Psi\rangle_{QE}$, then measuring the logical qubits $A$ with the \emph{transformed physical measurement} operators $P_{\vec{a}}^{(\mathbf{0})}$.
At maximal violation, self-testing forces $\theta_k = \pi/4$, and the transformed measurement operators coincide with the optimal Pauli observables on $A$, so that $\rho^{\rm phys}_{RE|\mathbf{0}} = \rho'_{RE|\mathbf{0}}$.
At $\varepsilon > 0$, however, the measurement channel $\mathcal{M}$ depends on the unknown angles $\{\theta_k\}_{k=1}^n$ and differs from the ideal measurement channel $\mathcal{M}_{\rm ideal}$, which would employ the optimal Pauli projectors $P_{\vec{a}}^{(\mathbf{0}), {\rm ideal}}$ on $A$; moreover the physical measurement no longer commutes exactly with the extraction isometry.

To separate the sources of randomness loss, define the \emph{extracted cq-state with ideal measurement}
\begin{equation*}
  \rho^{\rm extr}_{RE|\mathbf{0}} \coloneqq \operatorname{tr}_J\!\bigl[ \mathcal{M}_{\rm ideal}\!\bigl( |\Psi'\rangle\langle\Psi'|_{AJE} \bigr) \bigr],
\end{equation*}
where $\mathcal{M}_{\rm ideal}$ uses the optimal Pauli projectors $P_a^{(\mathbf{0}), {\rm ideal}}$ in place of $P_a^{(\mathbf{0})}$ in Eq.~\eqref{sm:dirg:measurement_map}.
By the triangle inequality,
\begin{equation}
  \begin{aligned}
    \bigl\| \rho^{\rm phys}_{RE|\mathbf{0}} - \rho^{\rm ideal}_{RE|\mathbf{0}} \bigr\|_1
    &\leq \bigl\| \rho^{\rm phys}_{RE|\mathbf{0}} - \rho'_{RE|\mathbf{0}} \bigr\|_1
       + \bigl\| \rho'_{RE|\mathbf{0}} - \rho^{\rm extr}_{RE|\mathbf{0}} \bigr\|_1
       + \bigl\| \rho^{\rm extr}_{RE|\mathbf{0}} - \rho^{\rm ideal}_{RE|\mathbf{0}} \bigr\|_1.
  \end{aligned}
  \label{sm:dirg:decomposition}
\end{equation}
The third term is the \emph{state-imperfection contribution}: it measures how far the extracted state $|\Psi'\rangle$ deviates from the ideal $|\Phi_n\rangle \otimes |\zeta\rangle$, {both measured with the same ideal channel} $\mathcal{M}_{\rm ideal}$.
By the contractivity argument of Lemma~\ref{lem:sm:dirg:post_measurement_distance}, this term is bounded by $2\sqrt{s_n\varepsilon}$.

The second term is the \emph{measurement-imperfection contribution}:
\begin{equation}
  \bigl\| \rho'_{RE|\mathbf{0}} - \rho^{\rm extr}_{RE|\mathbf{0}} \bigr\|_1
  \leq \bigl\| (\mathcal{M} - \mathcal{M}_{\rm ideal})(|\Psi'\rangle\langle\Psi'|_{AJE}) \bigr\|_1,
  \label{sm:dirg:measurement_gap}
\end{equation}
where the inequality follows from contractivity of the partial trace over $J$; it depends on the difference between the two measurement channels.
For qubit measurements parametrized by angles $\{\theta_k\}_{k=1}^n$ (see the main text), the transformed operators $P_{\vec{a}}^{(\mathbf{0})}$ are smooth functions of $\theta_k$, and $\mathcal{M} = \mathcal{M}_{\rm ideal}$ precisely when $\theta_k = \pi/4$ for all $k$.
Bounding the measurement-imperfection contribution $\|(\mathcal{M} - \mathcal{M}_{\rm ideal})(|\Psi'\rangle\langle\Psi'|)\|_1$ in terms of the MABK violation deficit $\varepsilon$ therefore constitutes a measurement-robust self-testing problem for the $n$-qubit MABK inequality.

The first term in Eq.~\eqref{sm:dirg:decomposition} is the \emph{measurement-translation residual} arising from the non-commutation of the physical measurement with the dephasing part of the extraction channel; it is bounded in Section~\ref{sm:dirg:completion} (Lemma~\ref{lem:sm:dirg:translation}) by $(\sqrt{2}/2)\,nw$.

In the remainder of this section, we bound the state-imperfection contribution and obtain the corresponding entropy lower bound.
The measurement-imperfection contribution and the translation residual are then quantified in Section~\ref{sm:dirg:completion} using the measurement-robustness self-testing result for the MABK inequality established in \cite{Kaniewski_2017}.

\begin{lemma}
  \label{lem:sm:dirg:post_measurement_distance}
  The state-imperfection contribution to the trace distance satisfies
  \begin{equation*}
    \bigl\| \rho^{\rm extr}_{RE|\mathbf{0}} - \rho^{\rm ideal}_{RE|\mathbf{0}} \bigr\|_1
    \leq 2\sqrt{s_n \varepsilon}.
  \end{equation*}
\end{lemma}

\begin{proof}
  Since $\mathcal{M}_{\rm ideal}$ is a CPTP map, it is contractive under the trace norm.
  Applying this to Lemma~\ref{lem:sm:dirg:purification_trace_distance},
  \begin{equation*}
    \begin{aligned}
      \bigl\| \rho^{\rm extr}_{RE|\mathbf{0}} - \rho^{\rm ideal}_{RE|\mathbf{0}} \bigr\|_1
      &= \bigl\| \mathcal{M}_{\rm ideal}(|\Psi'\rangle\langle\Psi'|)
           - \mathcal{M}_{\rm ideal}(|\Phi_n\rangle\langle\Phi_n| \otimes |\zeta\rangle\langle\zeta|) \bigr\|_1 \\
      &\leq \bigl\| |\Psi'\rangle\langle\Psi'|
           - |\Phi_n\rangle\langle\Phi_n| \otimes |\zeta\rangle\langle\zeta| \bigr\|_1 \\
      &\leq 2\sqrt{s_n \varepsilon}.
    \end{aligned}
  \end{equation*}
\end{proof}

\subsection{Conditional Entropy Continuity}
\label{sm:dirg:entropy_continuity}

To translate the trace-distance bound on the cq-states into a bound on the conditional entropy $H(R|E)$, we apply the tight continuity bound for the conditional von Neumann entropy due to Winter~\cite{Winter_2016}, which sharpens the Alicki--Fannes inequality~\cite{Alicki_Fannes_2004}.

\begin{lemma}[Tight continuity bound for the conditional entropy]
  \label{lem:sm:dirg:fannes_audenaert}
  Let $\rho_{AB}$ and $\sigma_{AB}$ be two classical-quantum (cq) states on the same Hilbert space $\mathcal{H}_A \otimes \mathcal{H}_B$ (with $A$ the classical register), and let $d_A = \dim(\mathcal{H}_A)$.
  Denote the full trace distance by $\delta = \|\rho_{AB} - \sigma_{AB}\|_1$.
  Then, for $\delta \leq 2$~\cite{Winter_2016},
  \begin{equation*}
    \bigl| H(A|B)_\rho - H(A|B)_\sigma \bigr|
    \leq \frac{\delta}{2} \log_2 d_A
      + \Bigl(1 + \frac{\delta}{2}\Bigr)\,h\!\left(\frac{\delta}{2+\delta}\right),
  \end{equation*}
  where $h(x) = -x\log_2 x - (1-x)\log_2(1-x)$ is the binary entropy (with $h(0) \coloneqq 0$) and $H(A|B)_\rho \coloneqq H(\rho_{AB}) - H(\rho_B)$.
  This is the classical-quantum case of Lemma~2 of~\cite{Winter_2016}; for two arbitrary quantum states the first term is replaced by $\delta \log_2 d_A$, which recovers the (sharpened) Alicki--Fannes bound~\cite{Alicki_Fannes_2004}.
  For $\delta \leq 1$, the bound $(1+\delta/2) \leq 2$ together with $h(\delta/(2+\delta)) \leq h(\delta/2)$ yields the simplified bound
  \begin{equation*}
    \bigl| H(A|B)_\rho - H(A|B)_\sigma \bigr|
    \leq \frac{\delta}{2} \log_2 d_A + 2\,h\!\left(\frac{\delta}{2}\right).
  \end{equation*}
  If $A$ is a classical register (i.e.\ $\rho_{AB}$ and $\sigma_{AB}$ are classical-quantum states), then $0 \leq H(A|B) \leq \log_2 d_A$ for every such state, and hence
  \begin{equation}
    \bigl| H(A|B)_\rho - H(A|B)_\sigma \bigr| \leq \log_2 d_A
    \label{sm:dirg:afa_cq_trivial}
  \end{equation}
  holds for \emph{all} $\delta \leq 2$.
\end{lemma}

We apply this inequality with $A = R$ (dimension $d_R = 2^n$) and $B = E$.
Setting $\delta = \|\rho^{\rm extr}_{RE|\mathbf{0}} - \rho^{\rm ideal}_{RE|\mathbf{0}}\|_1 \leq 2\sqrt{s_n\varepsilon}$ from Lemma~\ref{lem:sm:dirg:post_measurement_distance}, we obtain, for $\delta \leq 1$ (equivalently $\sqrt{s_n\varepsilon} \le 1/2$),
\begin{equation*}
  \begin{aligned}
    \bigl| H(R|E)_{\rm extr} - H(R|E)_{\rm ideal} \bigr|
    &\leq \frac{\delta}{2} \log_2(2^n) + 2\,h\!\left(\frac{\delta}{2}\right) \\
    &\leq n\sqrt{s_n\varepsilon} + 2\,h\!\bigl(\sqrt{s_n\varepsilon}\,\bigr),
  \end{aligned}
\end{equation*}
where the last inequality uses $h(\delta/2) \le h(\sqrt{s_n\varepsilon})$, valid by monotonicity of $h$ on $[0,1/2]$ provided $\sqrt{s_n\varepsilon} \le 1/2$.
For $\sqrt{s_n\varepsilon} \in (1/2, 1]$ (i.e.\ $\delta \in (1,2]$), the simplified bound of Lemma~\ref{lem:sm:dirg:fannes_audenaert} is not guaranteed to hold (its condition $\delta \le 1$ may fail); instead, since $\rho^{\rm extr}_{RE|\mathbf{0}}$ and $\rho^{\rm ideal}_{RE|\mathbf{0}}$ are classical-quantum states, the trivial bound $\bigl|H(R|E)_{\rm extr} - H(R|E)_{\rm ideal}\bigr| \leq \log_2 d_R = n$ of Eq.~\eqref{sm:dirg:afa_cq_trivial} holds.

For the ideal state, the decoupling discussed above yields $H(R|E)_{\rm ideal} = H_{\rm ideal}(n)$ as given in Proposition~\ref{sm:dirg:wooltorton}.
Therefore, in both regimes,
\begin{equation}
  H(R|E)_{\rm extr} \geq H_{\rm ideal}(n) - \Delta(\varepsilon, n),
  \label{sm:dirg:entropy_lower_bound}
\end{equation}
with the piecewise deficit
\begin{equation}
  \Delta(\varepsilon, n) \coloneqq
  \begin{cases}
    n \sqrt{s_n\varepsilon} + 2\,h\!\bigl( \sqrt{s_n\varepsilon} \,\bigr), & \sqrt{s_n\varepsilon} \le \tfrac{1}{2},\\[2pt]
    n, & \tfrac{1}{2} < \sqrt{s_n\varepsilon} \le 1,
  \end{cases}
  \label{sm:dirg:deficit_def}
\end{equation}
where the first branch is the simplified conditional-entropy deficit of Lemma~\ref{lem:sm:dirg:fannes_audenaert} (valid for $\delta \leq 1$) and the second is the trivial classical-quantum bound of Eq.~\eqref{sm:dirg:afa_cq_trivial}.
Since $\sqrt{s_n\varepsilon} = \sqrt{(1+\sqrt{2})\varepsilon}\,/\,2^{n/4}$, the first branch reads explicitly as $n\sqrt{(1+\sqrt{2})\varepsilon}\,2^{-n/4} + 2h\bigl(\sqrt{(1+\sqrt{2})\varepsilon}\,2^{-n/4}\bigr)$.
Throughout the stated range $\varepsilon \in [0, \beta_{nQ}-\beta_{nC}]$ one has $\sqrt{s_n\varepsilon} \le 1$, so that $\Delta(\varepsilon,n)$ is well defined and finite.

\subsection{State-Imperfection Entropy Bound}
\label{sm:dirg:state_entropy}

Combining the extractability-derived fidelity bound of Eq.~\eqref{sm:dirg:fidelity_bound}, the post-measurement decomposition of Lemma~\ref{lem:sm:dirg:post_measurement_distance}, and the entropy continuity bound of Lemma~\ref{lem:sm:dirg:fannes_audenaert}, we obtain the following certified lower bound on the extracted conditional entropy, \emph{assuming ideal measurements}.

\begin{lemma}[State-Imperfection Entropy Bound]
  \label{lem:sm:dirg:state_entropy}
  For $n \geq 3$, let $\beta = \beta_{nQ} - \varepsilon = 2^{(n+1)/2} - \varepsilon$ be the observed MABK violation with $\varepsilon \in [0, \beta_{nQ} - \beta_{nC}]$, where $\beta_{nC} = 2^{n/2}$ is the biseparable bound.
  Under the extraction channel $\Lambda = \bigotimes_k \Lambda_k$ constructed from the physical observables, and assuming the ideal measurement channel $\mathcal{M}_{\rm ideal}$, the extracted conditional entropy satisfies
  \begin{equation}
    H(R|E)_{\rm extr}\bigl(2^{(n+1)/2} - \varepsilon\bigr)
    \geq H_{\rm ideal}(n) - \Delta(\varepsilon, n),
    \label{sm:dirg:rate_explicit}
  \end{equation}
  where $\Delta(\varepsilon,n)$ is the piecewise deficit of Eq.~\eqref{sm:dirg:deficit_def}.
  In the simplified regime $\sqrt{s_n\varepsilon} \le 1/2$ this reads explicitly as
  \begin{equation}
    H(R|E)_{\rm extr}\bigl(2^{(n+1)/2} - \varepsilon\bigr)
    \geq H_{\rm ideal}(n) - n\,\frac{\sqrt{(1+\sqrt{2})\varepsilon}}{2^{n/4}}
           - 2\,h\!\left( \frac{\sqrt{(1+\sqrt{2})\varepsilon}}{2^{n/4}} \right),
    \label{sm:dirg:rate_explicit_simplified}
  \end{equation}
  while for $1/2 < \sqrt{s_n\varepsilon} \le 1$ the bound $H(R|E)_{\rm extr} \geq H_{\rm ideal}(n) - n$ holds via the trivial classical-quantum bound of Eq.~\eqref{sm:dirg:afa_cq_trivial}.
  The bound relies only on Proposition~\ref{prop:main} (which holds for all measurement angles and dimensions) and the conditional-entropy continuity bound of Lemma~\ref{lem:sm:dirg:fannes_audenaert} together with the classical-quantum entropy bound.
  Here $H_{\rm ideal}(n)$ is the ideal entropy of Proposition~\ref{sm:dirg:wooltorton}: $n$ for odd $n \equiv 3 \pmod 4$ and $n - 1$ for odd $n \equiv 1 \pmod 4$ (on input $\mathbf{0}$; on input $\mathbf{1}$ the latter becomes $n$), and $n + \frac{1}{2} - \log_2(1+\sqrt{2})/\sqrt{2}$ for even $n$.
\end{lemma}

\noindent \emph{Proof.}
The entropy lower bound was derived for the extracted state $\rho^{\rm extr}_{RE|\mathbf{0}} = \operatorname{tr}_J[\mathcal{M}_{\rm ideal}(|\Psi'\rangle\langle\Psi'|)]$ with the optimal measurement channel, where $|\Psi'\rangle = (V \otimes I_E)|\Psi\rangle_{QE}$ and $V = \bigotimes_k V_k$ is the Stinespring dilation of $\Lambda = \bigotimes_k \Lambda_k$.
The isometry $V$ acts only on the device registers and preserves correlations with Eve, so $H(R|E)_{\rm extr} = H(R|E)_{\rho^{\rm extr}}$.
Taking the infimum over all strategies compatible with violation $\geq \beta_{nQ} - \varepsilon$ yields Eq.~\eqref{sm:dirg:rate_explicit}.
\hfill $\square$

Note that the simplified bound of Lemma~\ref{lem:sm:dirg:fannes_audenaert} applies for $\delta \leq 1$, i.e.\ for $\sqrt{s_n\varepsilon} \le 1/2$, equivalently $\varepsilon \le 2^{n/2}/(4(1+\sqrt{2}))$.
Within the remaining part of the stated range, $1/2 < \sqrt{s_n\varepsilon} \le 1$ (equivalently $\varepsilon \in \bigl(2^{n/2}/(4(1+\sqrt{2})),\; \beta_{nQ}-\beta_{nC}\bigr]$), the simplified bound is not directly applicable; the deficit is then bounded by the trivial classical-quantum bound $\log_2 d_R = n$ of Eq.~\eqref{sm:dirg:afa_cq_trivial}, which is valid for all $\delta \le 2$ since the post-measurement states are classical-quantum.
This yields the piecewise deficit $\Delta(\varepsilon,n)$ of Eq.~\eqref{sm:dirg:deficit_def}.
One verifies $2^{n/2}/(1+\sqrt{2}) = \beta_{nQ} - \beta_{nC}$ for all $n \geq 3$, so the condition $\sqrt{s_n\varepsilon} \le 1$ holds throughout the stated range and $\Delta(\varepsilon,n)$ is finite there.

\subsection{Measurement-Imperfection Gap}
\label{sm:dirg:completion}

In this section, we analytically bound the measurement-imperfection contribution $\delta_{\mathcal{M}} = \|(\mathcal{M} - \mathcal{M}_{\rm ideal})(|\Psi'\rangle\langle\Psi'|)\|_1$ identified in Eq.~\eqref{sm:dirg:measurement_gap} by combining Kaniewski's commutation-based self-testing~\cite{Kaniewski_2017} with a diamond-norm analysis of the measurement channel, and we bound the measurement-translation residual arising from the non-commutation of the physical measurement with the extraction isometry (Lemma~\ref{lem:sm:dirg:translation}).
This further completes the robustness analysis of the physical DIRG protocol,  yielding a fully device-independent bound on the conditional entropy $H(R|E)$ that accounts for both state and measurement imperfections.

\subsubsection{Effective commutator and measurement angles}
\label{sm:dirg:comp_commutator}

As discussed in the main text,
the two observables of party~$k$ acting on the logical qubit take the form
\begin{equation}
  O_{0}^{(k)} = \cos\theta_k\,\sigma_x + \sin\theta_k\,\sigma_z,
  \qquad
  O_{1}^{(k)} = \cos\theta_k\,\sigma_x - \sin\theta_k\,\sigma_z,
  \label{sm:dirg:comp_obs}
\end{equation}
with angles $\theta_k \in [0,\pi/2]$.
At the optimal point $\theta_k = \pi/4$ the two observables are maximally complementary;
a direct calculation yields the commutator
\begin{equation}
  \bigl[O_{0}^{(k)}, O_{1}^{(k)}\bigr] = 2i\sin(2\theta_k)\,\sigma_y.
  \label{sm:dirg:comp_comm}
\end{equation}
Following Kaniewski~\cite{Kaniewski_2017}, we denote the \emph{matrix modulus}
of an operator $X$ by $|X| \coloneqq \sqrt{X^\dagger X}$,
which for the commutator above gives
\begin{equation}
  \bigl|\bigl[O_{0}^{(k)}, O_{1}^{(k)}\bigr]\bigr| = 2|\sin(2\theta_k)|\,I.
  \label{sm:dirg:comp_abscomm}
\end{equation}

The \emph{effective commutator} introduced by Kaniewski~\cite{Kaniewski_2017} is
\begin{equation}
  t_k \coloneqq \frac{1}{2}\operatorname{tr}\!\Bigl(\bigl|[A_0^{(k)}, A_1^{(k)}]\bigr|\;\rho_k\Bigr),
  \label{sm:dirg:comp_tk}
\end{equation}
where $\rho_k$ is the reduced state of party~$k$ on the support of the physical state.
By Jordan's lemma each pair of physical observables decomposes into a direct sum of $2\times2$ blocks, and the block-diagonal unitary part of the extraction (see the main text) maps each block $f$ to the logical form~\eqref{sm:dirg:comp_obs} with its own angle $\theta_k^f \in [0,\pi/2]$ (i.e.\ $A_0^{(k)} = U_k^\dagger O_0^{(k)} U_k$ and $A_1^{(k)} = U_k^\dagger O_1^{(k)} U_k$ for a block-diagonal unitary $U_k$).
The reduced state decomposes accordingly, $\rho_k = \sum_f p_{k,f}\,\rho_k^f$, with weights $p_{k,f} \ge 0$ and $\sum_f p_{k,f} = 1$.
Since the matrix modulus is unitarily invariant and block-diagonal, taking the trace with $\rho_k$ and using Eq.~(\ref{sm:dirg:comp_abscomm}) gives
\begin{equation}
  t_k = \sum_f p_{k,f}\,\bigl|\sin(2\theta_k^f)\bigr|
     = \sum_f p_{k,f}\,\cos(2\delta_k^f),
  \qquad
  \delta_k^f \coloneqq \theta_k^f - \frac{\pi}{4},
  \label{sm:dirg:comp_tk_angle}
\end{equation}
where the second equality uses $\theta_k^f \in [0,\pi/2]$, so that $2\delta_k^f \in [-\pi/2,\pi/2]$ and hence $|\sin 2\theta_k^f| = \cos 2\delta_k^f \ge 0$.
A lower bound on $t_k$ therefore constrains the state-weighted distribution of the angle deviations rather than any single block angle; this is all that is needed below.

\subsubsection{From MABK violation to the effective-commutator deficit}
\label{sm:dirg:comp_violation}

Kaniewski~\cite{Kaniewski_2017} proved a tight operator inequality for the $n$-party MABK
Bell operator (Proposition~B.2 of~\cite{Kaniewski_2017}).
Specialising to certify the $k$-th party and tracing with the physical state,
the observed violation $\beta = \langle W_n\rangle$ and the effective commutator
of \emph{each} party satisfy
\begin{equation}
  \beta \leq \sqrt{2^{\,n-2}}\;\sqrt{\,1 + t_k\,}, \qquad \forall\,k = 1,\dots,n.
  \label{sm:dirg:comp_kaniewski_orig}
\end{equation}

In our normalisation convention $W_1 = 2 O_{r_1}$ (whereas \cite{Kaniewski_2017} uses~$W_1 = A_0^1$),
the Bell operator is multiplied by an overall factor of~$2$, and the maximal quantum
violation reads $\beta_{nQ} = 2^{(n+1)/2}$.
Rescaling~(\ref{sm:dirg:comp_kaniewski_orig}) accordingly gives
\begin{equation}
  \beta \leq \beta_{nQ}\,\sqrt{\frac{1 + t_k}{2}}, \qquad
  \beta_{nQ} = 2^{(n+1)/2}.
  \label{sm:dirg:comp_kaniewski}
\end{equation}

Writing the observed violation as $\beta = \beta_{nQ} - \varepsilon$
with $\varepsilon \in [0,\, \beta_{nQ} - \beta_{nC}]$
and solving for $t_k$, we obtain
\begin{equation*}
  t_k \geq 2\Bigl(1 - \frac{\varepsilon}{\beta_{nQ}}\Bigr)^{\!2} - 1
       = 1 - \frac{4\varepsilon}{\beta_{nQ}} + \frac{2\varepsilon^{2}}{\beta_{nQ}^{2}}.
\end{equation*}

Together with~(\ref{sm:dirg:comp_tk_angle}) this yields
\begin{equation}
  \sum_f p_{k,f}\,\cos(2\delta_k^f) \geq 1 - w,
  \qquad
  w \coloneqq \frac{4\varepsilon}{\beta_{nQ}} - \frac{2\varepsilon^{2}}{\beta_{nQ}^{2}},
  \label{sm:dirg:comp_x_def}
\end{equation}
and, using $1 - \cos(2\delta) = 2\sin^{2}\delta$,
\begin{equation}
  \sum_f p_{k,f}\,\sin^{2}\delta_k^f \leq \frac{w}{2}.
  \label{sm:dirg:comp_blockbound}
\end{equation}
Equation~\eqref{sm:dirg:comp_blockbound} is the block-weighted analogue of an angle-deviation bound: it does not force any single block angle to be close to $\pi/4$, but it controls the state-weighted second moment of the deviations, which is precisely what enters the trace-distance analysis below.
For $\varepsilon\ll\beta_{nQ}$, $w \approx 4\varepsilon/\beta_{nQ}$, so the characteristic deviation scale is $\sqrt{w/2} \approx \sqrt{2\varepsilon/\beta_{nQ}} = \sqrt{\varepsilon}\;2^{-(n-1)/4}$.

\subsubsection{Measurement-translation residual}
\label{sm:dirg:comp_translation_sec}

As noted in Section~\ref{sm:dirg:post_measurement}, the physical measurement does not commute exactly with the extraction isometry.
In the extracted picture the physical projector of party~$k$ reads $\widetilde{A}_{a_k}^{(k)} \coloneqq V_k P_{a_k}^{(k,\rm phys)} V_k^\dagger$, whereas the logical measurement of input~$\mathbf{0}$ uses the projectors $P_{a_k}^{(k,0)}$ built from the observables of Eq.~\eqref{sm:dirg:comp_obs} (extended by the identity on the junk).
For a single block~$f$ of party~$k$, writing the dephasing dilation as $V_D|\psi\rangle = \alpha|\psi\rangle|0\rangle + \beta\Gamma|\psi\rangle|1\rangle$ with $\alpha^2 = (1+g_f)/2$ and $\beta^2 = (1-g_f)/2$, a direct calculation gives, for any block component $\sigma^f = V_D X^f V_D^\dagger$ of the extracted-picture state (with $X^f = U_k \tilde\rho^f U_k^\dagger$ and $\tilde\rho^f$ the corresponding block-$f$ component of the physical state),
\begin{equation}
  \operatorname{tr}_{A_k J_k}\!\bigl[ \bigl(P_{a_k}^{(k,0)} \otimes I_{J_k}\bigr)\,\sigma^f \bigr]
  - \operatorname{tr}_{A_k J_k}\!\bigl[ \widetilde{A}_{a_k}^{(k)}\,\sigma^f \bigr]
  = \beta_f^2\,\operatorname{tr}\!\bigl[ \bigl(\Gamma P_{a_k}^{(k,0)}\Gamma - P_{a_k}^{(k,0)}\bigr) \otimes I_{\mathrm{rest}}\,X^f \bigr],
  \label{sm:dirg:comp_translation_block}
\end{equation}
where the first term uses $V_D^\dagger(P \otimes I_J)V_D = \alpha^2 P + \beta^2 \Gamma P\Gamma$.
Telescoping over the $n$ parties as in Eq.~\eqref{sm:dirg:comp_telescoping} (the measurements of the other parties act on disjoint subsystems and do not alter the party-$j$ block weights $\{p_{j,f}\}_f$ of Eq.~\eqref{sm:dirg:comp_tk_angle}), and using $\|\Gamma P_{a_k}^{(k,0)}\Gamma - P_{a_k}^{(k,0)}\|_\infty = \min(\sin\theta_k, \cos\theta_k)$, we obtain
\begin{equation}
  \bigl\| \rho^{\rm phys}_{RE|\mathbf{0}} - \rho'_{RE|\mathbf{0}} \bigr\|_1
  \leq \sum_{j=1}^{n} \sum_{f} p_{j,f}\,\bigl(1 - g_{j,f}\bigr)\,
    \min\bigl(\sin\theta_{j,f},\, \cos\theta_{j,f}\bigr).
  \label{sm:dirg:comp_translation_bound}
\end{equation}

\begin{observation}
  \label{obv:sm:dirg:translation_ratio}
  For all $\theta \in [0,\pi/2]$,
  \begin{equation}
    \bigl(1 - g(\theta)\bigr)\min(\sin\theta, \cos\theta)
    \leq \sqrt{2}\,\sin^2\!\Bigl(\theta - \frac{\pi}{4}\Bigr).
    \label{sm:dirg:comp_ratio}
  \end{equation}
\end{observation}

\noindent \emph{Proof.}
Let $\varphi = |\theta - \pi/4| \in [0,\pi/4]$.
Since $\sin\theta + \cos\theta = \sqrt{2}\cos\varphi$, we have $1 - g(\theta) = (2+\sqrt{2})(1 - \cos\varphi)$.
Moreover $\min(\sin\theta, \cos\theta) = (\cos\varphi - \sin\varphi)/\sqrt{2} \leq 1/\sqrt{2}$ and
$(1-\cos\varphi)/\sin^2\varphi = 1/(2\cos^2(\varphi/2)) \leq 1/(2\cos^2(\pi/8))$.
Combining these bounds and using $2\cos^2(\pi/8) = (2+\sqrt{2})/2$ gives
$(1-g)\min \leq (2+\sqrt{2})\cdot\bigl(2\cos^2(\pi/8)\bigr)^{-1}\cdot(1/\sqrt{2})\cdot\sin^2\varphi = \sqrt{2}\,\sin^2\varphi$.
\hfill$\square$

\begin{lemma}[Measurement-Translation Residual]
  \label{lem:sm:dirg:translation}
  For any quantum strategy achieving MABK violation $\beta = \beta_{nQ} - \varepsilon$,
  the physical and extracted-picture post-measurement cq-states satisfy
  \begin{equation}
    \bigl\| \rho^{\rm phys}_{RE|\mathbf{0}} - \rho'_{RE|\mathbf{0}} \bigr\|_1
    \leq \frac{\sqrt{2}}{2}\,n\,w,
    \label{sm:dirg:comp_translation}
  \end{equation}
  with $w = 4\eta - 2\eta^2$ and $\eta = \varepsilon\,\beta_{nQ}^{-1}$.
\end{lemma}

\noindent \emph{Proof.}
Apply Observation~\ref{obv:sm:dirg:translation_ratio} block-wise in Eq.~\eqref{sm:dirg:comp_translation_bound} and invoke the block-weighted angle bound $\sum_f p_{j,f}\sin^2\delta_{j,f} \leq w/2$ of Eq.~\eqref{sm:dirg:comp_blockbound} for each party~$j$:
$\|\rho^{\rm phys} - \rho'\|_1 \leq \sqrt{2}\,\sum_{j=1}^{n} \sum_f p_{j,f}\sin^2\delta_{j,f} \leq \sqrt{2}\,n\,(w/2)$.
\hfill$\square$

\emph{Remark.}
At the self-testing point $\theta_k = \pi/4$ for all $k$ (and at the endpoints $\theta_k \in \{0,\pi/2\}$) the dephasing is trivial or the projector is $\Gamma$-invariant, and the translation residual vanishes, consistent with the exact translation of the physical measurement at maximal violation.

\subsubsection{Diamond-norm bound on the measurement channel}
\label{sm:dirg:comp_diamond}

We now translate the angle-deviation bound into a bound on the measurement channel.
Recall from Eq.~(\ref{sm:dirg:measurement_map}) that $\mathcal{M}$ records the outcomes
$\vec{a} \in \{\pm1\}^{n}$ of the measurement with input~$\mathbf{0}$ on the logical
qubits $A_1\cdots A_n$:
\begin{equation*}
  \mathcal{M}(\rho_{AJE}) = \sum_{\vec{a}} |\vec{a}\rangle\langle\vec{a}|_{R}
  \otimes \operatorname{tr}_{A}\!\bigl[(P_{\vec{a}}^{(\mathbf{0})} \otimes I_{JE})\,\rho_{AJE}\bigr],
\end{equation*}
where $P_{\vec{a}}^{(\mathbf{0})} = \bigotimes_{k=1}^{n} P_{a_k}^{(k,0)}$.
From~(\ref{sm:dirg:comp_obs}), the projectors for party~$k$ and outcome $a_k = \pm 1$ are
\begin{equation*}
  P_{a_k}^{(k,0)} = \frac{I + a_k(\cos\theta_k\,\sigma_x + \sin\theta_k\,\sigma_z)}{2},
  \qquad
  P_{a_k}^{(k,0,\rm ideal)} = \frac{I + a_k\,\frac{\sigma_x + \sigma_z}{\sqrt{2}}}{2}.
\end{equation*}

For a \emph{single} party, the measurement channel is a binary-outcome quantum instrument
mapping states on $A_k\otimes J\otimes E$ to classical-quantum states on $R_k\otimes J\otimes E$.

\begin{lemma}
  \label{lem:sm:dirg:comp_diamond_binary}
  Let $\Lambda_M, \Lambda_N$ be two binary-outcome measurement channels defined by
  POVM elements $\{M_0, M_1 = I-M_0\}$ and $\{N_0, N_1 = I-N_0\}$ acting on a
  system $A$, i.e.\ for any state $\rho$ on $A\otimes R$ (with $R$ an arbitrary
  reference system),
  \begin{equation*}
    \Lambda_M(\rho) = \sum_{b=0,1} |b\rangle\langle b| \otimes
    \operatorname{tr}_A\!\bigl[(M_b \otimes I_R)\,\rho\bigr],
  \end{equation*}
  and similarly for $\Lambda_N$.
  Then the diamond-norm distance between the two channels satisfies the exact identity
  \begin{equation*}
    \|\Lambda_{M} - \Lambda_{N}\|_{\diamond} = 2\,\|M_0 - N_0\|_{\infty}.
  \end{equation*}
\end{lemma}

\begin{proof}
  Let $\Delta \coloneqq M_0 - N_0$; then $M_1 - N_1 = (I-M_0) - (I-N_0) = -\Delta$.
  The channel difference acts as
  \begin{equation*}
    (\Lambda_M - \Lambda_N)(\rho) = (|0\rangle\langle 0| - |1\rangle\langle 1|)
    \otimes \operatorname{tr}_A\!\bigl[(\Delta \otimes I_R)\,\rho\bigr].
  \end{equation*}
  Tensoring with an identity channel on an arbitrary reference system $R'$,
  \begin{equation*}
    \bigl((\Lambda_M - \Lambda_N) \otimes \mathrm{id}_{R'}\bigr)(\omega)
    = (|0\rangle\langle 0| - |1\rangle\langle 1|)
    \otimes \operatorname{tr}_A\!\bigl[(\Delta \otimes I_{R \otimes R'})\,\omega\bigr]
  \end{equation*}
  for any state $\omega$ on $A \otimes R \otimes R'$.
  Taking the trace norm and using $\||0\rangle\langle 0| - |1\rangle\langle 1|\|_1 = 2$,
  \begin{equation*}
    \bigl\|\bigl((\Lambda_M - \Lambda_N) \otimes \mathrm{id}_{R'}\bigr)(\omega)\bigr\|_1
    = 2\,\bigl\|\operatorname{tr}_A\!\bigl[(\Delta \otimes I_{R \otimes R'})\,\omega\bigr]\bigr\|_1.
  \end{equation*}
  The supremum over $\omega$ of the right-hand side is precisely the
  $1 \to 1$ norm of the map $\sigma \mapsto \operatorname{tr}_A[(\Delta \otimes I)\,\sigma]$,
  which equals $\|\Delta\|_{\infty}$ (the operator norm of $\Delta$).
  To see this, note that for any $\sigma \ge 0$,
  $\|\operatorname{tr}_A[(\Delta\otimes I)\,\sigma]\|_1 \le \|\Delta\|_{\infty}\operatorname{tr}(\sigma)$
  by the H\"older inequality $\|\operatorname{tr}_A[(\Delta\otimes I)\,\sigma]\|_1 \le \|\Delta\|_{\infty}\,\|\sigma\|_1$, and equality is attained by choosing
  $\sigma = |\psi\rangle\langle\psi| \otimes \tau$ where $|\psi\rangle$ is an eigenvector
  of $\Delta$ corresponding to its largest absolute eigenvalue and $\tau$ is any normalised
  state on the remaining registers.
  Hence $\|\Lambda_M - \Lambda_N\|_{\diamond} = 2\|\Delta\|_{\infty} = 2\|M_0 - N_0\|_{\infty}$.
\end{proof}

Specialising to our case, the operator-norm difference of the projectors is
\begin{equation*}
  \begin{aligned}
    \|P_{+1}^{(k,0)} - P_{+1}^{(k,0,\rm ideal)}\|_{\infty}
    &= \frac{1}{2}\Bigl\|
      \bigl(\cos\theta_k - \tfrac{1}{\sqrt{2}}\bigr)\sigma_x
      + \bigl(\sin\theta_k - \tfrac{1}{\sqrt{2}}\bigr)\sigma_z
    \Bigr\|_{\infty} \\[4pt]
    &= \frac{1}{2}\sqrt{\bigl(\cos\theta_k - \tfrac{1}{\sqrt{2}}\bigr)^{2}
                     + \bigl(\sin\theta_k - \tfrac{1}{\sqrt{2}}\bigr)^{2}}.
  \end{aligned}
\end{equation*}

Setting $\delta_k \coloneqq \theta_k - \pi/4$ and using elementary trigonometric identities,
\begin{align*}
  \cos\theta_k &= \frac{\cos\delta_k - \sin\delta_k}{\sqrt{2}},
  \qquad
  \sin\theta_k = \frac{\cos\delta_k + \sin\delta_k}{\sqrt{2}}.
\end{align*}
Squaring and summing the coordinate differences yields
\begin{align*}
  \bigl(\cos\theta_k - \tfrac{1}{\sqrt{2}}\bigr)^{2}
  + \bigl(\sin\theta_k - \tfrac{1}{\sqrt{2}}\bigr)^{2}
  &= \frac{(\cos\delta_k - \sin\delta_k - 1)^{2}
          + (\cos\delta_k + \sin\delta_k - 1)^{2}}{2} \\
  &= 2(1 - \cos\delta_k) = 4\sin^{2}\!\left(\frac{\delta_k}{2}\right).
\end{align*}
Consequently,
\begin{equation}
  \|P_{+1}^{(k,0)} - P_{+1}^{(k,0,\rm ideal)}\|_{\infty}
  = \frac{1}{2}\sqrt{4\sin^{2}(\delta_k/2)}
  = \bigl|\sin(\delta_k/2)\bigr|.
  \label{sm:dirg:comp_proj_diff_exact}
\end{equation}
Lemma~\ref{lem:sm:dirg:comp_diamond_binary} then gives, for each single block~$f$,
\begin{equation}
  \bigl\|\mathcal{M}_k^{f} - \mathcal{M}_k^{\rm ideal}\bigr\|_{\diamond}
  = 2\,\bigl|\sin(\delta_k^{f}/2)\bigr|,
  \label{sm:dirg:comp_single_diamond}
\end{equation}
whose Taylor expansion
$2\sin(|\delta_k^{f}|/2) = |\delta_k^{f}| + O(|\delta_k^{f}|^{3})$
exhibits the linear order of the per-block measurement deviation.
The physical measurement channel on the logical register decomposes block-wise, $\mathcal{M}_k = \bigoplus_f \mathcal{M}_k^{f}$, and the logical reduced state is the corresponding mixture $\rho_{A_k} = \sum_f p_{k,f}\,\rho_{A_k}^{f}$; the extraction is block-respecting, so the weights coincide with those of Eq.~\eqref{sm:dirg:comp_tk_angle}.

\emph{From single-party to $n$-party via channel composition.}
Let $\widetilde{\mathcal{M}}_k$ denote the single-party measurement channel~$\mathcal{M}_k$
extended by the identity channel on all other parties, i.e.\ $\widetilde{\mathcal{M}}_k = \mathcal{M}_k \otimes \mathrm{id}^{\otimes(n-1)}$.
$\widetilde{\mathcal{M}}_k$ is a CPTP map acting on the full $n$-party space, recording the outcome of party~$k$
in its own classical register.  Since the $n$ measurements act on different subsystems, the corresponding
extended channels commute under composition, and the joint $n$-party channel factorizes as
\begin{equation*}
  \mathcal{M} = \widetilde{\mathcal{M}}_1 \circ \widetilde{\mathcal{M}}_2 \circ \cdots \circ \widetilde{\mathcal{M}}_n,
  \qquad
  \mathcal{M}_{\rm ideal} = \widetilde{\mathcal{M}}_1^{\rm ideal} \circ \widetilde{\mathcal{M}}_2^{\rm ideal} \circ \cdots \circ \widetilde{\mathcal{M}}_n^{\rm ideal}.
\end{equation*}
(Any ordering yields the same channel; composition is associative and the classical registers accumulate into the
joint $n$-bit register $R$.)

To verify that this composition coincides with the POVM-based definition~(\ref{sm:dirg:measurement_map}),
let $\rho$ be an arbitrary $n$-party state.
Each $\widetilde{\mathcal{M}}_k$ acts nontrivially only on qubit~$k$, recording outcome $a_k \in \{\pm1\}$
in its classical register.
Since $[\widetilde{\mathcal{M}}_i, \widetilde{\mathcal{M}}_j] = 0$ for $i \neq j$ (because they act on disjoint subsystems),
the composition is independent of ordering, and expanding gives
\begin{equation*}
  \bigl(\widetilde{\mathcal{M}}_1 \circ \cdots \circ \widetilde{\mathcal{M}}_n\bigr)(\rho)
  = \sum_{\vec{a} \in \{\pm1\}^{n}}
    |\vec{a}\rangle\langle\vec{a}|_{R}
    \otimes \operatorname{tr}_{A}\!\bigl[
      \bigl(\bigotimes_{k=1}^{n} P_{a_k}^{(k,0)}\bigr) \otimes I_{JE}\,\rho \bigr],
\end{equation*}
which is precisely the action of $\mathcal{M}$ in Eq.~\eqref{sm:dirg:measurement_map}.
The same holds for $\mathcal{M}_{\rm ideal}$ with projectors $P_{a_k}^{(k,0,\rm ideal)}$.

Writing the channel difference as a telescoping sum of compositions,
\begin{align}
  \mathcal{M} - \mathcal{M}_{\rm ideal}
  &= \sum_{j=1}^{n} \,
     \widetilde{\mathcal{M}}_1^{\rm ideal} \circ \cdots \circ \widetilde{\mathcal{M}}_{j-1}^{\rm ideal}
     \circ \bigl(\widetilde{\mathcal{M}}_j - \widetilde{\mathcal{M}}_j^{\rm ideal}\bigr)
     \circ \widetilde{\mathcal{M}}_{j+1} \circ \cdots \circ \widetilde{\mathcal{M}}_n .
  \label{sm:dirg:comp_telescoping}
\end{align}
Every map preceding or following the $j$-th difference in the composition is a CPTP channel and hence contractive in the trace norm.
Since the extraction is block-respecting and the measurements of the other parties act on different subsystems (preserving the reduced state of party~$j$), each intermediate state in the telescoping sum carries the same party-$j$ block weights $\{p_{j,f}\}_f$ as in Eq.~\eqref{sm:dirg:comp_tk_angle}.
Applying the telescoping sum to $|\Psi'\rangle\langle\Psi'|$ and using the per-block diamond norm~\eqref{sm:dirg:comp_single_diamond},
\begin{align}
  \bigl\|(\mathcal{M} - \mathcal{M}_{\rm ideal})\bigl(|\Psi'\rangle\langle\Psi'|\bigr)\bigr\|_1
  &\leq \sum_{j=1}^{n}\,\sum_{f} p_{j,f}\,2\bigl|\sin(\delta_j^{f}/2)\bigr|
  \nonumber\\
  &\leq \sum_{j=1}^{n} 2\sum_{f} p_{j,f}\,\bigl|\sin\delta_j^{f}\bigr|
  \leq \sum_{j=1}^{n} 2\sqrt{\sum_{f} p_{j,f}\sin^{2}\delta_j^{f}}
  \leq n\sqrt{2w},
  \label{sm:dirg:comp_ndiamond_sum}
\end{align}
where the second step uses $|\sin(\delta/2)| \le |\sin\delta|$ for $|\delta|\le\pi/2$, the third is the weighted Cauchy--Schwarz inequality, and the last invokes the block-weighted bound~\eqref{sm:dirg:comp_blockbound}.
Note that this bounds the trace distance evaluated on the actual logical state $|\Psi'\rangle$ rather than the diamond norm of the channel difference; bounding the diamond norm would require controlling the maximum deviation over all blocks, which the state-weighted bound~\eqref{sm:dirg:comp_blockbound} does not provide.

The prefactor $n$ in the leading term arises from the sum over parties.
Whether a more refined bound can suppress this factor is an interesting question that we leave for future investigation.

\subsubsection{Trace-distance combination and measurement-imperfection lemma}
\label{sm:dirg:comp_combination}

We now assemble the total trace-distance bound between the physical and ideal cq-states.
For concise reference we introduce
\begin{equation}
  \eta \coloneqq \frac{\varepsilon}{\beta_{nQ}} = \varepsilon\cdot 2^{-(n+1)/2},
  \qquad
  w = 4\eta - 2\eta^{2},
  \label{sm:dirg:comp_aux}
\end{equation}
so that the measurement-imperfection bound of
Eq.~\eqref{sm:dirg:comp_ndiamond_sum} reads
\begin{equation}
  \bigl\|(\mathcal{M} - \mathcal{M}_{\rm ideal})\bigl(|\Psi'\rangle\langle\Psi'|\bigr)\bigr\|_1
  \leq n\sqrt{2w}
  = 2n\sqrt{2\eta - \eta^{2}}.
  \label{sm:dirg:comp_diamond_final}
\end{equation}

Recall the triangle-inequality decomposition of the trace distance
between the physical and ideal cq-states,
Eq.~(\ref{sm:dirg:decomposition}), together with the
measurement-imperfection bound of Eq.~\eqref{sm:dirg:comp_ndiamond_sum},
the state-imperfection bound of Lemma~\ref{lem:sm:dirg:post_measurement_distance},
and the measurement-translation residual of Lemma~\ref{lem:sm:dirg:translation}:
\begin{equation*}
  \bigl\| \rho^{\rm phys}_{RE|\mathbf{0}} - \rho^{\rm ideal}_{RE|\mathbf{0}} \bigr\|_1
  \leq n\sqrt{2w} + 2\sqrt{s_n\varepsilon} + \frac{\sqrt{2}}{2}\,n\,w.
\end{equation*}

\begin{lemma}[Measurement-Imperfection Trace-Distance Bound]
  \label{lem:sm:dirg:measurement_bound}
  For any quantum strategy achieving MABK violation $\beta = \beta_{nQ} - \varepsilon$,
  the measurement-imperfection contribution to the cq-state trace distance satisfies
  \begin{equation}
    \bigl\| (\mathcal{M} - \mathcal{M}_{\rm ideal})(|\Psi'\rangle\langle\Psi'|) \bigr\|_1
    \leq n\sqrt{2w},
    \qquad
    w = 4\eta - 2\eta^{2},
    \label{sm:dirg:comp_measurement_bound}
  \end{equation}
  with $\eta = \varepsilon\,2^{-(n+1)/2}$.
  To leading order in $\varepsilon$ this reduces to
  $2\sqrt{2}\,n\sqrt{\varepsilon}\;2^{-(n+1)/4}$.
\end{lemma}

\noindent \emph{Proof.}
The bound follows from the telescoping decomposition of
Eq.~\eqref{sm:dirg:comp_telescoping}, the per-block single-party diamond norm of
Eq.~\eqref{sm:dirg:comp_single_diamond} (via Lemma~\ref{lem:sm:dirg:comp_diamond_binary}),
and the block-weighted angle bound of Eq.~\eqref{sm:dirg:comp_blockbound}
(through the Cauchy--Schwarz step in Eq.~\eqref{sm:dirg:comp_ndiamond_sum}).
\hfill $\square$

The third term in the decomposition is the state-imperfection contribution,
bounded by Lemma~\ref{lem:sm:dirg:post_measurement_distance}:
$\|\mathcal{M}_{\rm ideal}(|\Psi'\rangle\langle\Psi'|
   - |\Phi_n\rangle\langle\Phi_n|\otimes|\zeta\rangle\langle\zeta|)\|_1
 \leq 2\sqrt{s_n\varepsilon}$.
The first term (the measurement-translation residual) is bounded by
Lemma~\ref{lem:sm:dirg:translation}.
Altogether,
\begin{equation}
  \delta_{\rm total}
  \coloneqq \bigl\| \rho^{\rm phys}_{RE|\mathbf{0}} - \rho^{\rm ideal}_{RE|\mathbf{0}} \bigr\|_1
  \leq 2\sqrt{s_n\varepsilon}
     + n\sqrt{2w}
     + \frac{\sqrt{2}}{2}\,n\,w
     = 2\sqrt{s_n\varepsilon} + 2n\sqrt{2\eta - \eta^{2}} + \sqrt{2}\,n\,\eta\,(2-\eta),
  \label{sm:dirg:comp_total}
\end{equation}
where the three terms are bounded by
Lemmas~\ref{lem:sm:dirg:post_measurement_distance},
\ref{lem:sm:dirg:measurement_bound} and \ref{lem:sm:dirg:translation}, respectively.

The translation of this trace-distance bound into a conditional-entropy bound
via the conditional-entropy continuity bound of Lemma~\ref{lem:sm:dirg:fannes_audenaert} is carried out in
Proposition~\ref{prop:sm:dirg:complete}.


\subsection{Complete DIRG Randomness Bound}
\label{sm:dirg:complete}

Having bounded all sources of randomness deficit, we now combine them into a single device-independent lower bound on the conditional entropy $H(R|E)$ for the physical DIRG protocol.

\begin{proposition}[DIRG Randomness Bound]
  \label{prop:sm:dirg:complete}
  For $n \geq 3$, let $\beta = \beta_{nQ} - \varepsilon$ be the observed MABK violation
  with $\varepsilon \in [0,\, \beta_{nQ} - \beta_{nC}]$.
  Define the auxiliary quantities
  \begin{equation*}
    \eta \coloneqq \varepsilon\cdot 2^{-(n+1)/2},
    \qquad
    w \coloneqq 4\eta - 2\eta^{2},
    \qquad
    s_n \coloneqq \frac{1+\sqrt{2}}{2^{n/2}}.
  \end{equation*}
  Let $\delta_{\rm total} \coloneqq 2\sqrt{s_n\varepsilon}
  + n\sqrt{2w} + \frac{\sqrt{2}}{2}\,n\,w
  = 2\sqrt{s_n\varepsilon} + 2n\sqrt{2\eta - \eta^{2}} + \sqrt{2}\,n\,\eta\,(2-\eta)$ bound the total trace distance between the physical
  and ideal cq-states (Eq.~\eqref{sm:dirg:comp_total}).
  Under the extraction isometry $V = \bigotimes_k V_k$ and the physical measurement
  channel $\mathcal{M}$, the conditional entropy satisfies
  \begin{equation}
    H(R|E)\bigl(2^{(n+1)/2} - \varepsilon\bigr)
    \geq H_{\rm ideal}(n)
      - \frac{\delta_{\rm total}}{2}\,n
      - \Bigl(1 + \frac{\delta_{\rm total}}{2}\Bigr)\,
        h\!\left(\frac{\delta_{\rm total}}{2}\right),
    \label{sm:dirg:complete_explicit}
  \end{equation}
  where $h(\cdot)$ is the binary entropy and the inequality is valid for $\delta_{\rm total} \leq 1$ (Lemma~\ref{lem:sm:dirg:fannes_audenaert}).
  For $\delta_{\rm total} > 1$, the simplified bound of Lemma~\ref{lem:sm:dirg:fannes_audenaert} is replaced by the trivial classical-quantum bound of Eq.~\eqref{sm:dirg:afa_cq_trivial}, yielding $H(R|E) \geq H_{\rm ideal}(n) - n$.
  For $\delta_{\rm total} \leq 1$, the simplified form
  \begin{equation}
    H(R|E)(\beta_{nQ} - \varepsilon)
    \geq H_{\rm ideal}(n)
      - n\sqrt{s_n\varepsilon}
      - n^{2}\sqrt{w/2}
      - \frac{n^{2}w}{2\sqrt{2}}
      - 2\,h\!\left(\sqrt{s_n\varepsilon} + n\sqrt{w/2} + \frac{n w}{2\sqrt{2}}\right)
    \label{sm:dirg:complete_simplified}
  \end{equation}
  holds.
  In the limit $\varepsilon \to 0^{+}$, the bound recovers the ideal values
  $H_{\rm ideal}(n)$ of Proposition~\ref{sm:dirg:wooltorton}:
  $n$ for odd $n \equiv 3 \pmod 4$; $n - 1$ for odd $n \equiv 1 \pmod 4$ on input $\mathbf{0}$ (and $n$ on input $\mathbf{1}$);
  $n + \frac{1}{2} - \log_2(1+\sqrt{2})/\sqrt{2}$ for even $n$.
\end{proposition}

\noindent \emph{Proof.}
By the triangle-inequality decomposition of Eq.~\eqref{sm:dirg:decomposition},
the total trace distance $\delta_{\rm total}$ is bounded by the sum of the
state-imperfection distance $2\sqrt{s_n\varepsilon}$
(Lemma~\ref{lem:sm:dirg:post_measurement_distance}), the
measurement-imperfection distance
$n\sqrt{2w}$
(Lemma~\ref{lem:sm:dirg:measurement_bound}), and the
measurement-translation residual
$(\sqrt{2}/2)\,n\,w$
(Lemma~\ref{lem:sm:dirg:translation}).
Applying the conditional-entropy continuity bound of
(Lemma~\ref{lem:sm:dirg:fannes_audenaert}) with $d_R = 2^{n}$ and
$\delta = \delta_{\rm total}$ directly yields~\eqref{sm:dirg:complete_explicit}
whenever $\delta_{\rm total} \leq 1$.
For $\delta_{\rm total} > 1$, the simplified bound of Lemma~\ref{lem:sm:dirg:fannes_audenaert} does not apply; since the
post-measurement states are classical-quantum, the bound
$|H(R|E)_\rho - H(R|E)_\sigma| \leq \log_2 d_R = n$ of
Eq.~\eqref{sm:dirg:afa_cq_trivial} holds instead, giving
$H(R|E) \geq H_{\rm ideal}(n) - n$.
The simplified form~\eqref{sm:dirg:complete_simplified} follows from
the bound $(1+\delta/2) \leq 2$ (valid for $\delta \leq 1$),
together with the substitution $\delta_{\rm total}/2 = \sqrt{s_n\varepsilon}
+ \frac{n}{2}\sqrt{2w} + \frac{n w}{2\sqrt{2}}$ in the linear term.
\hfill $\square$

\subsection{Discussion}
\label{sm:dirg:discussion}

Prior work on DIRG with the MABK inequality has either studied the maximum achievable randomness via convex optimization~\cite{Wooltorton_Brown_Colbeck_2025}, derived guessing-probability bounds for $n=3$~\cite{Woodhead_Bourdoncle_Acin_2018}, or obtained entropy bounds for $n=3$ using reduction techniques~\cite{Grasselli_Murta_Kampermann_Bruth_2021}.
Different from these approaches, Proposition~\ref{prop:sm:dirg:complete} provides a fully analytical device-independent bound on the conditional entropy that accounts for both state and measurement imperfections.
The combined randomness deficit scales as $n^{2}\cdot 2^{-n/4}\sqrt{\varepsilon}$, which becomes negligible exponentially fast as the number of parties grows, and recovers the known ideal rates at maximal violation when $\varepsilon \to 0^+$.
The analysis applies uniformly to all $n \geq 3$: the ideal conditional entropy $H_{\rm ideal}(n)$ depends on $n$ through its parity and its class modulo~$4$ (the latter fixing which input is uniformly random, see Proposition~\ref{sm:dirg:wooltorton})~\cite{Wooltorton_Brown_Colbeck_2025}, while the deficit is independent of both.

Several features of the derived bound merit attention.

\textit{Square-root scaling.}
Both the state- and measurement-imperfection contributions scale as $O(\sqrt{\varepsilon})$, in contrast with the $O(\varepsilon)$ scaling of the extractability bound itself (Proposition~\ref{prop:main}).
For the state term, this arises from the Fuchs--van~de~Graaf conversion from fidelity to trace distance, which turns an $O(\varepsilon)$ fidelity deficit into an $O(\sqrt{\varepsilon})$ trace distance.
For the measurement term, it follows from the effective-commutator deficit $t_k \geq 1 - w$: the block-weighted bound $\sum_f p_{k,f}\sin^{2}\delta_k^f \le w/2$ (Eq.~\eqref{sm:dirg:comp_blockbound}) together with the Cauchy--Schwarz step of Eq.~\eqref{sm:dirg:comp_ndiamond_sum} turns the $O(\varepsilon)$ deficit $w \sim 4\varepsilon/\beta_{nQ}$ into the $O(\sqrt{\varepsilon})$ trace-distance contribution $\sqrt{2w}$.
The measurement-translation residual, which arises because the dephasing part of the extraction channel prevents the physical measurement from commuting exactly with the extraction isometry, contributes an additional $O(\varepsilon/\beta_{nQ})$ term per party to the trace distance (Lemma~\ref{lem:sm:dirg:translation}), which is subdominant relative to the $O(\sqrt{\varepsilon})$ contributions.
An $O(\varepsilon)$-scaling bound for the full protocol would require either a tighter conversion from fidelity to trace distance or a more direct entropic continuity argument that bypasses the trace-distance intermediate.

\textit{The factor $n^{2}$.}
The measurement-imperfection contribution $\Delta_{\rm meas}$ carries an $n^{2}$ prefactor, arising from the product of two independent sources: one factor of $n$ from the sum over parties in the composition decomposition of the $n$-party measurement channel (Eq.~\eqref{sm:dirg:comp_ndiamond_sum}), and a second factor of $n$ from $\log_2 d_R = n$ in the conditional-entropy continuity bound (Lemma~\ref{lem:sm:dirg:fannes_audenaert}).
The state-imperfection contribution $\Delta_{\rm state}$ carries only a single factor of $n$, entering solely through the conditional-entropy continuity bound.
A possible route to eliminating the prefactor $n$ from the measurement contribution is to exploit the structure of the GHZ state, whose relevant expectation values depend on the measurement angles through a product $\prod_k \sin(2\theta_k)$ rather than a sum.
Refining the $n$-dependence of the combined bound is left for future investigation.

Finally, we note that the analysis above sets aside finite-statistics effects; combining the present bounds with the finite-sampling confidence intervals of Section~\ref{sec:finite_sampling} would yield a fully device-independent certified randomness rate for DIRG with prescribed statistical confidence~\cite{Pironio_Acin_Massar_de_la_Giroday_Matsukevich_Maunz_Olmschenk_Hayes_Luo_Manning_et_al_2010}.

\end{document}